\documentclass[11pt]{article}
\usepackage[english]{babel}
\usepackage[utf8]{inputenc}
\usepackage[T1]{fontenc}
\usepackage[letterpaper, margin=1in]{geometry}
\usepackage{graphicx}
\usepackage{framed}
\usepackage[framemethod=tikz]{mdframed}
\usepackage[disable]{todonotes}
\usepackage{color}

\usepackage{tikz}
\usetikzlibrary{graphs}

\usepackage[showdeletions]{color-edits}
\addauthor[mihir]{ms}{purple}
\addauthor[sri]{sr}{green}

\definecolor{darkgreen}{rgb}{0,0.5,0}
\definecolor{darkgray}{rgb}{0.2,0.2,0.2}
\usepackage[backref=page]{hyperref}
\hypersetup{
    unicode=false,          % non-Latin characters in Acrobat’s bookmarks
    colorlinks=true,        % false: boxed links; true: colored links
    linkcolor=blue,          % color of internal links (change box color with linkbordercolor)
    citecolor=purple,        % color of links to bibliography
    filecolor=magenta,      % color of file links
    urlcolor=cyan           % color of external links
}

\usepackage{amsthm}
\usepackage{amsmath}
\usepackage{amssymb}
\usepackage{amsfonts}
\usepackage{mathrsfs}
\usepackage{mathtools}
\usepackage{verbatim}
\usepackage{footnote}
\usepackage{algorithm}
\usepackage[noend]{algpseudocode}
\newcommand*{\longDefiningEquals}{\stackrel{\text{def}}{=\joinrel=}}%%% long equality symbol that is used to define stuff
  
\usepackage{lineno}
\usepackage{caption}
\usepackage{framed}
\usepackage{enumerate}

\usepackage{tikzsymbols}
\usepackage{thmtools,thm-restate}
\usepackage{nicefrac}

\usepackage{subcaption}
\usepackage{pdfpages}
\usepackage{bbm}

\usepackage{mdframed}

\usepackage[capitalize, nameinlink]{cleveref}
\usepackage[section]{placeins}
\usepackage{soul}
\crefname{theorem}{Theorem}{Theorems}
\crefname{ineq}{inequality}{inequalities}
\Crefname{lemma}{Lemma}{Lemmas}
\Crefname{invariant}{Invariant}{Invariants}
\Crefname{claim}{Claim}{Claims}
\Crefname{observation}{Observation}{Observations}
\Crefname{algorithm}{Algorithm}{Algorithms}
\Crefname{figure}{Figure}{Figures}
\newtheorem{challenge}{Challenge}
\mdfdefinestyle{challenge}{
    linewidth=1pt,
    leftmargin=0,
    rightmargin=0,
    backgroundcolor=gray!5,
    innertopmargin=0pt,
    innerbottommargin=5pt,
}
\surroundwithmdframed[style=challenge]{challenge}

\newtheorem{theorem}{Theorem}[section]
\newtheorem{lemma}[theorem]{Lemma}

\newtheorem{corollary}[theorem]{Corollary}

\newtheorem{definition}[theorem]{Definition}

\newtheorem{observation}[theorem]{Observation}
\newtheorem{claim}[theorem]{Claim}

\newtheorem*{remark*}{Remark}

\newcommand{\rSC}{r_{\textsc{SC}}}

\newcommand{\eqdef}{\stackrel{\text{\tiny\rm def}}{=}}
\def\wt{\mathsf{wt}}
\def\opt{\mathsf{OPT}}
\def\poly{\text{poly}}
\def\cov{\mathsf{Cov}}

\def\LOCAL{\mathsf{LOCAL}}
\def\stari{i^{\star}}
\def\starj{j^{\star}}
\def\wterr{\widehat{\wt}}
\def\derr{\hat{d}}
\def\hS{\hat{S}}
\def\setCover{\mathsf{SetCover}}
\def\LCA{\mathsf{LCA}}
\def\MIS{\mathsf{MIS}}

\def\eps{\varepsilon}

\newcommand{\rb}[1]{\left( #1 \right)}
\newcommand{\xpect}[1]{\mathbb{E} \left[ #1 \right]}

\newcommand{\E}[1]{\xpect{#1}}
\newcommand{\prob}[1]{\Pr \left[ #1 \right]}
\def\OPT{\mathsf{OPT}}
\newcommand{\cC}{\mathcal{C}}
\newcommand{\cE}{\mathcal{E}}
\newcommand{\cS}{\mathcal{S}}
\newcommand{\cA}{\mathcal{A}}
\usepackage{soul}
\Crefname{ineq}{Inequality}{Ineqs}
\newcommand{\IsSetDense}[1]{\textsc{IsSetDense}(#1)}
\newcommand{\IsEleCov}[1]{\textsc{IsEleCov}(#1)}
\newcommand{\GetWeight}[1]{\textsc{GetWeight}(#1)}
\newcommand{\lcaWeight}[1]{\textsc{LCA-Weight}(#1)}
\newcommand{\true}{\textsc{true}\xspace}
\newcommand{\false}{\textsc{false}\xspace}
\newcommand{\wtfew}{\wt_{\textsc{few}}}
\newcommand{\wtmany}{\wt_{\textsc{many}}}

\newcommand\previ[1]{{#1}^{\text{prev}}}

\def\starj{j^{\star}}

\newcommand\GetDeg[1]{\Call{degreeEstimate}{#1}}
\newcommand\GetWt[1]{\Call{weightEstimate}{#1}}
\title{Improved Local Computation Algorithms for Greedy Set Cover \\
    via Retroactive Updates
}
\author{
Slobodan Mitrović \thanks{e-mail: \texttt{smitrovic@ucdavis.edu}} \thanks{S. Mitrovi\'c and S. Ramachandran were supported by NSF CAREER Award No.~2340048} \\ UC Davis 
\and Srikkanth Ramachandran\thanks{e-mail: \texttt{sramach@ucdavis.edu}} \footnotemark[2] \\ UC Davis
\and Ronitt Rubinfeld \thanks{email: \texttt{ronitt@csail.mit.edu}} \thanks{R. Rubinfeld was supported by  NSF awards DMS-2022448 and CCF-2310818} \\ CSAIL, MIT and TAU
\and Mihir Singhal \thanks{email: \texttt{mihirs@berkeley.edu}. Part of this work was conducted while the authors were visiting the Simons Institute for the Theory of Computing} \\ UC Berkeley
 }
\date{ }

\def \poly {\text{poly}}

\begin{document}

\maketitle

\begin{abstract}
    In this work, we focus on designing an efficient Local Computation Algorithm (LCA) for the set cover problem, which is a core optimization task. 
    The state-of-the-art LCA for computing $O(\log \Delta)$-approximate set cover, developed by Grunau, Mitrović, Rubinfeld, and Vakilian [SODA '20], achieves query complexity of $\Delta^{O(\log \Delta)} \cdot f^{O(\log \Delta \cdot (\log \log \Delta + \log \log f))}$, where $\Delta$ is the maximum set size, and $f$ is the maximum frequency of any element in sets. 
    We present a new LCA that solves this problem using $f^{O(\log \Delta)}$ queries. 
    Specifically, for instances where $f = \text{poly} \log \Delta$, our algorithm improves the query complexity from $\Delta^{O(\log \Delta)}$ to $\Delta^{O(\log \log \Delta)}$.

    Our central technical contribution in designing LCAs is to aggressively sparsify the input instance but to allow for \emph{retroactive updates}.
    Namely, our main LCA sometimes ``corrects'' decisions it made in the previous recursive LCA calls.
    It enables us to achieve stronger concentration guarantees, which in turn allows for more efficient and ``sparser'' LCA execution.
    We believe that this technique will be of independent interest.
\end{abstract}

\setcounter{page}{0}
\thispagestyle{empty}

\newpage

\tableofcontents

\setcounter{page}{0}
\thispagestyle{empty}

\newpage

\section{Introduction}
With the prevalence of large volumes of data, our perception of what makes an algorithm efficient has evolved significantly. 
In scenarios where accessing the entire input is impractical or even infeasible, the focus has shifted toward designing algorithms that use sublinear resources -- particularly time and space -- to solve a given problem.

Alon, Rubinfeld, Tamir, Vardi, and Xie~\cite{RTVX11,ARVX12} introduced the Local Computation Algorithm (LCA) model to formalize the study of algorithms that, rather than processing the entire input, access only a ``small'' fraction of the input to produce a ``small'' fraction of the solution.
The complexity of an LCA is then measured by the number of \emph{queries} it makes to the input.
For example, instead of computing the entire maximal independent set (MIS) of a graph or finding a complete set cover, 
an LCA provides localized solutions by answering questions such as ``Is a given vertex in the MIS?'' or ``Is a specific set part of the set cover?''
Ideally, LCAs are significantly faster than algorithms that compute the entire MIS or set cover, particularly for large inputs.

Since their inception, the development of LCAs has been an active research area.
In particular, extensive research has focused on LCAs for a broad range of graph problems, including local reconstruction of graph properties \cite{KalePS13}, 
maximal independent sets \cite{RTVX11,ARVX12,BarenboimEPS16,LeviRY17,EvenMR14,Ghaffari16,
GhaffariU19,Ghaffari-LCA-FOCS22}, 
coloring \cite{chang2019complexity,dorobisz2023local}, algorithmic Lovász local lemma~\cite{achlioptas2020simple,brandt2021randomized},
approximate maximum matchings \cite{NO08,YYI09,MRVX12,MansourV13,LeviRY17, GhaffariU19,KMNT20,behnezhad2023local}, $k$-spanners~\cite{biswas2024average,LenzenL18,ParterRVY19,ArvivCLP23}, graph orientation~\cite{mitrovic2024locally}, 
correlation clustering~\cite{behnezhad2022almost,dalirrooyfard2024pruned},
local reconstruction of graph properties \cite{CampagnaGR13}, and local generation of large random objects \cite{MRVX12,EvenLMR17}.
LCAs are of interest in broader algorithmic contexts:
Several works demonstrate novel applications of LCAs in developing faster dynamic matching algorithms~\cite{bhattacharya2023dynamic,behnezhad2023dynamic,bhattacharya2024dynamic} and in the study of proper learning of monotone functions~\cite{lange2022properly,lange2023agnostic}.
Important connections have also been established between LCAs and Differential Privacy~\cite{jha2013testing,lange2023local}.

\paragraph{Local simulation of distributed algorithms.}
Developing LCAs by leveraging distributed algorithms as their starting point has been a particularly fruitful line of research.
Parnas and Ron~\cite{parnas2007approximating} observed that, on graphs with maximum degree $d$, the output of a vertex in an $r$-round $\LOCAL$ algorithm depends on at most $d^r$ vertices within its $r$-hop neighborhood. 
This idea directly translates a $\LOCAL$ to an LCA with a query complexity of $O(d^r)$.
Over the last few years, we have seen a proliferation of ideas on \emph{sparsifying} this $r$-hop neighborhood to achieve more efficient LCAs.
Notably, a series of works on designing LCAs for MIS~\cite{NO08,YYI09,RTVX11,ARVX12,reingold2016new,levi2015local,Ghaffari16,levi2017centralized,even2018best,GhaffariU19} 
culminated in a recent work by Ghaffari~\cite{Ghaffari-LCA-FOCS22}, 
which demonstrates an MIS LCA with $2^{O(r)} \cdot \poly \log n$ query complexity. 
This improvement from $d^{O(r)}$ to $2^{O(r)}$ allows for the first polynomial time algorithm in $d$.
Approximate maximum matching and vertex cover have also been extensively studied in this regime, e.g., see~\cite{KMNT20,behnezhad2022time} and references therein.\footnote{The MIS result \cite{Ghaffari-LCA-FOCS22} also implies an algorithm for maximal matching and hence $2$-approximate vertex cover. Nevertheless, $2^{O(r)}$ LCA query complexity was known even before for higher approximate factors.}
The $\LOCAL$ algorithms simulated in those works use $r = O(\log d)$ rounds, and hence those results yield $\poly(d, \log n)$ LCA query complexity.

Our work focuses on the set cover problem, a foundational topic in algorithm design and computational complexity.
A set cover instance consists of a universe of elements $\cE$ and a family of sets $\cS \subseteq 2^\cE$. The set cover problem aims to find the smallest-cardinality collection $\cS' \subseteq \cS$ whose union equals $\cE$. We denote by $\Delta$ the maximum size of a set, and by $f$ we denote the maximum number of sets that contain any given element $e$.
This problem is well-known to be NP-hard.
In fact, unless $P = NP$, no polynomial-time algorithm can achieve better than a (multiplicative) $(\ln \Delta - O(\log \log \Delta))$-approximation for set cover~\cite{feige1998threshold,alon2006algorithmic,moshkovitz2012projection,dinur2014analytical,raz1997sub,trevisan2001non}.

Let $\rSC = O(\log \Delta \cdot \log f)$ be the round complexity of the $O(\log \Delta)$-approximate set cover algorithm of Kuhn, Moscibroda, and Wattenhofer~\cite{kuhn2006price}, which is the best-known $\LOCAL$ algorithm for computing $O(\log \Delta)$-approximate set cover.
From this result, the Parnas-Ron paradigm implies an LCA with $(\Delta f)^{O(\rSC)}$ complexity. 
That complexity was improved by Grunau, Mitrović, Rubinfeld, and Vakilian~\cite{GMRV20} who designed an LCA for approximating set cover that uses $\Delta^{O(\log \Delta)} \cdot 2^{O(\rSC \cdot \log \rSC)}$ queries; 
we refer the reader to the related work section (\cref{sec:prior-work,sec:related-work}) for more results on this topic.
Despite substantial advances in $\LCA$ design for core problems in algorithmic graph theory, the limitations of existing techniques and the types of methods required for achieving $2^{O(\rSC)}$ LCA query complexity for set cover remain unclear.
This motivates our central question:
\begin{center}
    \emph{Does there exist an LCA for $O(\log \Delta)$-approximate set cover that uses $2^{O(\rSC)}$ queries?}
\end{center}
We answer this question affirmatively.
Namely, our LCA has $2^{O(\log \Delta \cdot \log f)} = f^{O(\log \Delta)}$ query complexity deterministically, and the approximation is provided in expectation.
This significantly improves on the prior work \cite{GMRV20}.
Specifically, when $f = \poly \log \Delta$ the query complexity improves from $\Delta^{O(\log \Delta)}$~\cite{GMRV20} to $\Delta^{O(\log \log \Delta)}$.\footnote{Cases where $f \le \log \Delta$ can be handled more efficiently by $f$-approximate algorithms.}
Our approach can be outlined as follows: (1) we start by applying the Parnas-Ron paradigm to a distributed set cover algorithm, similar to that in \cite{kuhn2006price}; (2) we then \emph{sparsify} the volume of input required for the LCA computation.
The theme of LCA sparsification has seen a proliferation in recent years and has 
been used to significantly improve the efficiency of LCAs. 

Our work contributes to this active research area by introducing novel sparsification techniques for set cover LCAs. 
While some previous set-cover LCAs also apply sparsification, our methods diverge significantly.

\paragraph{Open problems.}
Our work makes an important step in LCA sparsification for set cover by showing how to simulate an $\rSC$-round $\LOCAL$ algorithm with $2^{O(\rSC)}$ queries.
Observe that an LCA in which each vertex explores only two of its neighbors, during the simulation of an $\rSC$-round algorithm, already visits $2^{\rSC}$ vertices. In that sense, our LCA quite sparsely explores the $\rSC$-hop neighborhood of a vertex.
It remains open whether an LCA for set cover with $\poly(\Delta, f)$ query complexity exists.
It would be even intriguing to obtain any asymptotic improvement in the exponent of this complexity, e.g., an LCA using $2^{o(\rSC)}$ queries.

\subsection{Overview of our result and approach}
Our work shows how to compute an \emph{integral} $O(\log \Delta)$-approximate set cover via LCA. More formally, we prove the following result.
\begin{theorem}
\label{thm:main}
    There exists an $\LCA$ for $\setCover$ that outputs a randomized cover which is $O(\log \Delta)$-approximate in expectation while using $2^{O(\log \Delta \cdot \log f)}$ probes, where $\Delta$ is the maximum set size and $f$ the maximum element frequency.
\end{theorem}

Plugging this $\LCA$ in the approach of \cite{YYI09} or \cite{NO08}, yields a new sublinear approximation for Set Cover. For algorithms that return $O(\log \Delta)$ approximation with probability at least $2/3$, our work implies that $\Delta^{O(\log f)}$ queries are sufficient, whereas the approach of \cite{YYI09} requires $\Delta^{O(\log \Delta + \log f)}$ queries. If one were to construct a sublinear approximation algorithm for Set Cover using the $\LCA$ of \cite{GMRV20}, then the query complexity would be $\Delta^{O(\log \Delta + \log f \cdot (\log \log \Delta + \log \log f))}$, which is not better than \cite{YYI09}. 

Furthermore, we can also use techniques from \cite{KMNT20} and reduce the space complexity of our $\LCA$ and the sublinear algorithm to $O((\Delta + f) \cdot \poly(\log n))$. If the common random string needed by our $\LCA$ and the sublinear algorithm is not counted in the space complexity, but instead provided by an oracle, then the space complexity can even be reduced to $O(\poly(\log \Delta, \log f))$. 

\def\topt{\tilde{\opt}}
\begin{corollary}
\label{corollary}
    There exists a sublinear approximation algorithm for $\setCover$ that makes at most $2^{O(\log \Delta \log f)}$ queries, uses at most $O((\Delta f) \poly(\log n))$ space, and returns an estimate $\topt$ such that w.p. at least $2/3$, $|\topt - \opt| \leq O(\log \Delta) \cdot \opt$. 
\end{corollary}

\subsubsection{Organization of the paper.}

In \cref{sec:warmup} we describe a simple pruning approach that allows us to obtain an $\LCA$ for fractional set cover with query complexity $\exp{(\tilde{O}(\log \Delta \cdot \log f))}$. These ideas are further refined in \cref{sec:prob-est,sec:main-alg} to remove log-factors in the exponent. This is perhaps the most novel and interesting part of our work, since the existing sparsification pipelines rely on obtaining high probability estimates of the state of the $\LOCAL$ algorithm, which would not be enough to attain the desired query complexity.
In \cref{sec:frac-to-integral} we describe how to modify the algorithm to obtain an integral set cover. In short, whenever the algorithm adds a weight $w$ to some set $S$ in the fractional algorithm, the integral algorithm instead simply adds it with probability $w$ to the set cover. We argue that this results in an integral set cover that is $O(\log \Delta)$-approximate in expectation.
In \cref{sec:proof-of-corollary} we show how to implement our $\LCA$ with low space complexity and describe the sublinear algorithms to prove \cref{corollary}.

We begin by starting with an algorithm that computes a fractional set cover, since it is easier to reason about the approximation ratio of fractional set cover algorithms as compared to integral ones. In \cref{sec:frac-to-integral} we show how to modify our algorithms to obtain an integral set cover. 

\begin{algorithm}
    \begin{algorithmic}[1]
        \Statex \textbf{Input:} A set cover instance $\langle \mathcal{S}, \mathcal{E} \rangle$,
        \State $\Delta \longDefiningEquals \max_{S \in \mathcal{S}} |S|$
        \State $f \longDefiningEquals \max_{e \in \mathcal{E}} |\{S \in \mathcal{S} \mid S \ni e\}|$.
        \State $\cov \gets \emptyset$ \Comment{Elements covered so far}
        \State For every $S \in \mathcal{S}$, set $\wt(S) \gets 0$ \Comment{Initialize set weights $\wt$ to $0$}
        \For{$i = 1, 2, \dots \log \Delta$}
            \For{$j = 1, 2, \dots \log f$}
                \For{every set $S \in \mathcal{S}$ (in parallel)}
                    \If{$d(S) \longDefiningEquals |S \setminus \cov| \geq \Delta / 2^i$} \Comment{$d(S)$ is the effective degree of a set}
                        \State $\wt(S) \gets \min(1, \wt(S) + 2^j / f)$ \label{line:gmrv20:add-set} \Comment{Add $2^j / f$ to $\wt(S)$}
                    \EndIf 
                \EndFor 
                \State $\cov \gets \{e \in \cE \mid \wt(e) \geq 1\}$
            \EndFor 
        \EndFor 
        \State \Return $\wt$
    \end{algorithmic}
    \caption{The set cover algorithm from \cite{berger1994efficient,GMRV20} adapted to computing fractional set cover.}
    \label{gmrv20:alg}
\end{algorithm}

\subsubsection{Starting point: Fractional Set Cover algorithm}
Our starting point is \cref{gmrv20:alg}, a fractional $\LOCAL$ algorithm for the set cover problem that can be attributed to several prior works~\cite{berger1994efficient,kuhn2006price,GMRV20}.
\paragraph{Details of \cref{gmrv20:alg}.} The algorithm consists of $\log \Delta$ \textit{phases}.
In each phase, a subset of elements is covered by adding weights to the sets such that the following holds:
\begin{itemize}
    \item After removing the elements covered in a phase, the maximum size of the sets halves.
    \item During each phase, the total weight added to the sets is at most $O(\opt)$, where $\opt$ denotes the minimum (integral) set cover size.
\end{itemize}
For convenience, we use the term ``effective degree'' of a set $S$, denoted by $d(S)$, that counts the number of elements in $S$ that are uncovered in the current phase.
Observe that the ``effective degree'' reduces throughout the course of the algorithm, reaching zero for every set by the time the computation is over.

We now describe the workings of a phase in which the maximum degree is at most $\Delta_i =\Delta / 2^{i-1}$. We say that a set $S$ is ``dense'' if $d(S) \geq \Delta_i / 2$. Note that a dense set covers as many elements as any other set (upto a constant factor).
Each phase contains $\log f$ \textit{iterations}. In each iteration, a subset of elements is covered, ensuring that (i) weight is added only to \textit{dense} sets and (ii) $\wt(e) \leq O(1)$ for each element during the first iteration in which it becomes covered.
These conditions ensure that the chosen sets cover nearly as many elements as any other family of sets, with only a few sets covering each element.
Thus, this covering solution costs $O(\opt)$. A formal proof is described later (see \cref{lemma:frac-base-approx}). 

\subsubsection{Introducing slack in thresholds and pruning (\cref{sec:warmup})}
\paragraph{Observation 1.} Our first observation is that \cref{gmrv20:alg} allows some slack in the thresholds at the expense of only a constant factor increase in the approximation ratio guarantee. 
In particular, \cref{gmrv20:alg} uses two thresholds:  
(i) in the test of \textit{dense} sets,
i.e., whether ~$d(S) \geq \Delta / 2^i$, and (ii) in the test of whether an element is covered, i.e., whether ~$\wt(e) \geq 1$. 
Both conditions can be relaxed. The relaxation differs in our two algorithms. We describe the first and simpler one below. 
Instead of testing whether $d(S) \geq \Delta / 2^i$, an algorithm can test whether a set of $t$ sampled elements of $S$ contains at least $\frac{1}{2} \cdot t / 2^i$ uncovered elements; the factor of $1/2$ and the variable $t$ are chosen to increase the probability of ``a correct test''. 
A similar relaxation can be performed for the other test. \cref{alg:frac:setcover:base} implements these relaxations.

\paragraph{Why is Observation 1 not sufficient?} 
To test whether a set $S$ has effective degree $d(S) \geq \Delta / 2^i$ at the beginning of some iteration $(i, j)$, we must test whether at least $2^i$ elements containing $S$ are covered at the end of the previous iteration, $(\previ{i}, \previ{j})$. Note that $(\previ{i}, \previ{j}) = (i-1, \log f)$ if $j = 1$ and $(i, j-1)$ otherwise.
Following the natural LCA design,
the query complexity $T(i, j)$ to simulate the $\LOCAL$ algorithm up to iteration $(i, j)$ satisfies the recurrence $T(i, j) = 2^i \cdot T(\previ{i}, \previ{j})$. 
% \rrtodo{have we defined T(i,j) yet? if so, i forgot waht is i,j}
It is important to note that this discussion does not yet consider the concentration required for accurate estimates; thus, even the above recurrence represents a hypothetical scenario.
Even with this approach, the best solution achievable is $T(i, j) = 2^{O(i^2\log f + j)}$, i.e., no better than Parnas-Ron reduction of $\Delta^{O(\log \Delta \log f)}$. 
Our second observation is essential for eliminating this term and thereby reducing the query complexity.

\paragraph{Observation 2.} 
The second observation is that if each phase of \cref{gmrv20:alg} can be simulated with sufficiently high probability, then many recursive calls in the natural LCA based on \cref{gmrv20:alg} can be pruned. 
The natural $\LCA$ design proceeds as follows: to determine the state of a set or element $v$ at the end of some iteration $t$ of the $\LOCAL$ algorithm, one recursively queries the states of \emph{all} $v$'s neighbors at iteration $t-1$ and then simulates iteration $t$ of \cref{gmrv20:alg}.
Our first observation enables querying only a subset of neighbors, as estimates suffice; however, this alone does not yield an efficient $\LCA$. 
The second observation is that for every $S$, its size at the end of phase $i$ drops to $\Delta / 2^i$. 

Let $\hS$ be a sample of $2^i$ elements $e \in S$ each sampled independently and uniformly at random with replacement. Of the $2^i$ elements in $\hS$, we expect at most half of them to be uncovered after each phase.
On average, only $O(1)$ many elements in $\hS$ are expected to survive until phase $i-1$ and, more generally, $2^{O(i - i')}$ many elements in $\hS$ are expected to survive until phase $i' < i$. 
The recurrence resulting from this pruning procedure is discussed in the next subsection.

\subsubsection{A useful query recurrence}
% The pruning for the main algorithm \cref{sec:main-alg} is more intricate than described above.
Let $T(i)$ denote the query complexity required by the $\LCA$ to simulate $i$ iterations of the $\LOCAL$ algorithm \cref{gmrv20:alg}. Based on our pruning discussion we outline the ideal recurrence we would like the algorithm to satisfy. The same recurrence is also used in \cite{Ghaffari-LCA-FOCS22}:
\begin{equation}\label{eq:recurrence}
    T(i) \leq \sum_{j < i} T(j) \cdot 2^{O(i - j)}.
\end{equation}
This recurrence yields a solution of $T(i) \leq T(0) \cdot 2^{O(i)}$ and can be verified by induction.
If our $\LCA$ satisfies this recurrence, then to simulate $\log \Delta \cdot \log f$ rounds, our query complexity would be $2^{O(\log \Delta \log f)}$, as in the claim of \cref{thm:main}.
We now outline one of the main challenges in achieving such a recursive relation. 

The previous approaches relied on near-perfectly simulating the distributed $\LOCAL$ algorithm \cref{gmrv20:alg}. 
More precisely, the simulation of every iteration of \cref{gmrv20:alg} for every element fails with probability at most $1 / \poly(\Delta)$.
This probability is small enough that those elements for which the simulation fails can be covered naively. 
Observe that there are at most $|\cE| / \poly(\Delta)$ simulation-failed elements, while an optimum solution has a size at least $|\cE| / \Delta$. 
% Hence, such low failure probability can be handled trivially.
If we were to follow this approach, at every iteration $(i, j)$ we would need to verify whether half of the elements in a given set were covered before iteration $(i, j)$ with failure probability $1 / \poly(\Delta)$. 
This would require at least $\log \Delta$ recursive calls at every iteration. Unfortunately, this would not allow the recurrence \cref{eq:recurrence}, in which we are allowed at most \textit{constant} many recursive calls to the previous iteration.
One way to address this is by allowing higher failure probability in the tests of the conditions.
%, which cannot provide a failure probability as small as $1 / \poly(\Delta)$. 
\begin{challenge}
\label{challenge:failure-prob}
    How can we obtain an $O(\log \Delta)$ approximation even if mistakes in the tests performed  in every iteration of \cref{gmrv20:alg} are made ``relatively often''?
\end{challenge}

To overcome this difficulty, we move away from this argument by saying that even if some iterations are not simulated perfectly, \cref{gmrv20:alg} still provides a $O(\log \Delta)$-approximation. Finding the characteristics of these mistakes, i.e., how frequently these mistakes can be made, addresses our first challenge.

\subsubsection{Weaker invariants -- addressing \cref{challenge:failure-prob}}
The main question we now aim to address is: ``How much slack do the conditions of \cref{gmrv20:alg} allow for?'' 
Two crucial components in the analysis of \cref{gmrv20:alg} are the following invariants:
\begin{enumerate}[(i)]
    \item At the end of phase $i$, the effective degree of every set is at most $\Delta / 2^i$.
    \item Weight is added in phase $i$ only to those sets whose effective degree is at least $\Delta / 2^i$.
\end{enumerate}
We show that the analysis can be modified so that the following \textbf{weaker} invariants hold, which can be seen as ``smoother'' (or ``fractional'') versions of the original:
\begin{enumerate}[(i)]
    \item At the end of phase $i$, the probability that a set has effective degree $\Delta / 2^{i'}$ is at most $2^{-\Omega(i-i')}$ for every $i' < i$.
    \item The probability that weight is added to a set during iteration $(i, j)$ whose effective degree was $d$ during iteration $(i, j')$ for every $j' < j$ is at most $(d/\Delta_i) \cdot 2^{-\Omega(j-j')}$.
\end{enumerate}

Given the allowed slack identified in \cref{challenge:failure-prob}, one can design a suitable natural $\LCA$, but such an $\LCA$ is hard to analyze. Since we allow the possibility that sets still have a large effective degree, we can no longer prune the $\LCA$s effectively. More specifically, the following guarantee is lost: For every set $S$, at most $O(\Delta/2^{i})$ of its elements are uncovered by phase $i$. The guarantee holds with probability $1$ in \cref{alg:frac:setcover:base}, but our modification to address \cref{challenge:failure-prob} reduces the probability to at most a constant. Hence, for a given set $S$, there is a possibility that almost all of its elements survive up to phase $i$. This would require us to simulate the state of too many neighbors of $S$ in the $\LCA$, and we cannot hope to obtain \cref{eq:recurrence}. Note that we wish to aim for \textit{worst case} bounds for the query complexity of the $\LCA$, so the natural $\LCA$ does not seem suitable. 

\begin{challenge} \label{challenge:lca-design}
    Design a $\setCover$ $\LCA$ that satisfies the weaker invariants and the recurrence \cref{eq:recurrence}.
\end{challenge}

\subsubsection{Retroactive updates -- addressing \cref{challenge:lca-design}}

To overcome this second challenge, we introduce the idea of \textit{retroactive} updates. Let us first inspect the recurrence we aim for and deduce what we can afford to do given this recurrence.
To obtain information about the state of a vertex, i.e., an element or a set, at the end of some iteration $t$, we can afford to obtain information about the state of $2^{O(t - t')}$ neighbors at the end of round $t'$. 
More crucially, we can also afford to \textit{re-execute} simulations of these neighbors up to $2^{O(t - t')}$ times \footnote{$2^{O(t - t')} \cdot 2^{O(t - t')}$ is still $2^{O(t - t')}$, with a larger constant}. Imagine that an oracle provides only probabilistic guarantees about the state. 
Then, re-executing these simulations can enable us to drastically reduce the probability of oracle failure. 
We shall see that the probability of failure can be made \textit{exponential} in $t - t'$, which turns out to be sufficient for our needs.

We now discuss how we implement the above idea. 
We transform the $\LOCAL$ algorithm $\cA$ designed for \cref{challenge:failure-prob} into an algorithm $\cA'$ suitable for $\LCA$ analysis as follows.
Recall that $\cA$ performs $\log \Delta$ phases sequentially.
The transformation is recursive in nature. Let $\cA(i)$ denote the truncation of $\cA$ up to $i$ phases. $\cA'$ also performs the same number of phases. We describe below $\cA'(i)$, defined analogously to $\cA(i)$. The final algorithm we want is $\cA'(\log \Delta)$.
\begin{itemize}
    \item Execute $\cA'(i-1)$. If $i = 1$, we do nothing.
    \item Re-execute $\cA'(i')$,  $2^{O(i - i')}$ times for every $i' < i$, in order from $i' = 1, 2, \dots i-1$. 
    Identify any potential mistakes made and \textit{selectively} account for them \textit{retroactively}.
    
    The purpose of these re-executions is to identify potential mistakes made in the previous phases and \textit{retroactively} correct them. That is, if in phase $i$ we find that in some iteration $(i', j')$ the set $S$ was dense -- but $S$ was incorrectly estimated to be light -- then we shall increase the weight of $S$ by the amount it should have been increased by during iteration $(i', j')$. 
    This will reduce the probability of failure of the invariants we need. After the retroactive updates for the previous phases, the failure probability is low enough that we can afford to execute phase $i$.

    When we want to correct the errors made in some phase $i'$, we must keep the randomness in phases $i{''} < i'$ \textit{fixed}. This is because we want to estimate more accurately the mistakes that were made \textit{only} \ in phase $i'$. 
    Potential mistakes in previous phases are to be carried over and we will account for them. Hence, during the re-executions of $\mathcal{A}'(i')$, we shall only \textit{freshly} sample randomness associated with phase $i'$. We explain this in more detail in the discussion of the main algorithm.
    \item Execute $\cA(i)$.
\end{itemize}

Fortunately, the nature of analysis in our solution to \cref{challenge:failure-prob} allows us to account for the retroactive updates made in the approximation ratio as well. The retroactive updates serve two purposes: (i) they allow us to prune the $\LCA$ designed effectively to realize \cref{eq:recurrence}, and (ii) they allow us to easily reason about the approximation ratio by using Chernoff and union bounds to analyze the failure probabilities. 

\subsubsection{Modifying fractional to integral (\cref{sec:frac-to-integral})}
We modify the analysis of \cite{GMRV20} to show that a simple modification to the fractional algorithm outputs an integral solution. Importantly, in expectation, this modification yields an algorithm that incurs only an additional constant factor loss in the approximation. Note that we do not round the computed fractional cover; rather, a natural modification of the fractional cover solution also yields an integral cover.

\subsection{Comparison to the Most Relevant Prior Work}
\label{sec:prior-work}
The most relevant to our approach are \cite{KMNT20,GMRV20,Ghaffari-LCA-FOCS22}. 
We now discuss the relationship between these results and ours. 

\emph{Remark:} This subsection is tuned toward readers who are at least vaguely familiar with the techniques and ideas used in the prior work mentioned.

\paragraph{Comparison of techniques.} Our main recurrence \cref{eq:recurrence} is inspired from \cite{Ghaffari-LCA-FOCS22}. The main technical portion of \cite{Ghaffari-LCA-FOCS22} is a method to transform the $O(\log \Delta)$ $\LOCAL$ algorithm of \cite{Ghaffari16} so that it is amenable to efficient $\LCA$ simulations whose oracles satisfy \cref{eq:recurrence}. This transformation is similar to our solution to \cref{challenge:lca-design}. The difference is that in \cite{Ghaffari-LCA-FOCS22}, the underlying $\LOCAL$ algorithm operates in $O(\log \Delta)$ rounds with the guarantee that a \textit{constant} fraction of those rounds are ``good''. The guarantees of the ``good'' rounds are identical to each other, so one can afford to \textit{delay} the occurrence of a good round. In our case, each phase has a slightly \textit{different} measure of progress. In phase $i$ we require that the effective sizes of sets ``roughly'' drop to $\Delta / 2^{i}$. So for instance, we cannot afford to run the phases in parallel as they must occur sequentially one after the other. This is why we require the \textit{retroactive} updates.

Our design of the $\LCA$ and analysis of its round complexity borrows some ideas from \cite{KMNT20}. The key idea here is to tie the query complexity to the algorithm's output, making the analysis of the early-stopping rule easier. The argument of the approximation ratio, however, is very different. \cite{KMNT20} study approximations of maximum matching and their approximation ratio arguments are done by relating to vertex cover. 
In addition, their approach is tailored to the maximum matching problem and requires novel concentration bounds. 
Our concentration bounds are simple compositions of Chernoff and the union bounds. Consequently, we hope our analysis is easier to comprehend and apply to other settings.

Our results are to be directly compared to \cite{GMRV20}. We solve the same problem with the same approximation guarantees but with better query complexity. The design of our and the $\LCA$ in \cite{GMRV20} are considerably different. The main optimizations of \cite{GMRV20} are in the simulations of an individual phase, with no direct pruning done across different phases. This is the reason behind the $\Delta^{O(\log \Delta)}$ factor in the query complexity of \cite{GMRV20}. 
A potential lesson here is that adaptively making the recursive calls seems more beneficial than the ``repeated sparsification'' technique in the $\LCA$ setting.

\subsection{Related Work}
\label{sec:related-work}
As discussed in the introduction, LCAs have been studied in many contexts and for many problems. This section lists related work in addition to that mentioned in the introduction.

Results that predate but can be viewed within the LCA model include locally (list)-decodable codes~\cite{BF90,L89,GLR+91,FF93,KT00,Yek10,
GRS00,AS03,STV01,GKZ08,IW97,KS09,BET10}), 
local decompression~ \cite{MuthuSZ,SG,FV07,GN}, local reconstruction and filters for monotone and
Lipshitz functions \cite{ACC+08,SS10,BhattacharyyaGJJRW12,jha2013testing,
AwasthiJMR12}, satisfying assignments for $k$-CNF \cite{RTVX11,ARVX12}, local computation mechanism design \cite{HassidimMV16}, and
local decompression \cite{DLRR13}, 

\paragraph{Pruning.} 
Different variants of pruning are implicit in prior works on LCAs such as \cite{NO08,YYI09,Behnezhad21,dalirrooyfard2024pruned}.
The underlying idea is that before trying to obtain the state of some node at the end of some iteration $t$ in the algorithm, make sure to obtain its state in iterations $1, 2, \dots t-1$, in this order. We may find information about the state in some intermediate iteration $t' \ll t$ that allows us to ignore the state of $v$ after $t'$. Thus, a large chunk of recursive calls can be pruned.
\cite{NO08} suggested this optimization for designing LCAs from the Randomized Greedy MIS algorithm. \cite{YYI09} analyzed the expected query cost of this optimization, and \cite{Behnezhad21} provided analysis for line graphs (i.e., the maximal matching problem). 
The analysis is quite involved and different in the two papers.
\cite{dalirrooyfard2024pruned} provide an optimal LCA for the correlation clustering problem. The pruning in \cite{dalirrooyfard2024pruned} is much simpler: once the recursion tree contains too many calls, stop and make do with the information that you have.

\paragraph{Vertex cover, i.e., the case of $f = 2$.}
The set cover problem with $f = 2$ is equivalent to vertex cover for graphs, which has been extensively studied.
The $\poly(\Delta, \log n)$ round MIS algorithm of Ghaffari~\cite{Ghaffari-LCA-FOCS22} already implies a $\poly(\Delta, \log n)$ algorithm for $2$-approximate vertex cover through the classical reduction to MIS; that is, maximal matching is a $2$-approximation of vertex cover and is also an MIS of the line graph. The line graph of $G$ is the graph $G'$ whose vertices correspond to edges of $G$ and an edges exists between vertices of $G'$ if the corresponding edges in $G$ share a vertex.

Even though it is not explicitly mentioned, it can be shown that the LCA algorithm for computing approximate maximum matching by \cite{KMNT20} outputs a vertex cover as well. 
Compared to \cite{Ghaffari-LCA-FOCS22}, the approach in \cite{KMNT20} incurs a higher approximation ratio, although only linear dependence on $\Delta$ in the query complexity.

\paragraph{Local simulation of random order greedy.}
Let $\Delta$ be the largest set size and $f$ the highest element frequency.
Nguyen and Onak~\cite{NO08} designed a sublinear-time algorithm to approximate set cover by simulating a classical greedy algorithm on randomly ordered sets.
Their algorithm uses $\rb{{2^{(\Delta f)}}^4}^{O(2^\Delta)}$ queries.
This complexity was substantially improved in a seminal work by Yoshida, Yamamoto, and Ito~\cite{YYI09}, who demonstrated how to solve the same problem by $(\Delta f)^{O(\Delta)}$ queries via average case analysis of the randomized greedy $\MIS$ algorithm. We note that their output provides an approximation guarantee of $H_{\Delta} \cdot \opt + \eps n$ where $n$ is the number of elements and $H_{\Delta}$ denotes the harmonic number. For e.g., setting $\eps = 1 / \Delta$, one gets an $H_{\Delta} + 1$ multiplicative approximation. One can also modify their algorithm to obtain $O(\log \Delta)$ approximation within $(\Delta f)^{O(\log \Delta)}$ queries (see \cref{sec:proof-of-corollary}).

\paragraph{Set cover in $O(\log n)$ $\LOCAL$ rounds.}
It is known how to compute approximate set cover in $O(\log n)$ $\LOCAL$ rounds~\cite{kuhn2016local,chang2023complexity}. 
Nevertheless, it is unclear how to leverage those algorithms in designing LCAs. 
In particular, $2^{O(\log n)} = \poly(n)$; hence, the known techniques would not directly lead to sublinear algorithms.
\section{Preliminaries}

\subsection{Set Cover and Local Computation Algorithms}

\begin{definition}[$\setCover$] \label{def:setcover}
    Given a universe of elements $\cE$ and a family of subsets of $\cE$, $\cS \subseteq 2^\cE$, we call $\cC \subseteq \cS$ a cover of $\cS$ if for every $e \in \cE$ there exists at least one $S \in \cS$ containing $e$. $\cS$ is said to be optimal if it has the minimum size. We denote the set cover instance as $\langle \cS, \cE \rangle$ and the size of the optimal set cover by $\opt(\langle \cS, \cE \rangle)$ or simply $\opt$ when the instance is clear in context.
\end{definition}

\begin{definition}[Local Computation Algorithm ($\LCA$)] \label{def:lca}
    Given a problem $\Pi$ with input $I$ and desired output $O$, both presumed to be of large size and ordered, a local computation algorithm for $\Pi$ is a procedure $P$ that takes as input, 
    \begin{enumerate}[(i)]
    \item an oracle that provides access to any element of $I$ upon query, 
    \item an index $i$ with $1 \leq i \leq |O|$,
    \item oracle access to a sufficiently large common random string,
    \end{enumerate}
    and outputs $O[i]$, i.e., the element of $O$ indexed by $i$. The query complexity of $P$ is the size of the portion of $I$ that it obtains from the input oracles to obtain $O[i]$.
\end{definition}

\paragraph{$\LCA$ conventions.} In this paper, our goal is to provide an $\LCA$ for the Set Cover problem. We shall assume that all set cover instances are represented as a bipartite graph whose one partition denotes the sets and the other denotes the elements. An edge exists in the graph between a set $S$ and element $e$ if and only if $e \in S$. We shall assume that this graph is stored in \textbf{adjacency list} representation and the LCA oracles we design can access these lists through input oracles. In particular, the $\LCA$ can provide a vertex $v$ in the graph and an integer $i$ to the input oracle which then returns the vertex $w$ that is stored in the $i^{th}$ position of the adjacency list of $v$. The $\LCA$ must support the following queries on a cover $\mathcal{C}$:
\begin{itemize}
    \item Given a set $S$, does $S$ belong to $\mathcal{C}$?
    \item Given an element $e$, output a set $S$ such that $S \ni e$ and $S \in \mathcal{C}$.
\end{itemize}
The $\LCA$ must support parallel queries, i.e., when the $\LCA$ is executed independently on all the sets, the output should be consistent with a given solution to the set cover instance. Additionally, in our algorithms, the randomness can be presampled. That is, given a large common random string, the $\LCA$ executes a deterministic process. Hence the query complexity is completely deterministic, but the guarantee on the approximation ratio is in expectation over the randomness.

\subsection{Chernoff Bounds}

\begin{theorem}[Chernoff bound]\label{lemma:chernoff}
	Let $X_1, \ldots, X_k$ be independent random variables taking values in $[0, 1]$. Let $X \eqdef \sum_{i = 1}^k X_i$ and $\mu \eqdef \E{X}$. Then,
	\begin{enumerate}[(A)]
		\item\label{item:delta-at-most-1} For any $\delta \in [0, 1]$ it holds that $\prob{|X - \mu| \ge \delta \mu} \le 2 \exp\rb{- \delta^2 \mu / 3}$.
		\item\label{item:delta-at-most-1-le} For any $\delta \in [0, 1]$ it holds that $\prob{X \le (1 - \delta) \mu} \le \exp\rb{- \delta^2 \mu / 2}$.
		\item\label{item:delta-at-most-1-ge} For any $\delta \in [0, 1]$ it holds that $\prob{X \ge (1 + \delta) \mu} \le \exp\rb{- \delta^2 \mu / 3}$.
		\item\label{item:delta-at-least-1} For any $\delta \ge 1$ it holds that $\prob{X \ge (1 + \delta) \mu} \le \exp\rb{- \delta \mu / 3}$.
	\end{enumerate}
\end{theorem}

\begin{corollary}
\label{clm:weird:chernoff}
    Suppose $X_1, X_2, \dots X_t$ are i.i.d Bernoulli random variables such that $\Pr[X_i = 1] = p$ and $tp \leq \mu$. Then, we have $\Pr[\sum X_i > 2 \mu] \leq \exp{(-\mu / 3)}$.
\end{corollary}
\begin{proof}
Let $\overline{\mu} = \xpect{\sum_i X_i}$.
    We apply the Chernoff upper tail bound setting $\delta$ such that $(1+\delta)\overline{\mu} = 2\mu$. We have  $\delta \geq 1$ since $\overline{\mu} = tp \leq \mu$. We get $\Pr[\sum X_i > 2 \mu] \leq \exp{(-\overline{\mu} \delta / 3)} = \exp{(-(1+\delta)\overline{\mu}/3)} \cdot \exp{(\overline{\mu}/3)} \leq \exp{(-2\mu / 3)} \cdot \exp{(\mu/3)} \leq \exp{(-\mu/3)}$.
\end{proof}
\section{Warmup: A Weaker $\LCA$ for Fractional Set Cover}
\label{sec:warmup}

Let $\Delta$ be the maximum size of a set and $f$ be the maximum number of sets any element is present in. 
The state-of-the-art set cover LCA~\cite{GMRV20} takes $2^{O(\log^2 \Delta + \log f \cdot \log \Delta (\log\log \Delta + \log \log f))}$ queries to output whether a set is in the cover. 
Observe that the case of $f = 2$ reduces to vertex cover, which, by the discussion in the introduction, can be computed by $\poly(\Delta)$ $\LCA$ queries. 

We gradually present our ideas to obtain the advertised query complexity. In this section, we provide a weaker version of our algorithm as a warm-up. 
In the weaker version, the algorithm returns a fractional set cover. 
It can be converted to an integral one, as we explain in \cref{sec:frac-to-integral}, but we defer the analysis for now. The conversion to an integral solution causes an additional loss in the approximation ratio by a factor of at most $9$.

This version already improves the query complexity from $2^{O(\log^2 \Delta + \log f \cdot \log \Delta \cdot \log \log (\Delta f))}$ to a complexity of $2^{O(\log f \cdot \log \Delta \cdot \log \log(\Delta f))}$. Our algorithm is never worse in asymptotic complexity, as the first term is completely shaved off. In particular, if we consider the range where $\log \Delta \in \omega(\log f \log \log(\Delta f))$, we already have a clear asymptotic improvement over the state of the art. 

The starting point of our approach is an algorithm for the fractional set cover problem.
Specifically, given a set cover instance with elements $\cE$ and sets $\cS$, the aim of the set cover problem is to assign weights $\wt : \cS \rightarrow [0, 1]$ to the sets such that the following constraint is satisfied
\begin{equation*}
    \forall e \in \cE: \ \ \ \ \ \wt(e) \longDefiningEquals \sum\limits_{e \in S} \wt(S) \geq 1,
\end{equation*}
and $\|\wt\|_1$ is minimized.
\cref{alg:frac:setcover:base} provides an approach that approximately solves fractional set cover. In that algorithm, it is instructive to think of $\ell = r = 1$ and $j_S = 1$ for every set $S$ and every phase $i$. General values of $\ell, r$ and $j_S$ are needed for later analysis, and hence, we present \cref{alg:frac:setcover:base} in that form.
Intuitively, the parameters $\ell$ and $r$ provide us some slack in computing $\wt$ while designing our $\LCA$. More specifically, in Lines 9, 10 we can add weight to sets that are ``not so dense'', but we must ensure that if weight is added during iteration $j$, then it must be added during further iterations as well. Lines 13 and 14 show that for elements whose weight lies in $[\ell, r]$ we can ``arbitrarily'' afford to pretend that they are covered even if they are not.

\begin{algorithm}[h]
    \begin{algorithmic}[1]
        \Statex \textbf{Input:} Integers $\Delta, f$, positive reals $ \ell \leq 1 \leq r$, and set cover instance $\langle \mathcal{S}, \mathcal{E} \rangle$
        \State $\cov \gets \emptyset$ \Comment{The set of elements covered so far}
        \State $\wt \gets [0]$ \Comment{The weight function so far}
        \For{$i = 1, 2, \dots \log \Delta$}
            {\color{blue}{
            \For{every set $S \in \mathcal{S}$ (in parallel)}
                \State Arbitrarily choose a ``cut-off'' threshold $j_S \in [1, \log f + 1]$
            \EndFor}} \Comment{\cref{gmrv20:alg} operates as setting $j_S = 1$ always}
            \For{$j = 1, 2, \dots \log f$}
                \For{every $S \in \mathcal{S}$ with $d(S) \geq \Delta / 2^i$ (in parallel)} \Comment{As a reminder, $d(S) \longDefiningEquals |S \setminus \cov|$}
                    \State $\wt(S) \gets \min(1, \wt(S) + 2^j / f)$
                \EndFor
                {\color{blue}{
                \For{every set $S$ with $d(S) \in [\Delta / 2^{i+2}, \Delta/2^{i})$ (in parallel)}
                    \If{$j < j_S$} 
                        $\wt(S) \gets \min(1, \wt(S) + 2^j / f)$
                    \EndIf 
                \EndFor}} \Comment{These lines are skipped in \cref{gmrv20:alg}, as $j_S = 1$}
                \For{every $e \in \cE$ (in parallel)} \label{line:rem:e}
                    \State {\color{blue}{Let $\tau_e$ be an arbitrary real from $[\ell, r]$
                    }} 
                    \Comment{In \cref{gmrv20:alg}, $\ell = r = 1$}
                    \If{$\wt(e) \ge \tau_e$}
                        \State $\cov \gets \cov \cup \{e\}$
                    \EndIf
                \EndFor
            \EndFor
        \EndFor
        \State \Return $\wt / \ell$ \Comment{All weights re-scaled by factor $1 / \ell$}
    \end{algorithmic}
    \caption{A $\LOCAL$ $O(\log \Delta)$ approximate fractional set cover algorithm. Differences from \cref{gmrv20:alg} are highlighted in blue.}
    \label{alg:frac:setcover:base}
\end{algorithm}

Using standard techniques, it is not difficult to show that \cref{alg:frac:setcover:base} returns a fractional set cover which is at most $O(\log \Delta)$ factor larger than the optimal solution.
\begin{restatable}{lemma}{lemmafracbaseapprox}
\label{lemma:frac-base-approx}
    Let $\opt$ denote the optimal integral set cover. \cref{alg:frac:setcover:base} returns a fractional set cover $\wt$ such that $$\sum \wt(S) \leq O(r / \ell \cdot \log \Delta) \cdot \opt.$$
\end{restatable}
We defer the proof of this claim to \cref{sec:proof-frac-base-approx}.

\subsection{(Approximate) \texorpdfstring{$\LCA$}{LCA} Simulation of \cref{alg:frac:setcover:base}}

In this section, we design a $\LCA$ that almost perfectly simulates \cref{alg:frac:setcover:base} while using much fewer queries than a naive simulation. The key component of our $\LCA$ is the design of two $\LCA$ oracles $\IsSetDense{}$ and $\IsEleCov{}$ that perform probabilistic tests of whether a set $S$ (resp. element $e$) is dense (resp. covered) during (resp. at the end of) iteration $(i, j)$. See \cref{alg:Is-Ele-Cov,alg:Is-Set-Dense}.

\paragraph{High-level description of oracles.} The $\IsSetDense$ oracle takes as input (i) an iteration $(i, j)$ and (ii) a set $S$, then performs a probabilistic test of whether the set $S$ was dense \textit{at the beginning of} iteration $(i, j)$. Similarly, the $\IsEleCov$ oracle takes as input (i) an iteration $(i, j)$ and an element $e$, then performs a probabilistic test of whether $e$ is covered \textit{at the end of} iteration $(i, j)$. These oracles recursively invoke each other as follows. $\IsSetDense{}$ queries a subset of elements in $S$ and invokes $\IsEleCov{}$ on these elements at \textit{previous iterations} to test whether $S$ is dense. Details of which elements are selected and at what iterations they are tested are discussed later. Similarly, $\IsEleCov{i, j, e}$ queries a subset of the sets $S$ containing $e$ and then invokes $\IsSetDense$ on these sets at iterations \textit{on or before} $(i, j)$. The details of how the sets are chosen and the iterations in which they are tested are discussed later. Similar oracles are constructed in \cite{GMRV20}. The key difference is our observation that several recursive invocations can be pruned.

Using the oracles $\IsSetDense{}$ and $\IsEleCov{}$ we can simulate \cref{alg:frac:setcover:base}. We design two $\LCA$ oracles $\GetWeight{}$ and $\lcaWeight{}$.
$\GetWeight{}$ takes as input a set $S$ and returns an \textit{estimate} of the weight $\wt(S)$ assigned by \cref{alg:frac:setcover:base} to $S$. The estimate is such that for every set $S$, the probability that $\wt(S)$ computed by the two algorithms differs is $o(1 / \Delta f)$. Due to this, it is not guaranteed that $\wt(S)$ returned is a valid cover. 
The oracle $\lcaWeight{}$ further modifies the weights so that $\wt$ is a valid cover.

\paragraph{Description of \GetWeight{$S$} and \lcaWeight{$S$} (\cref{alg:get-set-weight}).}
The goal of procedure \GetWeight{$S$} is to compute the fractional weight of a given set $S$ after a probabilistic simulation of \cref{alg:frac:setcover:base}.
As in \cref{alg:frac:setcover:base}, \GetWeight{$S$} runs iterations $(1, 1)$ through $(\log \Delta, \log f)$. 
In iteration $(i, j)$, if $S$ is dense enough in that iteration and $\wt(S) \le 1$, \GetWeight{$S$} increases $\wt(S)$ by $2^j / f$.
However, instead of testing the density by asking whether $d(S) \ge \Delta / 2^i$ as done in \cref{alg:frac:setcover:base}, \GetWeight{$S$} performs that test by invoking \IsSetDense{$i, j, S$}, e.g., see \cref{line:is-set-dense} of \GetWeight{$S$}. Note that the $\IsSetDense{}, \IsEleCov{}$ oracles are purely a function of the topology of the graph and the shared random bits of the algorithm. 

Since our $\LCA$ simulation of \cref{alg:frac:setcover:base} is noisy, it can happen that after the set weights $\wt$ are computed, it still holds some elements are not covered. To ensure that all elements are covered, we compute the weight of every element $e$. If it so happens that $e$ is still uncovered, then we cover it ``naively'', i.e., by choosing an arbitrary set that contains it.
Hence, we are guaranteed that the output of \lcaWeight{$S$} makes sure such elements are covered.

\begin{algorithm}[h]
    \begin{algorithmic}[1]

        \Procedure{GetWeight}{$S$}
            \State $\wt(S) \gets 0$
            \For{$i = 1, 2, \dots \log \Delta$}
                \For{$j = 1, 2, \dots \log f$}
                    \State $\mathsf{density} \gets \Call{IsSetDense}{i, j, S}$
                    \If{$\mathsf{density} = \mathsf{Dense}$ \label{line:is-set-dense}} 
                        \State $\wt(S) \gets \min(1, \wt(S) + 2^j / f)$ 
                        % \Comment{For the integral version, sample $S$ with probability $2^j / f$}
                    \ElsIf{$\mathsf{density} = \mathsf{Bad}$} \Comment{With very low probability, the tests fail}
                        \State $\wt(S) \gets 1$
                    \Else 
                        \State \textbf{break} \Comment{Ensure that density tests for $S$ in a given phase is monotonous}
                        \Statex \Comment{i.e., move to the next phase $i$, once $\Call{IsSetDense}{i, j, S}$ returns $\mathsf{False}$}
                \EndIf
                \EndFor 
            \EndFor
            \State \Return $\wt(S)$
        \EndProcedure
        \Procedure{LCA-weight}{$S$}
            \State $\wt(S) \gets \Call{GetWeight}{S}$
            \For{$e \in S$}
                \State $\wt(e) \gets \sum_{S' \ni e} \GetWeight{S'}$
                \If{$\wt(e) < 1$ and $S$ has the smallest ID among the sets containing $e$}
                    \State $\wt(S) \gets 1$ \Comment{Ensure that $e$ is covered}
                \EndIf
            \EndFor
            \State \Return $\wt(S)$
        \EndProcedure
    \end{algorithmic}
    \caption{An LCA for computing a fractional set cover}
    \label{alg:get-set-weight}
\end{algorithm}  
\paragraph{Description of \IsSetDense{$i, j, S$} (\cref{alg:Is-Set-Dense}).}
Given indices $i$ and $j$, and a set $S$, this procedure attempts to answer the question: ``Is set $S$ dense at the beginning of iteration $(i, j)$?''
A direct way of answering this question is to count the uncovered elements of $S$ by the beginning of iteration $(i, j)$.
We aim to perform the count more efficiently as follows: the algorithm samples $t$ elements on \cref{line:dense-est-sample} and counts how many of the sampled ones are uncovered. It is instructive to think of $t$ as significantly smaller than the size of $S$.
The idea here is that if $S$ is dense, then among the $t$ elements, sufficiently many of them will be uncovered.

\begin{algorithm}
\begin{algorithmic}[1]
    \Procedure{Sample}{$A$, $t$, $m$}
        \Statex Given an ordered set $A$ of size at most $m$, sample a multiset $S$ of at most $t$ elements from $A$ such that for each $a \in A$, $\Pr[S_i = a] = 1/m$.
        \State Initialize a multiset $S \gets \emptyset$
        \For{each $i = 0, 1, 2, \dots t-1$}
            \State Sample an integer $i \in [m]$
            \State If $i \leq |A|$, then add the $i^{th}$ element of $A$ to $S$
        \EndFor 
        \State \Return $S$.
    \EndProcedure 
\end{algorithmic}
\caption{A generic sampling procedure that we use.}
\label{alg:sample}
\end{algorithm}

\begin{algorithm}[h]
    \begin{algorithmic}[1]
        \Statex $L \longDefiningEquals \log \Delta \cdot \log f$
        \Procedure{IsSetDense}{$i, j, S$}
            \State \textbf{If} $|S| < \Delta / 2^i$ \textbf{Return } $\mathsf{Light}$
            \State $E_{i, j, S} \gets \Call{Sample}{S, 2^i L^3, \Delta}$ (See \cref{alg:sample}) \label{line:dense-est-sample} 
            \Statex \Comment{a multiset of $2^i \cdot L^3$ elements in $S$, each sampled u.a.r. from $S$ with repetition} 
            \Statex \Comment{The randomness used for this sampling is the same for a given $i, j, S$.}
            \For{$i'=1, 2, \dots i-1$ \label{line:for-i'-is-set-dense}}
                \For{every distinct $e \in E_{i, j, S}$} 
                \Statex \Comment{Crucial to visit only distinct elements for the desired query complexity}
                    \If{$\Call{IsEleCov}{i', \log f, e}$}
                    \State Remove every occurrence of $e$ from $E_{i, j, S}$ \label{line:is-set-dense-remove-e}
                    \EndIf
                \EndFor
                \If{ \# of distinct elements in $E_{i, j, S}> 2^{i-i'+2} \cdot L^3$\label{line:bad-set}} 
                    \State \Return $\mathsf{Bad}$ 
                \EndIf
            \EndFor
        \For{every distinct $e \in E_{i, j, S}$\label{line:for-over-Eis}}
            \If{$\Call{IsEleCov}{i, j-1, e}$}
                \State Remove every occurrence of $e$ from $E_{i, j, S}$
            \EndIf
        \EndFor
        \If{$|E_{i, j, S}| \geq L^3/2$}\label{line:test-set-density} 
        \Statex \Comment{$|E_{i, j, S}|$ is the size of the multiset, i.e., repeated elements are counted.}
            \State \Return $\mathsf{Dense}$
        \Else
            \State \Return $\mathsf{Light}$
        \EndIf
        \EndProcedure

    \end{algorithmic}
    \caption{A procedure that returns whether an input set $S$ is dense at beginning of $(i, j)$.}
    \label{alg:Is-Set-Dense}
\end{algorithm}  

The test of whether an element $e$ is covered by the end of iteration $(i, j)$ is carried out by the procedure \IsEleCov{$i, j, e$}.
The dependence of the number of queries performed by \IsEleCov{$i, j, e$} is expected to be exponential in $i$. 
That is why \cref{alg:Is-Set-Dense} makes sure to test the smaller values $i' < i$ first, before moving to $i$; 
see \cref{line:for-i'-is-set-dense}.
All covered elements are removed by \cref{line:is-set-dense-remove-e} and are not tested for larger values of $i'$.

Since the calculation is based on a sample of elements, it can sample too many or too few uncovered elements of $S$.
However, testing too many elements can lead to high query complexity. We prevent this by the if-block on \cref{line:bad-set}. When that if-condition evaluates to \true, it effectively means that somewhere in the computation process, our algorithm has made an error compared to the computation of \cref{alg:frac:setcover:base}. We show that such an error occurs relatively rarely.

\paragraph{Description of \IsEleCov{$i, j, e$} (\cref{alg:Is-Ele-Cov}).}
The task of procedure \IsEleCov{$i, j, e$} is to return whether an element $e$ is covered by the end of iteration $(i, j)$.
It is done by computing -- or more precisely, approximating -- the sum of fractional weights of the sets containing $e$; that sum is denoted by $\wt(e)$.
The same as in \IsSetDense{$i, j, S$}, $\wt(e)$ is approximated from a sample of the sets containing $e$; this sample is taken on \cref{line:sample-cS} of \cref{alg:Is-Ele-Cov}.
The rest of this procedure consists of gradually computing $\wt(e)$ by adding to it fractional weights of dense sets.
\begin{algorithm}[h]
    \begin{algorithmic}[1]
        \Statex $L \longDefiningEquals \log \Delta \cdot \log f$
        \Procedure{IsEleCov}{$i, j, e$}
            \If{$\Call{IsEleCov}{i, j-1, e}$\label{line:is-ele-cov-j-1}}
                \Return \true
            \EndIf
            \State $\cS_{i, j, e} \gets \Call{Sample}{N(e), 2^j L^3, f}$ \Comment{$N(e)$ denotes the sets containing $e$}
            \Statex \Comment{a multiset of $2^j L^3$ sets containing $e$ chosen u.a.r. with repetition.}
            \Statex \Comment{This sampling is common for every call for a given $i, j, e$.}
            \label{line:sample-cS}
            \State $\mathsf{wt}(e) \gets 0$
            \For{every $S \in \cS_{i, j, e}$} \Comment{Multiple occurrences of $S$ are visited multiple times.}
                \State $k \gets 1$
                \While{$k \leq j$}
                \If{$\Call{IsSetDense}{i, k, S}=\mathsf{Dense}$}
                    \State $\wt(e) \gets \wt(e) + 2^{k-j}/L^3$
                \If{$\wt(e) \geq 1$}
                    \State \Return \true \label{line:wte-ge-1}
                \EndIf
                    \State $k \gets k + 1$
                \ElsIf{$\Call{IsSetDense}{i,k,S}=\mathsf{Bad}$}
                    \State \Return \true
                \Else 
                    \State \textbf{break}
                \EndIf
                \EndWhile
            \EndFor
            \State \Return \false
        \EndProcedure

    \end{algorithmic}
    \caption{A procedure that returns whether an element $e$ is covered at the end of iteration $(i, j)$.}
    \label{alg:Is-Ele-Cov}
\end{algorithm}

The aim of the rest of this section is to prove the following claim.
\begin{theorem}\label{thm:alg-get-set-weight}
    \cref{alg:get-set-weight} returns an $O(\log \Delta)$ approximate fractional set cover in expectation using at most $\Delta^{O(\log f \cdot (\log \log \Delta + \log \log f))} = L^{O(L)}$ queries, where $L$ is the round complexity of \cref{alg:frac:setcover:base}.
\end{theorem}
In \cref{sec:alg-get-weight-query}, we analyze the query complexity, while \cref{sec:alg-get-weight-approx} analyzes the approximation.

\subsection{Query complexity}
\label{sec:alg-get-weight-query}
The main aim of this section is to analyze the query complexity of \cref{alg:Is-Ele-Cov,alg:Is-Set-Dense} and prove the following claim.
\begin{lemma}\label{lemma:is-ele-cov-queries}
    $\IsEleCov{i, j, e}$ (\cref{alg:Is-Ele-Cov}) and $\IsSetDense{i, j, S}$  (\cref{alg:Is-Set-Dense}) performs at most $C^{(i-1) \log f + j}$ queries, where we let $C = \poly(\log \Delta \log f)$. 
\end{lemma}

Having \cref{lemma:is-ele-cov-queries} in hand, the query complexity of \cref{alg:get-set-weight} follows by setting $i = \log \Delta $ and $j = \log f$.
\begin{lemma}
    $\GetWeight{S}$ (\cref{alg:get-set-weight}) performs $\Delta^{O(\log f \cdot (\log \log \Delta + \log \log f))}$ queries.
\end{lemma}

\begin{proof}
    Note that $\GetWeight{S}$ makes at most $L = \log \Delta \log f$ calls to $\Call{IsSetDense}{i, j, S}$ and at most $\Delta$ calls to $\Call{IsEleCov}{i, j, \cdot}$.
    By \cref{lemma:is-ele-cov-queries}, these calls result in $(L + \Delta) \cdot L^{O(L)}$ (recursive) invocations of \cref{alg:Is-Ele-Cov,alg:Is-Set-Dense}.
    Each invocation makes at most $\Delta + f$ queries to sample the respective sets or elements.

    Therefore, the total query complexity of $\Call{GetWeight}{S}$ is $(\Delta + f) \cdot (\Delta + L) \cdot L^{O(L)} \in L^{O(L)}$. 
\end{proof}

\begin{proof}[Proof of \cref{lemma:is-ele-cov-queries}]
    We analyze the worst-case number of recursive calls made by the algorithm. 
    Let $E(i, j)$ and $S(i,j)$ denote the (worst case) number of recursive calls generated by $\Call{IsEleCov}{i,j}$ (\cref{alg:Is-Ele-Cov}) and $\Call{IsSetDense}{i,j}$ (\cref{alg:Is-Set-Dense}), respectively. Our proof proceeds by induction.
    \begin{center}
        \emph{Inductive hypothesis}: $S(i, j), E(i, j) \le C^{(i-1)\log f + j}$, where $C = \poly(\log \Delta \log f)$. 
    \end{center}
    
    The base case is when $i=1, j=1$ and in this case, we have that $S(1, 1) = E(1, 1) = 0$. We use strong induction, i.e., we suppose that the statement is true for every $(i', j')$ such that either $i' < i$ or $i' = i$ and $j' < j$. We then prove the statement for $(i, j)$.

    \paragraph{An upper-bound on $S(i, j)$.}
        Our upper-bound on $S(i, j)$ is can be expressed as
    \[
        S(i, j) \leq L^3 \sum\limits_{i'=1}^{i-1} 2^{i-i'+3} E(i', \log f) + 8L^3 E(i, j-1).
    \]
    To derive this upper-bound, observe that $\Call{IsSetDense}{i, j, S}$ invokes $\Call{IsEleCov}{i', \log f, \cdot }$ for $i' < i$ and also invokes $\Call{IsEleCov}{i, j-1, \cdot}$ within the loop on \cref{line:for-over-Eis}. 
    When $\Call{IsEleCov}{i', \log f, \cdot}$ is invoked, the call $\Call{IsSetDense}{i, j, S}$ must not have terminated in the previous iteration $i'-1$. This implies that $|E_{i, j, S}| \leq 2^{i-(i'-1)+2} \cdot L^3$ and hence the number of calls is at most $2^{i-i'+3} \cdot L^3$.
    Moreover, when the loop over $i'$ is over, it is guaranteed that $|E_{i, j, S}| < 2^{i - (i - 1) + 2} L^3 = 8 L^3$. This, together with the loop on \cref{line:for-over-Eis}, implies the second term in the upper bound on $S(i, j)$.

    \paragraph{An upper-bound on $E(i, j)$.}
    During the call $\Call{IsEleCov}{i, j}$, suppose there are $\mu_k$ invocations to $\Call{IsSetDense}{i, k}$. 
    Then, we can upper-bound $E(i, j)$ as
    \begin{equation*}
        E(i, j) \leq E(i, j-1) + \sum\limits_{k=1}^j \mu_k S(i, k).
    \end{equation*}
    The first term in the upper-bound corresponds to \cref{line:is-ele-cov-j-1} of \cref{alg:Is-Ele-Cov}, while the sum follows by definition and the fact that the only remaining recursive invocations are to \cref{alg:Is-Set-Dense}.
    
    We also have that
    \begin{equation}\label{eq:sum-mu-k-upper-bound}
        \sum\limits_{k=2}^j \mu_k 2^{(k - 1)-j} / L^3 \leq 1
    \end{equation}
    \cref{eq:sum-mu-k-upper-bound} follows since every recursive call to $\Call{IsSetDense}{i, k, S}$ for $k > 1$ must have survived a call to $\Call{IsSetDense}{i, k-1, S}$, contributing $2^{(k-1)-j}/L^3$ to $\wt(e)$. 
    Moreover, because of \cref{line:wte-ge-1} of \cref{alg:Is-Ele-Cov}, note that $\wt(e)$ must have been less than $1$ prior to the call to $\Call{IsSetDense}{i, k, S}$.

    Since $|\cS| \le 2^j L^3$ by construction (see \cref{line:sample-cS} of \cref{alg:Is-Ele-Cov}), we trivially have that $\mu_1 \le 2^j L^3$. 
    From \cref{eq:sum-mu-k-upper-bound}, we derive $\sum\limits_{k=2}^j \mu_k \leq 2^j L^3$.
    Hence, we conclude
    \begin{equation}\label{eq:sum-mu-k}
    \sum\limits_{k=1}^j \mu_k = \mu_1 + \sum\limits_{k=2}^j \mu_k \leq 2^j L^3 + 2^j L^3 \leq 2^{j+1} L^3.
    \end{equation}
    We now substitute the expression for $S(i, j)$ into $E(i, j)$ and obtain the following
    \begin{align}
        E(i, j) &\leq E(i, j-1) + L^3 \sum\limits_{k=1}^j \rb{\mu_k \sum\limits_{i'=1}^{i-1} 2^{i-i'+3} E(i', \log f)} + 8L^3\sum\limits_{k=1}^j \mu_k E(i, k-1) \nonumber \\
        &= E(i, j-1) + L^3 \rb{\sum\limits_{k=1}^j \mu_k} \cdot \rb{\sum\limits_{i'=1}^{i-1} 2^{i-i'+3} E(i', \log f)} + 8L^3\sum\limits_{k=1}^j \mu_k E(i, k-1) \nonumber  \\
        &\stackrel{\textnormal{\cref{eq:sum-mu-k}}}{\leq} E(i, j-1) + 2^{j+1} L^6 \sum\limits_{i'=1}^{i-1} 2^{i-i'+3} E(i', \log f) + 8L^3 \sum\limits_{k=1}^j \mu_k E(i, k-1), \label{eq:E-ij-upper-bound}
    \end{align}
    where, to achieve the equality above, we used the fact that the inner summation, which is over $i'$, does not depend on $k$.
    
    We have that $E(i', \log f) \leq C^{i' \log f}$ from the induction hypothesis. Using this, we upper bound the summation in the second term of \cref{eq:E-ij-upper-bound} as follows,
    \begin{align*}
        \sum\limits_{i'=1}^{i-1} 2^{i-i'+3} E(i', \log f) & \leq 2^{i+3} \sum\limits_{i'=1}^{i-1} (C^{\log f} / 2)^{i'}
        = 2^{i+3} \rb{\frac{(C^{\log f} / 2)^{i} - 1}{C^{\log f} / 2 - 1}} 
        \\
        & \le 2^{i+3} \frac{(C^{\log f} / 2)^{i}}{\frac{1}{2} \cdot C^{\log f} / 2} = 2^{i+4} (C^{\log f} / 2)^{(i-1)} = 32 C^{(i-1)\log f} 
    \end{align*}
    For the third term of \cref{eq:E-ij-upper-bound}, we use induction hypothesis to bound $E(i,k-1) \leq C^{(i-1)\log f + k-1}$ which gives,
    \begin{equation*} \begin{split}
        \sum\limits_{k=1}^j \mu_k E(i, k-1) & \leq \sum\limits_{k=1}^j \mu_k C^{(i-1)\log f + k-1} 
        = C^{(i-1)\log f} \cdot 2^j \cdot \sum\limits_{k=1}^j \rb{\mu_k 2^{k - j - 1} \cdot C^{k-1} 2^{-(k - 1)}} \\
        & \leq C^{(i-1)\log f} \cdot 2^j \cdot L^3 \cdot \rb{\sum\limits_{k=1}^{j} \mu_k 2^{k-j-1} / L^3} \cdot \sum\limits_{k=1}^j (C / 2)^{k-1} \\
        & \leq C^{(i-1) \log f} \cdot 2^j L^3 \cdot 2 \cdot (C/2)^{j-1} \leq 4 \cdot L^3 \cdot C^{(i-1) \log f + j-1}
    \end{split} \end{equation*}

    Using these bounds, we derive
    \[
        E(i, j) \leq C^{(i-1) \log f + j} \cdot (C^{-1} + 64 L^6 (C/2)^{-j} + 32 L^6 C^{-1}) \leq C^{(i-1) \log f + j},
    \]
    whenever $C \geq 256 L^6$. 
    Using the upper bound on $S(i, j)$ we get
    \[
        S(i, j) \leq L^3 \sum\limits_{i'=1}^{i-1} 2^{i-i'+3} E(i', \log f) + 8L^3 E(i, j-1).
    \]
    The first term on the RHS is upper-bounded by $L^3 \cdot 32C^{(i-1) \log f}$. 
    The second term can be upper-bounded by $8L^3 C^{(i-1) \log f + (j-1)}$. 
    Substituting we upper-bound $S(i,j) \leq C^{(i-1)\log f + j} \cdot \rb{8L^3 C^{-1} + 32 L^3 C^{-j}} \leq C^{(i-1)\log f + j}$. 
    
\end{proof}

\subsection{Approximation ratio}
\label{sec:alg-get-weight-approx}
We next show the approximation ratio of \cref{alg:get-set-weight}. Towards this, we define certain ``bad'' events. If none of these events occur, then the LCA simulates \cref{alg:frac:setcover:base} exactly for suitably chosen thresholds $\ell$ and $r$. 
We then show that for every set $S$ with high probability, i.e., with probability at least $1-(\Delta f)^{-1}$, no bad event occurs during the calls $\Call{IsSetDense}{i,j,S'}, \Call{IsEleCov}{i, j, e'}$ for all $i,j$ and any $e', S'$ in the $L$-hop neighborhood of $S$. 
Consequently, even if we include all such sets wherein such bad events happen, the expected total weight of sets remains $O(\log \Delta \cdot \opt)$.

\begin{mdframed}[backgroundcolor=gray!10, linecolor=brown!40!black, roundcorner=5pt]
    \textbf{Bad Event $\cE_1$ -- Too dense set}: \\
    The set $E_{i, j, S}$ obtained on \cref{line:dense-est-sample} of \IsSetDense{$i, j, S$} contains at least $2^{i-i'+2} L^3$ elements that are uncovered at the beginning of iteration $(i', 1)$ for some $i' < i$.
\end{mdframed}
Observe that $\cE_1$ occurs when the condition on \cref{line:bad-set} of \cref{alg:Is-Set-Dense} evaluates to \true.

\begin{mdframed}[backgroundcolor=gray!10, linecolor=brown!40!black, roundcorner=5pt]
    \textbf{Bad Event $\cE_2$ -- Poor element sample}: \\
    Consider some iteration $(i, j)$ and set $S$. Let $d(S)$ denote the number of uncovered elements of $S$ in \cref{alg:get-set-weight} at the beginning of iteration $(i, j)$. 
    Let $E$ be the set $E_{i, j, S}$ just before the return statement of \cref{alg:Is-Set-Dense}.
    We say $\cE_2$ happened during the call $\IsSetDense{i,j, S}$ if either
    \begin{enumerate}[(i)]
        \item $d(S) \geq \Delta / 2^{i}$ and $|E| < L^3 / 2$, or
        \item $d(S) < \Delta / 2^{i+2}$ and $|E| \geq L^3 / 2$.
    \end{enumerate}
\end{mdframed}

\begin{mdframed}[backgroundcolor=gray!10, linecolor=brown!40!black, roundcorner=5pt]
    \textbf{Bad Event $\cE_3$ -- Loose element-weight estimate}:\\
    Consider some iteration $(i, j)$ and element $e$. Let $\widehat{\wt}(e)$ denote the weight of $e$ calculated in \cref{alg:Is-Ele-Cov} during iteration $(i, j)$. 
    % \stodo{Did you want to say in \cref{alg:Is-Ele-Cov} instead?}
    We say $\cE_3$ happened if either 
    \begin{enumerate}[(i)]
        \item $\widehat{\wt}(e) < 1/2$ and $\wt(e) \geq 1$, or
        \item $\widehat{\wt}(e) \geq 2.5$ and $\wt(e) < 1$.
    \end{enumerate}
\end{mdframed}

\begin{lemma}
\label{sim:claim}
For every set $S$ such that no bad events occurred during the calls of the LCA oracles within the $L$-hop neighborhood of $S$, there exists an assignment of thresholds to \cref{alg:frac:setcover:base} such that $\wt(S)$ is identical in both algorithms.

\end{lemma}
\begin{proof}
For every set $S$ and phase $i$, set $j_S$ to be the first iteration during which $S$ is not considered $\mathsf{Dense}$. Observe that weight is not added to the sets after iteration $j_S$ due to the break statement in \cref{alg:get-set-weight}. Also, observe that in $\IsEleCov{}$, for every set $S$, the density tests for any set $S$ at iteration $(i, j)$ occur only if the test at $(i, j-1)$ passes. Note that if there exists any call $\IsSetDense{i, j', S}$ for $j' > j$ that returns $\mathsf{True}$ whereas $\IsSetDense{i, j, S}$ returns $\mathsf{False}$, then this call is never made in a run of $\Call{getWeight}$.

Since no bad event occurred, we have that $\widehat{\wt}(e)$ and $\wt(e)$ during every iteration are at most $2.5$ factor away from each other. Hence there exists a real number $x \in [1/2.5, 2.5]$ such that $\widehat{\wt}(e) = x \cdot \wt(e)$, so $1/x$ is the suitable threshold.
\end{proof}

\begin{lemma}
\label{lem:1}
     The probability that one of the events $\cE_1, \cE_2, \cE_3$ occurs during the calls to \\ $\Call{IsSetDense}{i, j, S}$ and $\Call{IsEleCov}{i, j, e}$, including all recursive calls arising out of them, is at most $\exp{(-\Theta(L^3))}$.
\end{lemma}

\begin{proof}
    We prove the statement by induction on $(i, j)$. Define $t_{i,j} = (i-1) \log f + j$, i.e., $t_{i, j}$ is the index of iteration $(i, j)$.
    
    We show that the probability of $\cE_1$, $\cE_2$ or $\cE_3$ occurring from iteration $(1, 1)$ through $(i, j)$ is upper bounded by $((2\Delta+2f)^{t_{i,j} \cdot L}-1) \cdot L \cdot \exp{(-L^3/12)}$. 

    First, we upper-bound the probability of the bad events occurring, \textbf{given} that no bad event has occurred in the $L$-hop neighborhood of set $S$ \emph{before} iteration $(i, j)$. We analyze each of the events separately.

    \paragraph{Analyzing event $\cE_1$.}
    Since by our assumption, no bad event occurred \emph{before} iteration $(i, j)$, \cref{alg:frac:setcover:base} is simulated exactly until iteration $(i, j)$; the notion of ``exact simulation'' in this context is formally stated and proved by \cref{sim:claim}. 
    This implies that the number of uncovered elements in $S$ during phase $i' < i$ is at most $\Delta / 2^{i'-1}$; otherwise, at the end of iteration $(i'-1, \log f)$, $S$ would have been dense, and its weight would have increased to $1$. 

    There are at most $\Delta / 2^{i'-1}$ many uncovered elements in $S$ (because no bad event occurred). The set $E_{i,j,S}$ contains $2^i \cdot L^3$ randomly sampled elements. For each sample, the probability that the element is uncovered is at most $2^{1-i'}$. Note that the probability that any two sampled elements are uncovered is independent because of the sampling with repetition.
    The expected number of sampled elements is $ \mu \leq 2^{i-i'+1} \cdot L^3$. 
    
    Let $\overline{E}_{i, j, S}$ denote the subset of $E$ that contains only the uncovered elements in $E_{i, j, S}$.
    By \cref{clm:weird:chernoff}, we have,
    \[
        \Pr[|\overline{E}_{i, j, S}| > 2^{i-i'+2}L^3] \le \prob{|\overline{E}_{i, j, S}| > 2 \mu} \leq \exp(-(2^{i-i'+1}L^3 / 3)) \leq \exp{(-L^3/3)}. 
    \]

    \paragraph{Analyzing event $\cE_2$.}
    We consider two cases. 
    
    Case (i): $d(S) \geq \Delta / 2^i$. As a reminder, $d(S)$ is the number of uncovered elements in $S$ in \cref{alg:get-set-weight} at the beginning of iteration $(i, j)$. 
    Observe that \cref{alg:Is-Set-Dense} essentially removes from $E_{i, j, S}$ the elements covered before iteration $(i, j)$, i.e., it removes the elements not contributing to $d(S)$. So, by the time \cref{alg:Is-Set-Dense} reaches \cref{line:test-set-density}, each element in $E_{i, j, S}$ is not covered by \cref{alg:get-set-weight} by the beginning of iteration $(i, j)$.
    
    The expected number of those uncovered elements sampled to $E_{i, j, S}$ by \cref{alg:Is-Set-Dense} is $d(S) \cdot \frac{2^i L^3}{\Delta} \ge L^3$. 
    By the Chernoff bound, at least $L^3 / 2$ of those uncovered elements are sampled to $E$ with probability at least $1-\exp{(-L^3 / 2)}$. 
        
    Case (ii): $d(S) < \Delta / 2^{i + 2}$. 
    Similar to Case~(i), the expected number of elements in $E_{i, j, S}$ that are not covered by \cref{alg:get-set-weight} by the beginning of iteration $(i, j)$ is at most $L^3 / 4$. 
    By \cref{clm:weird:chernoff}, the probability that at least $L^3 / 2$ of those uncovered elements are sampled to $E_{i, j, S}$ by \cref{alg:Is-Set-Dense} is at most $\exp{(-L^3 / 12)}$.

    \paragraph{Analyzing event $\cE_3$.}

        Let $\widehat{\wt}(e)$ denote the weight of $e$ in \cref{alg:frac:setcover:base} at iteration $(i, j)$. 
        We write $\widehat{\wt}(e) = \sum\limits_{k=1}^j \mu_k 2^{k} / f$ where $\mu_k$ is the number of sets containing $e$ that remain dense from iteration $(i, 1)$ through $(i, k)$ in both \cref{alg:frac:setcover:base} and our LCA simulation; recall that we assume that our LCA simulation has so far simulated \cref{alg:frac:setcover:base} exactly. 
        Let $\wt(e)$ denote the estimate of $\widehat{\wt}(e)$ obtained by the LCA. 
        We split $\wt(e) = \wtfew(e) + \wtmany(e)$ into two parts: (i) $\wtfew(e)$ is the weight contribution of the sets corresponding to $\mu_k < f / 2^{j + 2}$ and (ii) $\wtmany(e)$ is the weight contribution of the sets corresponding to $\mu_k \geq f / 2^{j+2}$.
        
        (Case (i)) By definition, we have
            \[
                \wtfew(e) \le \sum_{k = 1}^j \frac{f}{2^{j + 2}} \cdot \frac{2^{k}}{f} = \frac{1}{2^{j + 2}} \sum_{k = 1}^j 2^k < \frac{1}{2}.
            \]

        (Case (ii))
        To upper-bound $\wtmany(e)$, we first note that $\cS_{i, j}$ obtained on \cref{line:sample-cS} of \cref{alg:Is-Ele-Cov}, in expectation, contains at least $2^j L^3 \cdot \mu_k / f \ge L^3/4$ sets which are dense upto iteration $(i, k)$. 
        
        Therefore, at least half and at most twice the fraction of them are sampled with probability at least $1 - 2\exp(-L^3/12)$. Thus $2 \widehat{\wt}(e) \geq \wt(e) \geq (\widehat{\wt}(e) - 1/2) / 2$.

        Hence except with probability $2\exp(-L^3/12)$, if $\widehat{\wt}(e) < 1/2$ then $\wt(e) < 1$ and if $\widehat{\wt}(e) \geq 2.5$, then $\wt(e) \geq 1$.
        
    \paragraph{Combining the analysis.}
    For large enough $L$, one of the three events occurs with probability at most $L \exp{(-L^3 / 12)}$. 
    For a given set $S$ or element $e$, by the induction hypothesis, the probability that it sees a bad event in its $L$-hop neighborhood until end of iteration prior to $(i, j)$ is at most $((2\Delta+2f)^{(t_{i,j} - 1)L} - 1) \cdot L \cdot \exp{(-L^3 / 12)}$. Our total failure probability can be bound by applying the union bound that such a failure happens in the $L$-hop neighborhood of the set/element plus the failure of one of the events $\cE_1, \cE_2, \cE_3$ discussed above. This gives
    \begin{align*}
        & L \exp{(-L^3 / 12)} + 2 \cdot (\Delta + f)^L \cdot ((2\Delta + 2f)^{(t_{i,j}-1)L}-1) \cdot L \cdot \exp{(-L^3 / 12)} \\
        \leq & ((2\Delta + 2f)^{t_{i,j}L} - 1) \cdot L \cdot \exp{(-L^3/12)}
    \end{align*}

    For the last expression, we used $x / 2^{L-1} \leq x - 1$ when $x \geq 2$ and $L \geq 2$.
    Finally setting $t_{i, j} = L$ and large enough $\Delta, f$, the failure probability is at most $\exp{(-L^3/24)}$.
\end{proof}

From \cref{lem:1}, we can conclude that for large enough $\Delta, f$, in the recursive \cref{alg:get-set-weight} the call $\Call{getWeight}{S}$ which makes at most $\poly(\Delta f)$ queries to $\Call{IsSetDense}$ and $\Call{IsEleCov}$ observes a bad event with probability at most $\exp(-L^3 / 24) \leq 1 / (\Delta f)$ (for large enough $\Delta, f$) and otherwise the weight is identical to that output by \cref{alg:frac:setcover:base} (with suitably chosen thresholds) from \cref{sim:claim}. 
Hence, the expected number of elements that get naively covered at the end is at most $n / \Delta \leq \opt$.
This concludes the proof of \cref{thm:alg-get-set-weight}.

\subsection{Proof of approximation ratio of \cref{alg:frac:setcover:base}}
\label{sec:proof-frac-base-approx}
\lemmafracbaseapprox*
\begin{proof}
    Observe that by \cref{line:rem:e}, at the point an element $e$ just got covered in \cref{alg:frac:setcover:base} it holds that $\wt(e) \geq \ell$. 
    Moreover, $e$ is covered at some iteration as otherwise during iteration $(\log \Delta, \log f)$ every set adjacent to $e$ has weight at least $1 \geq \ell$.
    Suppose $e$ gets covered at the end of iteration $(i, j)$. 
    
    We know that $\wt(e) < r$ at the end of iteration $(i, j-1)$, as $e$ was not added to $\cov$ before iteration $(i, j)$. Consider a set $S$ that contains $e$ whose weight changed during $(i, j)$. 
    % \rtodo{}
    At the end of iteration $(i, j-1)$ we have
    \[
        \wt(S) \geq (1 + 2 + \dots 2^{j-1}) / f = (2^j - 1) / f. 
    \]
    $\wt(S)$ does not increase by more than a factor of $3$ at the end of iteration $(i, j)$. 
    Summing over all the sets containing $e$, we conclude that $\wt(e) \leq 3r$ at the time it was removed in \cref{line:rem:e}.

    We now split $\wt(S)$ into $\log \Delta$ parts $\wt_1(S), \wt_2(S) \dots$, where $\wt_i(S)$ is the portion of $\wt(S)$ added during iterations $(i,1)$ to $(i, \log f)$.

    Our strategy is to show that $\wt_i(S)$, individually for every $i$, is an $O(1)$ approximation of the optimal set cover. 
    To that end, consider a set-cover instance whose elements are $\cE_i = \cE \setminus \cov$ and the sets are $\cS_i = \{S \in \cS \mid \sum\limits_{k=1}^{i-1} \wt_k(S) < \ell\}$. 
    Observe that if $\wt(S) < \ell$ at the end of iteration $i-1$, then it must be the case that $d(S) < \Delta / 2^{i-1}$; otherwise, $\wt(S)$ would be increased each iteration from $(i-1, 1)$ to $(i-1, \log f)$, resulting in $\wt(S) \ge 1 \ge \ell$. 
    The instance $\langle \cS_i, \cE_i\rangle$ has max set size $\Delta_i = \Delta / {2^{i-1}}$ and max frequency $f_i = f$. 
    Let $X_e$ denote the total weight added to $e$ during iterations $j = 1$ to $j = \log f$. 
    By our previous discussions, $X_e \leq 3 r$. Every set $S$ contributes at least $d(S) \cdot \wt_i(S)$ to $\sum X_e$.  We have $d(S) \geq \Delta / 2^{i+2} = \Delta_i / 8$.
    Hence, $\Delta_i \wt_i(S) / 8 \leq \sum\limits_{e \in \cE_i} X_e \leq 3r |\cE_i|$.
    It follows that $\wt_i(S)$ is a $24r$-approximation of the instance $\langle \cS_i, \cE_i \rangle$ and hence the original instance. Finally, the re-scaling loses a $1 / \ell$ factor, giving our result.
\end{proof}
\section{Approximation Ratio Analysis of a Noisier Algorithm}
\label{sec:prob-est}
In this section, we show that in \cref{alg:frac:setcover:base}, it is not necessary to have high probability estimates of the effective degree of sets. 
We show that if the estimates are obtained only with constant probability and some \emph{additional constraints}, we can still achieve an $O(\log \Delta)$ approximation.

\begin{algorithm}
    \begin{algorithmic}[1]
        \State Set $\wt(S) \gets 0$ for every set $S$
        \For{$i=1,2,\dots \log \Delta$}
        \For{$j=1,2,\dots \log f$}
            \For{every set $S$ (in parallel)}
                \State $\derr(S) \gets $
                 an estimate for $d(S) \longDefiningEquals |\{e \in S \mid \wt(e) < 1\}|$ \label{line:est-alg:err1}
                \If{$\derr(S) \geq \Delta / 2^i$}
                    \State $\wt(S) \gets \min(1, \wt(S) + 2^j / f)$
                \EndIf
            \EndFor
        \EndFor
        \EndFor
        \State Increase $\wt(S)$ by a factor $4$.
        \For{every $e$ with $\wt(e) < 1$ (in parallel)}
            \State set $\wt(S) \gets 1$ for an arbitrary $S \ni e$. \label{line:naive-covering-first}
        \EndFor
        \State \Return $\{\wt(S) \mid S \in \mathcal{S}\}$
    \end{algorithmic}
    \caption{In this section, we ask what properties the estimate in \cref{line:est-alg:err1} should satisfy so that we can still obtain $O(\log \Delta)$ approximate set cover.}
    \label{alg:err1}
\end{algorithm}

\cref{alg:err1} is just a generic form of the previous base algorithm where the estimation details are not specified. 
The estimation process in \cref{alg:frac:setcover:base} was the following: A set is considered dense during iteration $(i, j)$ if and only if it was considered dense during all iterations $(i, k)$ for $k < j$ and a significant portion of a sufficiently large random sample of elements in it are uncovered.
The purpose of this section is to identify characteristics of the estimation of $d(S)$ such that \cref{alg:err1} guarantees $O(\log \Delta)$ approximate set cover (in expectation). 

Identifying such characteristics will give us insight into designing $\LCA$s. 
The hope is that the approximation ratio arguments still hold even if the estimates are not so concentrated.
Then, we will look into obtaining these estimates by fewer random samples and design $\LCA$s that are more efficient than those in \cref{sec:warmup}. 

\subsection{First attempt -- conditions leading efficient LCA}

Observe that the approximation ratio argument for the base algorithm is as follows: Let $E^{(i)}$ denote the set of uncovered elements at the beginning of phase $i$. It holds that the maximum degree of a set $S$ is at most $\Delta /2^{i-1}$. During phase $i$, weight is added to the sets in such a way that, the total amount is at most $O(\opt)$. At the end of phase $i$, the goal is that after removing the covered elements, the size of every set drops to $\Delta  /2^i$. 
In this section, we shall show that the following relaxation to these properties still retains $O(\log \Delta)$ approximation. 

\paragraph{Dense sets need only a constant probability of failure.} Suppose we have the guarantee that for every dense set, its weight becomes halved with probability \textit{at least} $3/4$. That is, instead of $d(S)$ dropping deterministically by a factor of $1/2$, it drops with a constant probability.
We desire the probability to be \textit{strictly} greater than $1/2$, for e.g., we took $3/4$. In this case, the loss in the approximation ratio can be amortized to a constant via the geometric progression, for which we outline an intuition below.

If this guarantee holds \textit{independently} for each phase, then we are still okay in terms of approximation because the probability that a set $S$ has size $\Delta / 2^{i - i'}$ would be at most $4^{-i'}$. Intuitively we can think of the set $S$ ``paying'' approximation ratio $2^{i'}$, but only with probability $4^{-i'}$. 
The average approximation then amounts to only $\sum\limits_{i'} 2^{i'} \cdot 4^{-i'} \leq O(1)$. 
This is not a formal argument but a rough intuition as to why we may hope to identify dense sets only with constant probability.
% \stodo{Get back to this and emphasize if there is a challenge ensuring independence.}

\paragraph{Light sets need only $O(d(S) / \Delta \cdot 4^{-j})$ probability of failure.} 
At iteration $j$, our algorithm adds a weight of $2^j / f$ to the sets identified to be dense. 
When $j = \log f$, it adds weight $1$ to essentially \textit{every} set identified as dense. 
Hence, we cannot afford to make mistakes, e.g., marking a light set as dense, with probability $\Omega(1)$ when $j = \log f$.

Below, we outline why we can tolerate if a light set is mistakenly marked as dense with probability $O(d(S) / \Delta \cdot 4^{-j})$.
Suppose that during iteration $(i, j)$ we have a guarantee that a light set $S$, i.e., $d(S) < \Delta / 2^j$, is estimated to be light with probability $1 -  O(d(S) / \Delta \cdot 4^{-j})$. 
This implies that with probability $O(d(S) / \Delta \cdot 4^{-j})$, a weight of $2^j / f$ will be added to $S$.
Hence, the total cost of all sets mistakenly marked as dense during iteration $(i, j)$ equals, in expectation,
\begin{equation}\label{eq:cost-of-failed-sets}
    \sum_{S} O(d(S) / \Delta \cdot 4^{-j}) \cdot 2^j / f = \sum_{S} O(d(S) / \Delta) \cdot 2^{-j} / f.
\end{equation}
On the other hand, the total number of elements in the set cover instance is at most $\opt \cdot \Delta$. 
Also, over all the light sets, the total number of edges incident to these sets is $\sum_S d(S) \le \opt \cdot \Delta \cdot f$. This further implies that $\sum_S d(S) / \Delta \cdot 1 / f \le \opt$.
Combining this with \cref{eq:cost-of-failed-sets}, we derive that in iteration $(i, j)$ the total cost of sets mistakenly marked as dense is in expectation $O(\OPT \cdot 2^{-j})$. 
Over all iterations in a particular phase, the total cost sums to $O(\opt)$.

\paragraph{Constant size sets can be covered arbitrarily.} 
Finally, suppose we have left some elements uncovered after the $\log \Delta$ iterations, but we have the guarantee that every set has at most $O(1)$ of those uncovered elements. 
This is a trivial instance because the following greedy solution is an $O(1)$ approximate: for each element $e$, add an arbitrary set containing $e$ to the cover. 
Hence, if we can guarantee that, in expectation, the size of every set will be reduced to a constant, then we are done.
We can fix an optimal cover and perform linearity of expectation over this cover to upper-bound the expected cost.

\begin{theorem}\label{thm:approx-promise}
    \cref{alg:err1} outputs an expected $O(\log \Delta)$ approximate set cover if the estimates $\derr(S)$ and $\wterr(e)$, in every iteration $(i, j)$ for every set $S$ and element $e$, are computed such that the following equations hold and that for the first two equations, they hold independently of the randomness in previous iterations:

    Let $d_{i, j}(S)$ denote the number of uncovered elements $e \in S$ at the beginning of iteration $(i, j)$.
    \begin{align}
            \Pr\Big[\derr(S) \geq \Delta / 2^i \mid d_{i, j}(S) < \Delta / 2^{i+1} \Big] &\leq \frac{d_{i,j}(S)}{\Delta} \cdot \frac{2^{i-1}}{4^j} \label{cond:1} \\ 
            \Pr\Big[\derr(S) < \Delta / 2^i \mid d_{i,j}(S) \geq \Delta / 2^{i-1} \Big] &\leq \frac{1}{4} \label{cond:2} \\ 
            % \Pr\Big[\wterr(e) \geq 2^k \wt(e) \Big] &\leq 1/4^k \\
            \Pr[|\{e' \in S \mid \wt(e') < 1/4 \}| \geq k \cdot \Delta / 2^i] &\leq 1/k^3 \label{cond:3}
    \end{align}
\end{theorem}

\def\slight{\text{Cost}_{\text{light}}}

\begin{proof}
    Note that after \cref*{line:naive-covering-first} of \cref{alg:err1}, we guarantee that the weights assigned form a fractional cover. 
    It remains only to prove the (expected) approximation ratio. 
    
    We first show that the total weight assigned before \cref*{line:naive-covering-first} is $O(\log \Delta \cdot \OPT)$. 
    We only need guarantees \cref{cond:1,cond:2} for this. 
    Then, using \cref{cond:3}, we will upper-bound the total weight assigned in \cref*{line:naive-covering-first}.
    Let the input set cover instance be $\langle E, \mathcal{S} \rangle$.
    We show that for any $i$, the expected total weight added to sets during phase $i$ is $O(\opt)$.

    We call a set \textit{dense} if $d(S) \geq \Delta / 2^{i-1}$, \textit{moderate} if $d(S) \in [\Delta/2^{i+1}, \Delta / 2^{i-1})$ and \textit{light} otherwise.
    
    Note that the algorithm does not keep track of whether sets are dense, moderate or light. These definitions are only for analysis. 

    \textbf{Notation}. We begin by defining some notation that will be helpful in analysis. 
    
    \begin{itemize}
    \item Let $d_{i,j}(S)$ denote the value of $d(S)$ in the beginning of iteration $(i, j)$. 
    
    \item Let $\cov_{i,j}$ be the set of those elements whose weight is at least $1/4$ at the end of $(i, j)$. Intuitively, we think of $\cov_{i,j}$ as the set of elements that are, within a factor or $4$, already covered by iteration $(i, j)$.

    \item Let $\wt_{i,j}(S)$ denote the weight added to a set $S$ added \textit{during} iteration $(i,j)$.
    \item Let $\hat{\wt}_{i,j}(e)$ denote the estimated weight added to $e$ \textit{during} iteration $(i,j)$.
    \item Let $\Delta_i = \Delta / 2^{i-1}$. 
    \item Let $E_{i,j}= E \setminus \cov_{i, j}$ be the set of elements that are uncovered after iteration $(i, j)$.
    \end{itemize}
    
    In the perfect version of the algorithm, all sets $S$ satisfy $d_{i, j}(S) \leq \Delta_i = \Delta / 2^{i-1}$. However, this is not true for the imperfect version. 
    It is helpful to partition $E_{i, j}$ into subsets as follows,
    \begin{definition}
    \label{defn:uncov-partition}
    We partition the set of uncovered elements $E_{i, j}$ into $\log \Delta$ sets $E^{(k)}_{i, j}$ for $k = 1,2, \dots \log \Delta$ defined as follows,
    \begin{equation}
        E^{(k)}_{i, j} = \left\{e \in E_{i,j} \mid \max\limits_{S \ni e} d(S) \in \left(\Delta / 2^{k}, \Delta / 2^{k-1} \right]\right\}
    \end{equation} 
    We also write $E^{(\geq t)}_{i, j} = \bigcup\limits_{k\geq t} E^{(k)}_{i, j}$.
    \end{definition}
    \begin{observation}
    For a fixed $k$, and increasing iteration $(i, j)$, elements can never enter $E^{(\geq k)}_{i, j}$. Formally, for any two iterations $(i', j') \geq (i, j)$,
    \begin{equation}
        \label{eqn:part-subset-obs}
    E^{(\geq k)}_{i', j'} \subseteq E^{(\geq k)}_{i,j}
    \end{equation}
    \end{observation}
    \begin{proof}
        Suppose that an element $e \in E^{(\geq k)}_{i', j'}$. This implies that there exists a set $S \ni e$ with $d(S, i', j') \geq \Delta / 2^{k}$. The effective degree of all sets only decreases. Hence, $d(S, i, j) \geq d(S, i', j') \geq \Delta / 2^{k}$. It follows that $e \in E^{(\geq k)}_{i, j}$.
    \end{proof}

    We now prove several properties that we will later combine into a proof of this theorem.
    \begin{lemma}
    \label{lem:bound-e}
    Recall that $E^{(k)}_{i, j}$ denotes the set of uncovered elements at the beginning of iteration $(i, j)$ that are adjacent to sets with effective degree at most $\Delta / 2^{k-1}$ and contain a set with effective degree at least $\Delta / 2^k$. We derive the following bounds on the expected size of these sets:
        \begin{align*}
            |E^{(\geq i)}_{i, j}| &\leq \opt \cdot \Delta_i \\
            \xpect{|E^{(k)}_{i, j}|} &\leq \opt \cdot \Delta_i \cdot \left( \frac{1}{2} \right)^{i-k-1} \ \ \ \ \ \text{ for any } k < i 
        \end{align*}
    \end{lemma}
    \begin{proof}
        Every element in $E^{(\geq i)}_{i, j}$ can only be covered by sets of size at most $\Delta / 2^{i-1}$ in $\mathcal{S}$, and hence
        \[
            \left|E^{(\geq i)}_{i, j}\right| \leq \opt \cdot \Delta / 2^{i-1} = \opt \cdot \Delta_i.
        \]

        Since $d_{i,j}(S)$ does not increase with $i$, every element $e \in E^{(k)}_{i,j}$ for $k < i$ would have been adjacent to a dense or moderate set during phases $[k, i-1]$. 
        During these $i-k$ phases, these elements were not added to $\cov$. 
        Element $e$ is added to $\cov$ if the density of set $S$ was measured correctly at the end of one of the phases. 
        Since, by our assumption, the estimations of $d(S)$ in different phases are \textit{independent}, the probability that $e$ was not added to $\cov$ is at most $(1/4)^{i-k}$.

        \begin{observation}
            \label{obs:bound-e-exact-k}
            It holds that, $E^{(k)}_{i, j} \subseteq \bigcup\limits_{r = 1}^{k} E^{(\geq r)}_{r,1}$
        \end{observation}
        \begin{proof}
            Consider an element $e \in E^{(k)}_{i,j}$. There exists a set $S \ni e$ such that $d_{i,j}(S) \geq \Delta / 2^k$. Since $d(S)$ can only decrease, $d_{k, j}(S) \geq \Delta / 2^k$. 
            
            We would like to say that $e \in E^{(k)}_{k,j}$, but this is not true. It could be that during $(k, j)$ there exists some set $S'$ (possible same as $S$) with $d_{k,j}(S') > \Delta / 2^{k-1}$, say $d(S') \in (\Delta / 2^{k'}, \Delta / 2^{k'-1}]$. More formally, the guarantee we have is that $e \in E^{(\geq k)}_{k, j}$ and if $e \not\in E^{(k)}_{k, j}$, then there exists $k' < k$ such that $e \in E^{(\leq k')}_{k', j}$. It follows by induction that there exists $r < i$ such that $e \in E^{(r)}_{r, j}$.
        \end{proof}

        Applying the union bound on \cref{obs:bound-e-exact-k} along with linearity of expectation over every $e \in E_{r, j}^{(\geq r)}$ we get,
        $$\xpect{|E^{(k)}_{i, j}|} \leq \sum\limits_{r \leq k} |E^{(\geq r)}_{r, 1}| \cdot 4^{r-i} \leq \sum\limits_{r \leq k} \Delta / 2^{r-1} \cdot 4^{r-i} = \Delta / 2^{i-1} \cdot \sum\limits_{r \leq k}  2^{r-i} \leq \Delta_i \cdot 2^{k-i+1}$$ 
    \end{proof}

    We now analyze the expected cost assigned to the light-weight sets.
    \begin{lemma}[Cost of light-weight sets]
    \label{lem:light-weight}
        It holds that
        \[
            \sum\limits_S \xpect{\wt_{i, j}(S) \cdot \mathbbm{1}[S \text{ is light in iteration $(i, j)$}]} \leq 3 \cdot \opt / 2^j.
        \]
    \end{lemma}
    \begin{proof}
        Let $S$ be a set.
        For the sake of analysis, we pretend that $\wt_{i, j}(S)$ is either $0$ or $2^j / f$ depending on whether $\hat{d}(S) < \Delta / 2^i$ or not. Observe that actually $\wt_{i, j}(S)$ could be smaller than $2^j / f$, because $\wt(S)$ is capped to $1$. In our analysis, we pay $2^j / f$ even in such cases.
        For a light set $S$, \cref{cond:1} implies that $2^j / f$ is added to $\wt(S)$ with probability at most $d(S) / \Delta_i \cdot \frac{1}{4^j}$. 
        Hence, we have
        \[
            \sum_S \xpect{\wt_{i, j}(S) \cdot \mathbbm{1}[S \text{ is light in iteration $(i, j)$}]} \leq \frac{1}{\Delta_i f 2^j} \sum_S \xpect{d_{i, j}(S)}. 
        \]
        Now we observe that $\sum_S {d_{i, j}(S)}$ is the total number of edges in the bipartite graph of the instance $\langle E \setminus \cov_{i, j}, \mathcal{S} \rangle$ and is upper bounded by $f \cdot \sum_k |E^{(k)}_{i, j}|$. 
        Substituting this along with \cref{lem:bound-e} we get
        \[
            \sum_S \xpect{\wt_{i, j}(S) \cdot \mathbbm{1}[S \text{ is light in iteration $(i, j)$}]} \leq 3 \opt / 2^j.
        \]
    \end{proof}

    To upper-bound the weight of the moderate/dense sets, we follow the proof of \cref{alg:frac:setcover:base}. 
    Define $\wt_i(S) = \sum\limits_j \wt_{i,j} (S)$, and let $\wt_{i}(e) $ denote the total weight of element $e$ that comes from \textbf{only} dense/moderate sets during phase $i$ until the last iteration $(i, j)$ during which $\wt(e) < 1$. 
    We stress that $\wt_i(e)$ does \textbf{not} count the total weight of $e$, rather only a portion of it that is added during iterations $(i, 1)$ through $(i, \log f)$ and only that portion that comes from dense or moderate sets. 
    We have that
    \[
        \sum_{e} \wt_{i}(e) = \sum_{j, \text{not light } S} d_{i,j}(S) \cdot \wt_{i, j}(S) \geq \frac{\Delta_i}{2} \sum_{\text{not light }S} \wt_{i}(S).
    \]
    Now obtaining an upper bound on $\xpect{\wt_i(e)}$, suffices to upper bound the expected value of the RHS. We defer the proof of the upper bound to \cref{sec:proof-of-lem:xp-wt-e}, and just state it as a lemma here.

    \textbf{Note:} The following lemma is the only difference between the fractional and integral versions. 
    
    \begin{lemma}[Element weight from non-light sets]
    \label{lem:xp-wt-e}
        It holds that $\xpect{\wt_{i}(e)} \leq 64$.
    \end{lemma}

    Using \cref{lem:xp-wt-e} we get that,

    \begin{lemma}[Cost of moderate and dense sets]
        \label{lem:heavy-weight}
    \[
        \xpect{\wt_i(S) \cdot \mathbbm{1}[S \text{ is moderate or dense }]} \leq 64 \cdot \xpect{|E_{i,1}|}  / \Delta_i.
    \]
    \end{lemma}
    We use \cref{lem:bound-e} to derive that
    \[
        \xpect{|E_{i, 1}|} \leq \sum_k \xpect{|E^{(k)}_{i, 1}|} \leq 3 \cdot \Delta_i \cdot \opt,
    \]    
    which gives $\xpect{\wt_i(S) \cdot \mathbbm{1}[S \text{ is dense }]} \leq 192 \cdot \opt$.
    Adding the upper bounds on the weights of light and dense sets, we conclude that $\xpect{\wt_i(S)} \leq 198 \cdot \opt$.

    Finally, we need to upper-bound the weight of sets in \cref*{line:naive-covering-first}. For this, we apply \cref{cond:3} on any optimal set cover for the instance. We show that on expectation, the total number of elements in each set that has $\wt(e) > 1/4$ is at most $O(1)$. Applying linearity of expectation over these sets completes the proof.
    \begin{lemma}{(Cost of naive-covering).}
    \label{lem:cost-naive-covering}
        For every set $S$, just before Line 10 of \cref{alg:err1} we have, 
        $$\xpect{|e \in S \mid \wt(e) < 1/4|} \leq O(1)$$
    \end{lemma}
    \begin{proof}
        Suppose there are $k$ elements $e \in S$ with $\wt(e) < 1/4$ at the end of $(\log \Delta, \log f)$ iteration of \cref{alg:err1}. Applying \cref{cond:3}, the probability that this happens is at most $1/k^3$. Therefore, the expected number of such elements is at most $\sum\limits_{k=1}^\Delta k \cdot 1/k^3 \leq \pi^2 / 6$.
    \end{proof}
    Combining \cref{lem:light-weight,lem:heavy-weight,lem:cost-naive-covering} we complete the proof of \cref{thm:approx-promise}.
\end{proof}

\subsection{Element weight from non-light sets -- Proof of \cref{lem:xp-wt-e}}
\label{sec:proof-of-lem:xp-wt-e}
    If $e$ is not covered, then the claim follows trivially. If $e$ was covered before phase $i$, then $\wt_i(e) = 0$. Let $(i, j)$ denote the first iteration during which $\wt(e) \geq 1$. 
    
    \textbf{Proof Intuition.} Let $n_t$ denote the number of dense sets that contain $e$ at the beginning of iteration $(i, t)$. If $n_j = O(f / 2^j)$, then we are done. On the other hand if $n_j > (f / 2^j) \cdot 2^{q}$ for some integer $q$, then it is highly unlikely that $e$ remained uncovered during iterations $(i, j-q)$ through $(i, j-1)$. So, we can pay for the extra weight assigned by the small probability that the previous events happened. We now make this argument more formal.

    Let $F(t, n_t)$ denote the maximum value of the expected weight added to an element $e$ starting from iteration $t$, until at least $1$ unit of weight is added to it given that at least $n_t$ sets where dense at the beginning of iteration $(i, t)$. We need to show $F(1, f) \leq O(1)$. We shall show that $F(t, n_t) \leq O(n_t 2^t / f)$ by induction on decreasing $t$. The base case is that $F(\log f, n_{\log f}) \leq f$ which is immediate since at most $f$ sets contain $e$ and the weight added is at most $1$.
    Now suppose that $F(t+1, n_{t+1}) \leq \max\{D, C \cdot n_{t+1} \cdot 2^{t+1} / f \}$. for large enough constants $C, D$ (to be fixed to complete the induction). We have $F(t, n_t) \leq \max_{n_{t+1} \leq n_t} \Pr[e \text{ is covered at } (i, t)] \cdot n_t 2^t / f + \Pr[e \text{ is not covered at } (i, t)] \cdot F(t+1, n_{t+1})$.
    We consider two cases. 

    Case (i). $n_t \leq f / 2^{t-4}$. In this case we have $F(t, n_t) \leq F(t+1, n_{t}) \leq \max\{D, 32C\}$. 
    
    Case (ii) $n_t > f / 2^{t-4}$. 
    In this case, we reason about the probability that $e$ is not covered using Markov's Inequality.

    \begin{claim}
     $\Pr[e \text{ is not covered at } (i, t)] \leq 2/7$ given that $n_t > 16 f / 2^t$.
     \end{claim}
     \begin{proof}
        There are $n_t$ dense sets containing $e$. By \cref{cond:1}, each of them is missed with probability at most $1/4$. The expected number of dense sets with no added weight is at most $n_t / 4$. By Markov's inequality, the probability that at least $7 n_t / 8$ are missed is at most $2 / 7$. Otherwise, at least $f / 2^t$ sets containing $e$ have weight $2^t / f$ added to them, and so $e$ would be covered at $(i, t)$.
     \end{proof}

     Now applying the above we can derive,
     \begin{align*}
     F(t, n_t) &\leq n_t 2^t / f + 2/7 \cdot F(t + 1, n_{t+1}) \\ 
     &\leq n_t \cdot 2^t / f + 4 n_t/7 \cdot C \cdot 2^{t} / f  \\
     &\leq C \cdot n_t \cdot 2^t / f
     \end{align*}
     where the last inequality holds provided $C \geq 7/3$ and that $n_t 2^t / f > 16$. Thus, to complete the induction, it suffices to set $D = 64$ and $C = 7/3$. We get $F(1, n_t) \leq \max\{64, 7 n_t / 3f\} \leq 64$.

\subsection{Second attempt -- retroactive updates}

\cref{thm:approx-promise} is promising because it shows that not every phase needs to be simulated correctly with high probability, i.e., with probability $1 - O(1 / \poly(\Delta))$. 
However, obtaining estimates of the effective degree of the sets per this theorem is non-trivial. In particular, the obvious approach of simply reducing the number of samples does not (directly) yield \cref{cond:1,cond:2,cond:3}.
Let $d(i, j, S)$ denote the number of uncovered elements in $S$ at the beginning of iteration $(i, j)$.

Let us begin with the natural approach to test whether $d(i, j, S) \geq \Delta / 2^i$: Sample $2^i \cdot K$ elements $e \in S$. We shall then test if at least $K / 2$ of them are uncovered. 
Our test for whether an element is covered is also error-prone in both directions, i.e., some element $e$ could be covered, but the test returns that it is not. Conversely, an element might not be covered, but the test might return that it is. 

The tests fail with probability that is exponential in the number of samples and hence doubly exponential in $i$.
Thus, we can achieve \cref{cond:2} and \cref{cond:3}. 
However, \cref{cond:1} is trickier since, for instance, there might be sets $S$ with $d(S) = 0$. 
In this case, \cref{cond:1} asks for a perfect estimate, which might not be feasible since the LCA might report an element to be uncovered incorrectly. 

In \cref{alg:Is-Set-Dense}, we expected that for the random subset of $S$ computed on \cref*{line:dense-est-sample}, it holds that at most $2^{-i'}$ fraction of them remain uncovered after $i'$ iterations. 
Moreover, this fact was true with a high enough probability that we can afford to naively cover the elements belonging to sets that did not satisfy this property. 
This time, we can no longer afford that luxury. 
Instead, we resort to the following. 

Before simulating phase $i$, we \textbf{re-execute} phases $i'$ for each $i' < i$ with fresh randomness $2^{O(i-i')}$ times. 
This ensures that the guarantees about phase $i'$ now hold with probability $1 - 2^{-\Omega(i - i')}$, which is small enough to argue about the approximation ratio. 
Even though this re-execution increases the LCA complexity, the probe complexity still satisfies the recurrence relation $T(i) = \sum\limits_{i' < i} 2^{O(i - i')} \cdot T(i')$ which yields our desired complexity.
These re-executions are used to \textbf{retroactively} updates set weights. We now elaborate on that.

\begin{algorithm}
    \begin{algorithmic}[1]
        \State Set $\wt(S), \derr(i, j, S) $ to be $0$ for every set $S$ and $(i, j) \in [\log \Delta] \times [\log f]$
        \Statex \Comment{Initialize the fractional cover and the estimates of effective degrees of sets}
        \Statex \Comment{Recall that $d(S) = \{e \in S \mid \wt(e) < 1\}$ and $d(i, j, S)$ is the value of $d(S)$ at iteration $(i, j)$}

                \For{every set $S$ (in parallel)}
        \For{$(i, j) = (1, 1) \cdots (\log \Delta, \log f)$} \Comment{Iterate in row major order}
            
            \For{$(i',j')=(1, 1) \cdots (i, j)$\label{line:retroactive-update-for-loop}}
                    \State $\derr_{new}(i', j', S) \gets \Call{degreeEstimate}{i', j', i, j, S}$ \Comment{$\LCA$ oracle given as \cref{alg:GetDeg}}
                    \State Set $\derr(i', j', S) \gets \max(\derr(i', j', S), \min\limits_{(ii, jj) \leq (i, j)} \derr_{new}(ii, jj, S))$ \label{line:alg:upd}
                    \If{$\derr(i',j',S)$ exceeded $\Delta / 2^{i'}$ for the first time after the previous line (\cref*{line:alg:upd}) \label{line:if-retroactive-needed}}  
                        \State Add $2^{j'} / f$ to $\wt(S)$ \label{line:retroactive-set-add-weight}
                        % \stodo{You would add $2^{j'} / f$ at most once, right? If yes, it would be good to somehow emphasize that.}
                    \EndIf
                    \If{$\wt(S) \geq 1$}
                        \State Set $\wt(S) \gets 1$ and break out of both loops
                    \EndIf
                \EndFor
            \EndFor  
        \EndFor 
    
        \State Increase $\wt(S)$ for every $S$ by a factor $4$. \footnote{This line has more to do with our analysis, and setting of parameters $\delta, K$ for the oracles $\GetDeg{},\GetWt{}$. Conceptually, this line is not needed.}
        \For{every $e$ with $\wt(e) < 1$ (in parallel)}
            \State set $\wt(S) \gets 1$ for an arbitrary $S \ni e$. \label{line:naive-covering}
        \EndFor
        \State \Return $\{\wt(S) \mid S \in \mathcal{S}\}$
    \end{algorithmic}
    \caption{A $\LOCAL$ algorithm with retroactive updates}
    \label{alg:err2}
\end{algorithm}
Consider \cref{alg:err2}, which is designed for $\LOCAL$. 
Analogous to the algorithm in the previous section, $\GetDeg{}$ is an oracle that estimates the degree of a set $S$. We provide an LCA implementation of that oracle as \cref{alg:GetDeg}.

Say that in iteration $(i, j)$ the algorithm \textbf{does not} add $2^j / f$ weight to $\wt(S)$ in \cref*{line:retroactive-set-add-weight}.
There are several reasons why this weight was not added: $S$ was light to begin with; or $S$ was heavy, but its size estimate obtained via $\GetDeg{}$ was inaccurate and classified $S$ as light.
The latter reason essentially means that $S$ has more uncovered elements than our estimate; for the sake of explanation, call such set $S$ too heavy.
On the one hand, in iteration $(i, j)$, we would like to query as few elements still uncovered in iteration $(i, j)$ as possible. Intuitively, querying an element $e$ uncovered in iteration $(i, j)$ is expensive, as the LCA would test whether $e$ was covered in each of the iterations preceding $(i, j)$.
On the other hand, when the weight of a too-heavy set $S$ is not increased, potentially many elements remain uncovered while they are supposed to get covered. This directly increases the LCA complexity.

To alleviate that, when we execute iteration $(i, j)$ for set $S$ in \cref{alg:err2}, we also recalculate our size/degree estimates for $S$ in iterations $(1, 1)$ through $(i, j-1)$; this is done within the loop on \cref*{line:retroactive-update-for-loop}.
Crucially, a recalculation for iteration $(i', j') = (1, 1)$ done during iteration $(i, j) = (10, 1)$ uses a lot more samples than the calculation used to perform whether $S$ was heavy during iteration $(i, j) = (1, 1)$.
This enables us to increase the confidence of our estimators.
Suppose it turns out that our estimator was too inaccurate in iteration $(i', j')$. In that case, we retroactively increase $\wt(S)$ by what it ``should have been'' increased in iteration $(i', j')$; this is implemented by \cref*{line:if-retroactive-needed,line:retroactive-set-add-weight}.

\begin{theorem} \label{thm:eqn-to-approx}
    \cref{alg:err2} returns an expected $O(\log \Delta)$ approximate fractional cover, if the following equations hold \textbf{independently} of the randomness in iterations prior to $(i, j)$. 

    Let $\derr_{new}(i', j', i, j, S)$ be the estimated value of $d(i', j', S)$ estimated at the beginning of iteration $(i, j) (\geq (i', j'))$. Similarly, let $\wt(i', j', i, j,e)$ be the weight of $e$ at the end of iteration $(i', j')$, including retroactive updates made up to iteration $(i, j)$. Let $b$ denote the number of iterations in phase $i$ that are executed after $(i', j')$ and before $(i, j)$, i.e., $b = j - j' + 1$ if $i = i'$ and $j$ otherwise.
    \begin{align}
        \Pr\left[\derr_{new}(i', j', i, j, S) < \Delta / 2^{i'} \mid d(i', j', S) \geq \Delta / 2^{i'-1} \right] &\leq \frac{1}{8} \cdot \frac{1}{2^{i-i'}}\label{eqn:deg:under-estimate} \\ 
        \Pr\left[ \derr_{new}( i', j', i, j, S) \geq \Delta / 2^{i-1} \mid d(i', j', S) < \Delta / 2^{i+1} 
        \right] &\leq \frac{1}{8^{b}} \cdot \frac{1}{2^{i-i'}} \label{eqn:deg:over-estimate} \\
        \forall t \ \geq 1, \ \ \Pr\left[ |\{ e' \in S \mid \wt(i',\log f, i, 1, e') < 1 \}| \geq t \cdot \Delta / 2^{i'} \right] &\leq \frac{1}{t^3} \cdot \frac{1}{2^{i-i'+1}}\label{eqn:set:sparsification} 
    \end{align}
\end{theorem}

\def\uncov{\overline{E}}
\def\cost{\mathsf{cost}}
\begin{proof}
    Let $\cost_{i, i'}$ denote the expected weight assigned to sets as retroactive updates for phase $i'$ during phase $i$. We shall show that $\cost_{i, i'} = O(\opt / 2^{i-i'})$, from which it follows that the overall expected cost is $O(\log \Delta) \cdot \opt$. To show this, we pay for the added cost for phase $i'$ in phase $i$ with the probability that an error was made during the previous $i-i'$ phases.

    The proof strategy is similar to the previous algorithm. We repeat the process step by step, accounting for the retroactive updates in each. Consider some iteration $(i', j')$. Let $\uncov$ denote the set of uncovered elements in the algorithm at the beginning of $(i', j')$. Note that at this point, retroactive updates for phases $i'+1, i'+2, \dots i-1$ have not been added. 

    Let $\uncov^{(k)} = \{ e \in \uncov \mid \max\limits_{S \ni e} d(S) \in (\Delta / 2^{k}, \Delta / 2^{k-1}] \}$. We also write $\uncov^{(\geq t)} = \bigcup\limits_{k \geq t} \uncov^{(k)}$. 

    The proof of \cref{lem:bound-e-2} is analogous to \cref{lem:bound-e} in the previous section, except for the constants.

    \begin{lemma}
    \label{lem:bound-e-2}
        $\xpect{|\uncov^{(k)}|} \leq \opt \cdot \Delta_{i
        '} \cdot (15/16)^{i'-k+1}$ when $k < i'$ and $\xpect{|\uncov^{(\geq i')}} \leq \Delta_i \cdot \opt$.
    \end{lemma}

    \begin{lemma}
        $\sum\limits_{S} \sum\limits_{(i',j')} \xpect{\wt_{i',j'}(S) \cdot \mathbbm{1}[S \ \text{ is light }]} \leq O(\log \Delta \cdot \opt)$
    \end{lemma}

    \begin{proof}
        Suppose a set $S$ is determined to be $\textit{dense}$ during iteration $(i, j)$ but was actually light. We consider two cases. Let $(i', j')$ be the \textit{last} iteration during which $d(S, i', j') \geq \Delta / 2^i$. 
        \begin{itemize}
            \item[Case 1.] $d(S, i', j') \geq \Delta / 2^{i'}$. We have already paid the $2^{j'} / f$ weight for this term when accounting for the dense sets. The total weight of these sets can thus be bound by,
             $$\sum\limits_{S} \sum\limits_{(i',j')} \sum\limits_{(i,j)\geq(i',j')}\mathbbm{1}[d_{i',j'}(S) \geq \Delta / 2^{i'}] \cdot \frac{1}{8^{j-j'}} \cdot 2^{j - j'} \cdot \frac{1}{2^{i-i'}}$$ 
             The summands with terms $(i,j)$ and $(i',j')$ all add to only a constant. Hence, the overall cost is $O(\log \Delta \cdot \opt)$.
            \item[Case 2.] $d(S, i', j') < \Delta / 2^{i'}$. In this case the failure probability is $\frac{1}{8^{j-j'+1}} \cdot \frac{1}{2^{i-i'}} \leq \frac{d(S,i',j')}{\Delta_{i'}} \cdot \frac{1}{8^{j-j'+1}}$. Observe that this term is the same as the bound for the light sets in the previous section. Using that argument, we can bound the weight added to these sets.
        \end{itemize}
    \end{proof}

    Next, we move to the dense sets. We follow the same proof for the weight added when $i = i'$. 
    When $i' < i$, we pay for the weight added retroactively at phase $i$ by the probability that it was missed during phase $i-1$. Note that the argument in this case was: $\Delta_i \cdot \sum \wt_{i, j}(S) \leq \sum \wt(e)$. We then showed that $\xpect{\wt(e)} \leq O(1)$ and then the result follows by the fact that, in expectation, the total number of elements must be $\Delta_i \cdot \opt$. 

    To modify this argument for the new algorithm, we need to account for the retroactive updates. For an element $e$, suppose it was first covered (i.e., $\wt(e)$ exceeded $1$) during iteration $(i, j)$. We define $X_e$ to be the total weight added to sets containing $e$ from iteration $(i', 1)$ through $(i', j)$ \textbf{including} the retro-active updates to be added from phase $i'+1$ through $\log \Delta$.

    Suppose we prove that $\xpect{X_e} \leq O(1)$, then we are done as we can apply the same argument.

    \begin{lemma}
        It holds that $\xpect{X_e} \leq O(1)$.
    \end{lemma}

    \begin{proof}
        Suppose during iteration $(i', j')$ there were $k$ dense sets containing $e$. The total weight added to $e$ during this iteration is at most $k 2^{j'} / f$. The actual number might be smaller due to the mistakes in estimating the density of some sets. 
        The missed weight might be added as retroactive updates in later phases by \cref{alg:err2}. 
        
        In the previous arguments, we anyway pay $k 2^{j'} / f$, and since the retroactive updates are carried only once, the same arguments follow.
    \end{proof}
\end{proof}

\section{Designing the main \texorpdfstring{$\LCA$}{LCA}}
\label{sec:main-alg}

We now describe and analyze the query complexity of the $\LCA$s. Analogous to the previous algorithm, we design two oracles, $\GetWt{}$ and $\GetDeg{}$ in \cref{alg:GetWt,alg:GetDeg} that estimate the weights of elements and effective degree of sets at some given iteration $(i, j)$ after the retro-active updates until phase $i^{\star}$. These oracles recursively call each other. To obtain an $\LCA$ for the weight of a set, one should simulate \cref{alg:err2} using the oracle (\cref{alg:GetDeg}) designed in this section.
They are designed such that they satisfy \cref{eqn:deg:over-estimate}, \cref{eqn:deg:under-estimate} and \cref{eqn:set:sparsification}. We now explain the algorithms in more detail.

\subsection{Sampling strategy and tracking randomness}

In this subsection, we briefly overview our sampling procedure and how we shall aim to achieve \cref{eqn:deg:under-estimate} and \cref{eqn:deg:over-estimate}. During an iteration $(\stari, \starj)$, to determine if $d(S, \stari, \starj) \geq \Delta / 2^{\stari}$, we make $2^{\stari} \cdot K$ random samples to $S$. Based on the pruning strategy of \cref{sec:warmup}, we should first remove the elements $e \in S$ whose estimated weight after phase $1$ is large. If we simply use the $\LCA$ oracles designed for estimating phase $1$, our estimate would not be accurate enough. Firstly, the success probability is not large enough for our guarantee. Secondly, our estimate would then not be independent of the randomness in previous phases.

To fix this, observe that for phase $i$, we are allowed up to $2^{O(i)}$ queries, thus we may re-execute the oracles for previous phases multiple times, i.e., phase $i'$ can be re-executed $2^{O(i- i')}$ times. We take advantage of this as follows.
Before beginning iteration $(\stari, \starj)$ we go over each previous iteration $(i, j)\leq (\stari, \starj)$ and strengthen the weight estimates of elements after this iteration. Then we use these strengthened estimates for the weights to compute strengthened estimates of the degrees.

\paragraph{Weight estimates} For an element $e$, we pay for any errors in estimating its weight, purely by the random sampling of sets $S \ni e$ that contain it. Our guarantees for the estimation of the weights of elements during iteration $(\stari, \starj)$ will depend \textit{purely} on the the local random bits allocated for $e$ during $(\stari, \starj)$.

\paragraph{Degree estimates} During iteration $(\stari, \starj)$, to estimate $d(S)$ for a set $S$, we first sample a subset of $2^{\stari} \cdot K$ elements. If these elements capture too many or not enough elements $e$ that were covered during previous iterations, then we say that the estimate failed. If the sampling is sufficient, then the error in estimating the weights of these elements is then based only on an independent sample of sets that contain it, which from the previous paragraph, is independent of any element samples. 
Thus the guarantees of the tests $d(S) \geq \Delta / 2^{\stari}$ for the different iterations $(\stari, \starj)$ are independent.

\subsection{High level overview of oracles}

Below we give an overview of our $\LCA$ oracles.

\paragraph{Oracle \textsc{degreeEstimate} (\cref{alg:GetDeg}).} This $\LCA$ oracle provides an estimate of the effective degree of set $S$ \textit{at the beginning of} iteration $(i, j)$. The oracle returns a probabilistic estimate of the degree of the set $S$. The guarantees about this estimate are tuned by another set of parameters $(\stari, \starj)$. The estimates are returned to the precision needed at iteration $(\stari, \starj)$. We shall later describe how the guarantees vary with each iteration. Intuitively, this guarantee strengthens with every passing iteration in such a way that the number of samples needed increases by at most a \textit{constant factor} with each iteration. 

\paragraph{Oracle \textsc{weightEstimate} (\cref{alg:GetWt}).} This $\LCA$ oracle provides an estimate of the weight of an element $e$ \textit{at the end of} iteration $(i, j)$. The oracle returns a probabilistic estimate similar to the $\textsc{degreeEstimate}$ oracle.

Oracles \textsc{degreeEstimate} and \textsc{weightEstimate} invoke each other. To estimate the degree of a set $S$ at some iteration, we sample elements of $S$ and invoke the \textsc{weightEstimate} oracle to estimate their weight at the end of the previous iteration. To estimate the degree of a set $S$ at the end of iteration $(i, j)$ after processing up to iteration $(\stari, \starj)$ of the algorithm, we sample $2^{\stari} \cdot K$ elements $e \in S$. This should be interpreted as strengthening the guarantees by a factor of $2^{\stari - i}$. We shall then prune out the elements $e \in S$ that are covered before iteration $(i, j)$ smoothly by beginning with the first iteration. If we see too many elements staying uncovered, then we give up on the estimate- this is a bad sampling, and our analysis should account for any such errors made. Once $i-1$ phases are simulated, we shall have a small enough sample of the set $S$ wherein we can simulate the $i^{th}$ phase. Yet again, we incrementally prune out the elements covered in the iterations $j' = 1, 2, \dots j-1$. The tests are performed such that the probability that an element is estimated to be covered in iteration $j' < j$ but \textit{not} in iteration $j$ is $2^{-\Omega(j - j')}$. This is achieved by boosting the number of samples in the estimate of $j'$ by a factor of $2^{j - j'}$.

Recall that in a given phase, weights are added to the elements in a geometric progression. We can also afford to pretend that most of the weight that covered the given element $e$ came during the last iteration. The \textsc{weightEstimate} oracle on an element $e$ estimates how many sets $S$ that contain $e$ were dense during a given iteration $(i, j)$. If we find that many sets were dense, then we can estimate that the weight added to that element during this iteration was high enough that it will be covered. Thus, we need to sample at least $2^j \cdot K$ such sets $S$ to obtain an estimate with failure probability $\Omega(1)$.

\begin{algorithm}
    \begin{algorithmic}[1]
        \Procedure{weightEstimate}{$i, j, i^{\star}, \starj, e$} 
            \Statex \Comment{Estimate the weight of $e$ at the end of iteration $(i, j)$ after the retro-active updates up to phase $i^{\star}$, with failure probability boosted up to $\starj$}
            \Statex \textbf{Input:} Phases $i, i^{\star} (\geq i)$, iteration $j$ and element $e$ with $0 \leq i \leq \log \Delta$ and $1 \leq j \leq \log f$.
            \Statex \textbf{Parameters}: $K, \delta$ are sufficiently large constants.
            \Statex \textbf{Base Case:} If $i = 1$ and $j = 0$, then return $0$
            \State If $\Call{weightEstimate}{e, i, j-1, \stari, \starj} \geq 1/2$ then \Return $1/2$
            \State $b \gets \#$ of iterations between $(i, j)$ and $(\stari, \starj)$
            \State $D \gets $ A sample of $2^{j + \delta \cdot b} \cdot K$ sets $S \ni e$ each chosen u.a.r with repetition. \label{line:ele-sample}
            \For{$j' = 1, 2, \dots j$}
                % \If{$|D| \geq 2^{j + b - j' + \delta} \cdot K$}
                \For{every set $S \in D$} \Comment{Guaranteed that $|D| < 2^{j - j'+ 1 + \delta \cdot b + \delta} \cdot K$}
                    \If{$\Call{degreeEstimate}{S, i, j', \stari, \starj} < \Delta / 2^i$ or returned $\mathsf{fail}$}
                        \State Remove $S$ from $D$
                    \EndIf 
                \EndFor 
                \If{$|D| \geq 2^{j - j' + \delta \cdot b + \delta} \cdot K$}
                    \State \Return $1/2$
                \EndIf 
                % \EndIf 
            \EndFor 
            \State \Return $0$
        \EndProcedure 
    \end{algorithmic}
    \caption{LCA oracle to estimate weight of an element.}
    \label{alg:GetWt}
\end{algorithm}

\begin{algorithm}
    \begin{algorithmic}[1]
        \Procedure{degreeEstimate}{$i, j, \stari, \starj, S$}
            \Statex \Comment{Estimate the effective degree of a set $S$ during iteration $(i, j)$ after the retro-active updates until phase $i^{\star}$ with probability boosted to the quality needed at $\starj$}
             \Statex \textbf{Input:} Phases $i, \stari (\geq i)$ and iteration $j$, set $S$ with $0 \leq i \leq \log \Delta$ and $1 \leq j \leq \log f$.
             \Statex \textbf{Parameters:} $K, \delta$ are sufficiently large constants.
             \Statex \textbf{Base Cases:} If $i = 0$, return $|S|$
            \State Set $b \gets \#$ of iterations between $(i, j)$ and $(\stari, \starj)$.
            \State If $\Call{degreeEstimate}{S, i, j-1, \stari, \starj}$ failed, then return $\mathsf{fail}$.
            \State $\hS \gets $ A multiset of $2^{i + \delta \cdot b} \cdot K$ elements from $S$ each chosen u.a.r. 
            \label{line:alg:degEst:sampling}
            
            \For{$i' = 1, 2, \dots i-1$}
                \State $b' \gets \#$ of iterations from $(i', \log f)$ to $(i, j)$
                \For{every distinct element $e \in \hS$ with $\Call{weightEstimate}{e, i', \log f, b'} \geq 1/2$}
                    \State Remove all occurrences of $e$ from $\hS$ \label{line:alg:degEst:pruning}
                \EndFor 
                \If{$|\hS| > 2^{i - i' + \delta \cdot b + \delta} \cdot K$} \label{line:bad-event-GetDeg}
                    \State \Return $\mathsf{fail}$ \Comment{Estimate failed}
                \EndIf 
            \EndFor 
            \For{$j' = 1, 2, \dots j-1$} \Comment{Guaranteed that $|\hS| \leq 2^{\delta \cdot b + \delta + 1} \cdot K$}
                \For{every element $e \in \hS$ with $\Call{weightEstimate}{e, i, j', j - j'} \geq 1/2$}
                    \State Remove $e$ from $\hS$ \label{line:alg:degEst:wtTest}
                \EndFor 
            \EndFor 
            \State \Return $|\hS| \cdot |S| / (2^{\delta \cdot b} \cdot K)$ \Comment{Re-scaling subsample to $|S|$}
        \EndProcedure 
    \end{algorithmic}
    \caption{LCA oracle to estimate the effective degree of a set.}
    \label{alg:GetDeg}
\end{algorithm}

\subsection{Query Complexity}

\begin{theorem}
\label{thm:alg:lca-frac:time}
    The probe complexity of the $\LCA$ that simulates \cref{alg:err2} is $2^{O(\log \Delta \log f)}$.
\end{theorem}

Let $S(i, j, b), E(i, j, b)$ denote the total number of recursive calls made during the invocations to \\ $\GetDeg{S, i, j, \stari, \starj}$ and $\GetWt{e, i, j, \stari, \starj}$ respectively where $b$ is the number of iterations between $(i, j)$ and $(\stari, \starj)$. Let $t_{i,j} = (i-1)\log f + j$, i.e., the index of iteration $(i, j)$. 

We prove by induction that $S(i, j, b) \leq C^{t_{i,j}} \cdot 2^{\delta \cdot b}$ and $E(i, j, b) \leq C^{t_{i,j} +  1/2} \cdot 2^{\delta \cdot b}$ for some large enough constant $C$.

For the base case, we have $S(1, 1, b) \leq C \cdot 2^{\delta \cdot b}$ and $E(1, 1, b) \leq  C \cdot 2^{\delta \cdot b}$. Choosing $C \geq 2^{\delta \cdot b}$ suffices to show the base case.
We next write recurrence relations for these functions from the pseudo-code of the algorithm.

\begin{align}
    S(i, j, b) &\leq S(i, j-1, b) + \sum\limits_{i'=1}^{i-1} E(i', \log f, t_{i,j} - t_{i', \log f}) \cdot 2^{i - i' + 1 + \delta \cdot b + \delta} \cdot K \nonumber \\
    &\ \ + \sum\limits_{j'=1}^{j-1} E(i, j', j - j') \cdot 2^{1 + \delta \cdot b + \delta} \cdot K \\
    E(i, j, b) &\leq E(i, j-1, b) + \sum\limits_{j'=1}^{j} S(i, j', j-j'+b) \cdot 2^{j - j' + 1 + \delta \cdot b + \delta} \cdot K
\end{align}

Using strong induction, we can simply plug-in our hypothesis on the right hand side as follows.
\begin{align}
    S(i, j, b) &\leq C^{t_{i,j-1}} \cdot 2^{\delta \cdot b} + \sum\limits_{i'=1}^{i-1} C^{i' \log f + 1/2} \cdot 2^{\delta \cdot b + \delta + (i-i'+1) + \delta \cdot (t_{i,j} - t_{i', \log f})} \cdot K + \\
    &\sum\limits_{j'=1}^{j-1} C^{t_{i,j'}} \cdot 2^{1 + \delta\cdot b + \delta + t_{i,j} - t_{i,j'}}\cdot K \\
    &\leq C^{t_{i,j}} \cdot 2^{\delta \cdot b} \cdot \Bigg\{ C^{-1} + \left( \sum\limits_{i'=1}^{i-1} C^{t_{i', \log f} - t_{i,j} + 1/2} \cdot 2^{\delta \cdot (t_{i,j} - t_{i', \log f} + 1)} \cdot K \right) + \\
    & \left( \sum\limits_{j'=1}^{j-1} C^{t_{i,j'} - t_{i,j}} \cdot 2^{\delta + 1 + t_{i,j} - t_{i,j'}} \cdot K \right) \Bigg\} \\
    &\leq C^{t_{i,j}} \cdot 2^b \ \ \ \ \ \text{[As long as } C \geq 4^{\delta + 3} \cdot K^2 \text{]}
\end{align}

Similarly we can bound $E(i, j, b)$.

\begin{align}
    E(i, j, b) &\leq E(i, j-1, b) + \sum\limits_{j'=1}^{j} S(i, j', j-j'+b) \cdot 2^{j - j' + 1 + \delta \cdot b + \delta} \cdot K \\
    &\leq C^{t_{i,j}+1/2} \cdot 2^{\delta \cdot b} \cdot \left( (C^{-1/2}) + \sum\limits_{j'=1}^j C^{t_{i,j'} - t_{i,j} - 1/2} \cdot 2^{j - j' + 1 + \delta} \cdot K \right) \\
    &\leq C^{t_{i,j} + 1/2} \cdot 2^{\delta \cdot b} \ \ \ \ \ \text{[As long as } C \geq 4^{\delta + 3} \cdot K^2 \text{]}
\end{align}

This completes the inductive proof on $S(i, j, b)$ and $E(i, j, b)$.
Finally, observe that to simulate \cref{alg:err2}, we need only $O(\Delta f)$ queries to these oracles and so the total amounts to $2^{O(\log \Delta \log f)}$.

\def\od{d^{\circ}}
\def\dplus{d^{+}}
\def\dmin{d^{-}}

\subsection{Approximation Ratio}

In this section we prove \cref{lem:main-approx-ratio}.

\begin{lemma}
\label{lem:main-approx-ratio}
    An LCA implementation of \cref{alg:err2} with the LCA oracles $\GetDeg{}$ given in \cref{alg:GetDeg} returns an LCA for the fractional set cover whose cost (in expectation) is $O(\log \Delta) \cdot \opt$ where $\opt$ is the optimal integral solution for the set cover instance.
\end{lemma}

Our main task is to show that \cref{alg:err2} when run with \cref{alg:GetDeg} satisfies \cref{eqn:deg:over-estimate,eqn:deg:under-estimate,eqn:set:sparsification} and then \cref{lem:main-approx-ratio} follows by \cref{thm:eqn-to-approx}.

We first define some terms. Let $\derr_{new}(i, j, \stari, \starj, S)$ denote the estimated effective degree of a set $S$ at iteration $(i, j)$ after retroactive updates until iteration phase $\stari$ measured with precision needed for iteration $j$. Similarly, $\wt(i, j, \stari, \log f, e')$ is the weight of the element $e'$ at the end of iteration $(i, j)$ measured with quality up to $(\stari, \starj)$. Let $\od(i, j, \stari, \starj, S)$ denote the number of elements $e \in S$ such that $\Call{weightEstimate}{i, j, \stari, \starj, e}$ returns less than $1/2$. $\dplus(i, j, \stari, \starj, S)$ denote the number of elements $e \in S$ with $\wt(i, j, \stari, \starj, e) < 1/4$ and $\dmin(i, j, \stari, \starj, S)$ denote the number of elements $e \in S$ with $\wt(i, j, \stari, \starj, e) < 1$. In the analysis, we estimate with sufficiently high probability that $e$ is covered when $\wt(e) \geq 1$, and that $e$ is uncovered when $\wt(e) < 1/4$. In between, we allow for the algorithm to treat $e$ arbitrarily.

In the following discussion, we \textbf{fix} the randomness in the algorithm prior to iteration $(\stari, \starj)$. That is, all the random bits that were sampled are considered fixed. The only change is the randomness associated with phase $i$. This includes the retroactive updates made for phases $i' = 1, 2, \dots i-1$, in this order.
 
\begin{lemma} \label{lem:alg-main:approx}
    \cref{alg:err2} satisfies the following equations during every phase $\stari$ and for every pair of iterations $(i, j) \leq (\stari, \starj)$ provided $\delta, K$ are constants chosen large enough. Let $b$ denote the number of iterations between $(i, j)$ and $(\stari, \starj)$ that are in phase $\stari$. That is, $b = \starj - j + 1$ if $i = \stari$ and $\starj$ otherwise. We have,
    \begin{align}
        \Pr\left[\derr_{new}(i, j, \stari, \starj, S) < \Delta / 2^{i} \mid \dplus(i, j, S) \geq \Delta / 2^{i-1} \right] &\leq \frac{1}{8} \cdot \frac{1}{2^{\stari-i}}\label{eqn:deg:under-estimate2} \\ 
        \Pr\left[ \derr_{new}(i, j, \stari, \starj, S) \geq \Delta / 2^{i} \mid \dmin(i, j, S) < \Delta / 2^{i+3} \right] &\leq \frac{1}{8^{b}} \cdot \frac{1}{2^{\stari-i}} \label{eqn:deg:over-estimate2} \\
        \forall t \ \geq 1, \ \ \Pr\left[ \dmin(i, \log f, \stari, \log f, S) \geq t \cdot \Delta / 2^{i-1} \right] &\leq \frac{1}{t^3} \cdot \frac{1}{2^{\stari-i+1}}\label{eqn:set:sparsification2} \\
        \forall t \ \geq 1, \ \ \Pr\left[ \od(i, \log f, \stari, \log f, S) \geq t \cdot \Delta / 2^{i-1} \right] &\leq \frac{1}{t^3} \cdot \frac{1}{2^{\stari-i+1}}\label{eqn:set:sparsification2p} \\
         \Pr\left[ \wterr(i, j, \stari, \starj, e) < 1/2 \mid \wt(i, j, e) \geq 1 \right] &\leq \frac{1}{8} \cdot \frac{1}{2^{\stari-i}}  \label{eqn:ele:under-estimate2}\\
         \Pr\left[ \wterr(i, j, \stari, \starj, e) \geq 1/2 \mid \wt(i, j, e) < 1/4 \right] &\leq \frac{1}{8^b} \cdot \frac{1}{2^{\stari-i}} \label{eqn:ele:over-estimate2} 
    \end{align}
\end{lemma}

\begin{proof}
    We prove \cref{lem:alg-main:approx} by induction on $(i, j)$ and in order from \cref{eqn:deg:under-estimate2,eqn:deg:over-estimate2,eqn:set:sparsification2,eqn:set:sparsification2p,eqn:ele:under-estimate2,eqn:ele:over-estimate2}. Note that $\dplus(i, j, S)$ refers to the effective degree of the set $S$ at the beginning of iteration $(i, j)$ of the algorithm (possibly different from current phase $\stari$).

    For the base case we have $i = 0$, i.e., the base case of the recursion in the algorithm. In this case, the output of the algorithms are error free and all the inequalities are trivially true.

    \begin{align*}\label{ind-hypo-main-lemma}
        \text{Induction Hypothesis:} \ \ \ \ \ \ \ &\text{\cref{eqn:deg:under-estimate2,eqn:deg:over-estimate2,eqn:set:sparsification2,eqn:set:sparsification2p,eqn:ele:under-estimate2,eqn:ele:over-estimate2}  are true  for every }(i, j, \stari, \starj) \\ &\text{ even if the randomness in phases prior to phase } \stari \text{ are kept fixed}
    \end{align*}

    We use strong induction. We define some bad events with the usual connotation, if none of the bad events occur, then the events associated with the probabilities in the lemma must occur. Bounding the probabilities that any of these events occur suffices to bound the probabilities.

    \paragraph{Proving \cref{eqn:deg:under-estimate2}}

    \begin{mdframed}[backgroundcolor=gray!10, linecolor=brown!40!black, roundcorner=5pt]
    \textbf{Bad Event $\cE_1$ -- Too many covered elements for a set}: \\
        For some $i' < i$, $|\hS|$ sampled in \cref{alg:GetDeg} contains at least $2^{\stari-i'+\delta} \cdot K$ distinct elements that survived until \cref*{line:bad-event-GetDeg}.
    \end{mdframed}

    From \cref{eqn:set:sparsification2} applied at phase $(i'-1, \log f)$, we have $\Pr[\od(i'-1, \log f, \stari, S) \geq t \cdot \Delta / 2^{i'-2}] \leq 1/t^3 \cdot \frac{1}{2^{\stari - i' + 2}}$. We consider two cases.

    \begin{itemize}
        \item (Case 1: $\od(i'-1, \log f, \stari, S) \geq \Delta \cdot 2^{\delta - i' - 1}$). In this case the probability is low enough that we can afford to count the execution as failed. The probability of error is at most $2^{-(\stari-i'+2+3\delta)}$.
        \item (Case 2: $\od(i'-1, \log f, \stari, S) < \Delta \cdot 2^{\delta - i' - 1}$). Out of the $2^{\stari} \cdot K$ elements that were sampled in $\hS$, the expected number of them that pass $\GetWt{}$ tests is at most $2^{\stari - i' + \delta - 1} \cdot K$. Using \cref{clm:weird:chernoff}, the desired failure probability is at most $\exp{(-2^{\stari-i'+\delta-1} \cdot K / 3)}$.
    \end{itemize}

    Applying the union bound over all $i'$ the probability of failure can be bound by $2^{-(\stari-i+\delta)}$.
    
    \begin{mdframed}[backgroundcolor=gray!10, linecolor=brown!40!black, roundcorner=5pt]
    \textbf{Bad Event $\cE_2$ -- In a dense set, too few uncovered elements sampled in $\hS$}: \\
        $\dplus(\hS) < 2^{\stari - i} \cdot K / 2$ given that $\dplus(i, j, S) \geq \Delta / 2^{i}$.
    \end{mdframed}
    Given that $\dplus(i, j, S) \geq \Delta / 2^i$, the expected number of uncovered elements is at least $2^{\stari - i} \cdot K$. By \cref{item:delta-at-most-1-le}, the probability that less than $2^{\stari - i} \cdot K / 2$ is sampled is at most $\exp{-(2^{\stari - i} \cdot K / 4)}$.
    
    \begin{mdframed}[backgroundcolor=gray!10, linecolor=brown!40!black, roundcorner=5pt]
    \textbf{Bad Event $\cE_3$ -- In a dense set, too few uncovered elements survived phase filtering}: \\
        Out of the $2^{\stari - i} \cdot K / 2$ elements that were sampled in $\hS$, less than $2^{\stari - i} \cdot 3 K / 8$ survived the weight estimate tests of the phase filtering process. 
    \end{mdframed}

    Consider an element $e$ with $\wt(e) < 1/4$ at the end of iteration $(i, j)$. The probability that it is estimated to be covered at iteration $(i', \log f)$ is at most $1/8 \cdot 1/2^{\stari - i'}$. By the union bound the probability that it is estimated to be covered upto iteration $(i-1, \log f)$ is at most $1/8 \cdot 1/2^{\stari - i} \leq 1/8$. The expected number of such elements out of the $2^{\stari - i} \cdot K / 2$ elements is at most $2^{\stari - i} \cdot K / 16$. By \cref{clm:weird:chernoff}, the probability that at least $2^{\stari - i} \cdot K / 8$ are estimated to be covered is at most  $\exp{-(2^{\stari - i} \cdot K / 48)}$.

    \begin{mdframed}[backgroundcolor=gray!10, linecolor=brown!40!black, roundcorner=5pt]
    \textbf{Bad Event $\cE_4$ -- In a dense set, too few uncovered elements are uncovered upto iteration $(i, j)$}: \\
        Out of the $2^{\stari - i} \cdot K / 4$ elements that survived phase filtering tests, less than $2^{\stari - i} \cdot K / 8$ passed the density tests.
    \end{mdframed}

    We use \cref{eqn:ele:over-estimate2} on each of the iterations $(i, j')$ with $j' < j$. For a given element $e$, the probability that it remains uncovered in at least one of the iterations is given by $1/8 \cdot 1/2^{\stari - i} \cdot \sum\limits_{j' < j} 8^{-j'} \leq 7^{-1}$. Now the expected number of elements covered is at most $2^{\stari - i} \cdot K / 28 < 2^{\stari - i} \cdot K / 8$. Applying \cref{clm:weird:chernoff}, the desired probability is at most $\exp{-(2^{\stari-i} \cdot K/24)}$.
    
    \begin{lemma} \label{proof:eqn:deg:under-estimate}
        If $\cE_1, \cE_2$, $\cE_3$ and $\cE_4$ do not occur during the call $\GetDeg{i, j, \stari, S}$, then the output is at least $\Delta / 2^i$. Furthermore, the probability of any of these events occurring is bounded as in \cref{eqn:deg:under-estimate2}.
    \end{lemma}

    \begin{proof}  
        For large enough $K$,$\delta$ the sum of the failure probabilities of $\cE_1, \cE_2, \cE_3$ and $\cE_4$ is bounded by $2^{-(\stari-i+3)}$ as desired. Finally, we need to compute the probability that the degree estimate filters do not return $0$. However, the degree estimates essentially execute the same testing procedure, except that failure probability is boosted by a factor of $2^{j-j'}$. 
        We can then afford to union-bound these failures, and the final probability is at most twice the computed value. 
    \end{proof}

    \paragraph{Proving \cref{eqn:deg:over-estimate2}.} To prove this bound, we focus on the run of $\GetDeg{i, j, \stari, \starj}$. We show that, with the desired failure probability, this procedure returns that $S$ was not dense during $(i, j)$. In the following, $\hS$ refers to the set sampled during the call $\GetDeg{i, j, \stari, \starj}$.

    \begin{mdframed}[backgroundcolor=gray!10, linecolor=brown!40!black, roundcorner=5pt]
    \textbf{Bad Event $\cE_5$ -- Too many uncovered elements sampled for light sets}: \\
        For a set $S$ with $\dmin(i, j', S) < \Delta / 2^{i + 3}$, we have $\dmin(\hS) \geq \dmin(i, j', S) / \Delta \cdot 2^{\stari - i - 2 + j-j'} \cdot K$
    \end{mdframed}

    The expected number of elements $e$ in $S$ with $\wt(e) < 1$ is at most $\dmin(i, j, S) / \Delta \cdot 2^{\stari-i-3 + j-j'} \cdot K$. Using \cref{lemma:chernoff} \cref{item:delta-at-least-1}, and setting $\delta = 1$, we get the desired probability is bound by, 
    $$\exp{-(\dmin(i, j, S) / \Delta \cdot 2^{\stari-i-3 + j-j'} \cdot K / 3)}$$

    Henceforth, we can assume that at least $2^{\stari-i+j-j'} \cdot K \cdot 3/4$ elements are covered in $\hS$. We next bound the probability that a third of them survive. We assume that all the uncovered elements survive the weight tests correctly, since it only affects us adversarially. Together, the estimated effective degree is less than $\Delta / 2^i$.

    Repeating the same arguments as that for $\cE_1$, the desired failure probability is 
    $$\exp{-(2^{\stari - i + \delta - 1} \cdot K \cdot 2^{j-j'}) / 3}$$
    
    \begin{mdframed}[backgroundcolor=gray!10, linecolor=brown!40!black, roundcorner=5pt]
    \textbf{Bad Event $\cE_7$ -- In a light set, too many covered elements are estimated uncovered up to iteration $(i, j)$ }: \\
        Out of the at least $2^{\stari - i - 2} \cdot 2^{j-j'} \cdot K/2$ covered elements, at least $2^{\stari - i - 2} \cdot 2^{j-j'} \cdot K/4$ of them have estimated weight less than $1/2$ at the end of iteration $(i, j)$
    \end{mdframed}
    
    \begin{lemma} \label{proof:eqn:deg-over}
        If $\cE_5$ and $\cE_7$ do not occur during the call $\GetDeg{i, j, \stari, \starj, S}$, then the output is less than $\Delta / 2^i$. Furthermore, the probability of any of these events occurring is bounded as in \cref{eqn:deg:over-estimate2}.
    \end{lemma}
    \begin{proof}
        Each of the events occur with probability $\exp{-\Theta(2^{\stari - i + b} \cdot K)}$. For large enough $x$, we have $\exp{(-x)} \leq x^{-3}$ which gives the desired inequality.
    \end{proof}

    \paragraph{Proving \cref{eqn:set:sparsification2}.}
    
    \begin{lemma} \label{proof:eqn:set-sparse}
        $\forall t \ \geq 1, \ \ \Pr\left[ \dmin(i, \log f, \stari, \log f, S) \geq t \cdot \Delta / 2^{i-1} \right] \leq \frac{1}{t^3} \cdot \frac{1}{2^{\stari-i+1}} $
    \end{lemma}
    \begin{proof}
        If $\dmin(i, \log f, \stari, \log f, S) \geq t \cdot \Delta / 2^{i-1}$, then $\dmin(i', \log f, S) \geq t \cdot \Delta / 2^{i-1}$ for every $i' \geq i - \log_2 t$. At the end of each of these iterations, the degree estimate to set $S$ fails with probability at most $1/8 \cdot 1/2^{\stari - i + 1}$ by \cref{eqn:deg:under-estimate2}. Note that otherwise, we would add a weight of $1$ to $S$. These tests are independent, and hence, the total failure probability is bound by the given expression.
    \end{proof}

    \paragraph{Proving \cref{eqn:set:sparsification2p}.}
    
    \begin{lemma}\label{proof:eqn:set:sparsification2p} 
    $\forall t \ \geq 1, \ \ \Pr\left[ \od(i, \log f, \stari, \log f, S) \geq t \cdot \Delta / 2^{i-1} \right] \leq \frac{1}{t^3} \cdot \frac{1}{2^{\stari-i+1}} $
    \end{lemma}
    \begin{proof}
        We follow the proof of \cref{proof:eqn:set-sparse}. Note that we can replace $\dmin$ in \cref{eqn:deg:over-estimate2} with $\od$, and the inequality still holds since we do not have to account for the probability of failure of the weight estimates for some of the elements. Hence, the same bound as that of \cref{eqn:set:sparsification2} holds.
    \end{proof}
    
    \paragraph{Proving \cref{eqn:ele:under-estimate2}.}

    For the discussion below, we call an element \textit{heavy} during some iteration $(i, j)$ if there exists at least $f / 2^{j-2}$ sets \textit{estimated} to be dense and containing it, after the retroactive updates until phase $\stari$. Similarly we call an element \textit{light} if there exists at most $f / 2^{j+3}$ sets that are estimated to be dense and contain $e$.
    
    \begin{mdframed}[backgroundcolor=gray!10, linecolor=brown!40!black, roundcorner=5pt]
    \textbf{Bad Event $\cE_8$ -- For a heavy element $e$, too few dense sets are sampled }: \\
        Out of the $2^j \cdot K \cdot 2^{\stari - i}$ sets that were sampled, at most $2^{\stari - i + 1} \cdot K $ sets are dense.
    \end{mdframed}

    Using \cref{lemma:chernoff} \cref{item:delta-at-most-1-le}, we get failure probability $\exp{(-\Theta(2^{\stari - i + 1} \cdot K))}$.
    
    \begin{mdframed}[backgroundcolor=gray!10, linecolor=brown!40!black, roundcorner=5pt]
    \textbf{Bad Event $\cE_9$ -- Too few dense sets survive upto iteration $(i, j)$}: \\
        Out of the $2^j \cdot K \cdot 2^{\stari - i + 1}$ dense sets that were sampled, less than $2^{\stari - i} \cdot 2^{j-j'} \cdot K $ sets pass the dense set test estimates.
    \end{mdframed}
    
    \begin{lemma}\label{proof:eqn:ele:under-estimate2}
        \cref{eqn:ele:under-estimate2} is true.
    \end{lemma}

    \paragraph{Proving \cref{eqn:ele:over-estimate2}}

    We consider an iteration $(i, j')$ during which an element $e$ is uncovered. We study the probability that the oracle $\GetWt{i, j', \stari, j}$ estimates that $\wt(e) \geq 1/2$ given that $\wt(e) < 1/4$.
    
    \begin{mdframed}[backgroundcolor=gray!10, linecolor=brown!40!black, roundcorner=5pt]
    \textbf{Bad Event $\cE_{10}$ -- For a light element $e$, too many dense sets are sampled}: \\
        Given that an element $e$ has at most $f / 2^{j'+2}$ sets that are containing it and estimated to be dense, at least $2^{j - j' - 2} \cdot 2^{\stari - i} \cdot K$  dense sets containing $e$ are sampled in $\hat{D}$ in \cref*{line:dense-est-sample}.
    \end{mdframed}

    \begin{mdframed}[backgroundcolor=gray!10, linecolor=brown!40!black, roundcorner=5pt]
    \textbf{Bad Event $\cE_{11}$ -- For a light element $e$, too many light sets containing $e$ pass dense estimate tests}: \\
        Of the at least $2^{j - j'}\cdot 2^{\stari - i} \cdot 3K/4$ light sets containing $e$, at least $2^{j - j'} \cdot 2^{\stari - i} \cdot K / 4$ are marked as dense.
    \end{mdframed}
    
    \begin{lemma}\label{proof:eqn:ele:over-estimate2}
        \cref{eqn:ele:over-estimate2} is true.
    \end{lemma}

    The proof of these inequalities completes the induction.
\end{proof}

\section{Integral Set Cover}
\label{sec:frac-to-integral}
In this section, we provide the missing details to prove \cref{thm:main}, i.e., how to convert the fractional algorithms to integral ones.
The ideas used for obtaining an integral version follow from the prior work~\cite{GMRV20} and our analysis from the preceding sections. 
Nevertheless, we provide some details for completeness and outline the algorithm performing the rounding.

To convert a fractional algorithm into an integral one, the common change is that whenever a weight of $w$ is added to a set $S$, $S$ is included in the integral set cover with probability $w$. This randomness can be sampled in advance, i.e., before an algorithm is executed. 
We will also need to modify the final naive covering step; the modification is natural. 
For example, if a set $S$ contains an element $e$ that is not covered and $S$ has the least ID among all sets that contain $e$, then $S$ will be added to the cover.

\subsection{Converting \cref{alg:get-set-weight} into an integral LCA}

\textbf{Modifying \cref{alg:Is-Ele-Cov}.} 
To convert \cref{alg:get-set-weight} into an integral LCA, we need one additional change. In the procedure for \cref{alg:Is-Ele-Cov}, we can no longer afford to say that an element is ``covered'' if its ``weight'' exceeds $1$, since it might happen that none of the sets were sampled. We will show below that, on expectation, the weight of every element is $O(1)$ when it gets covered. However, this is still insufficient to obtain LCAs with deterministic query complexity, since $\wt(e)$ might be large, albeit with small probability. 
To remedy this, we use the fact that the random sets chosen to be sampled in iteration $(i, j)$ can be pre-sampled and we can use these pre-samples to prune out redundant density tests. Essentially, if a set $S$ was not sampled during iteration $(i, j)$, then we do not need to simulate its state in this iteration. Also, with high probability, we can guarantee that $\mu_k = \Theta(2^{j-k} L^3)$ w.h.p. Then, the distribution of $\mu_k$ is almost uniform, making the analysis actually simpler.
 
The argument for the approximation ratio of the integral algorithm in \cref{alg:get-set-weight} is identical to the approach in \cite{GMRV20}. The only non-trivial step in the proof is to argue that $X_e$, the number of sets containing $e$ that get added to the cover when $e$ is covered for the first time, is $O(1)$. The remaining steps follow from the existing proof. The proof of the following lemma is also unchanged from \cite{GMRV20}. The only criteria we need is that the density test estimates for a set remain consistent across the $\log f$ iterations in a single phase. Since this is true (w.h.p.), the same analysis follows.

\begin{lemma}
\label{lem:int1}
    In the integral version of \cref{alg:get-set-weight}, we have that for every element $e$, $\xpect{X_e} \leq 12$.
\end{lemma}

\begin{proof}
    The idea is to fix a non-increasing sequence $n_1, n_2, \dots n_{\log f}$ where $n_i$ represents the number of dense sets incident on $e$ during iteration $j$ (of some arbitrary phase).

    Note that there is some probability (depending on the density set tests) that this sequence is observed during a run of the algorithm. However, given such a sequence, the value of $X_e$ depends purely on the randomness in the sampling of the sets (not in the density estimates). 

    Let us define $F(j)$ for some $1 \leq j \leq \log f$ as the expected number of sets covering an element $e$ given that it is uncovered upto iteration $j-1$. We can write,
    \begin{align*}
        F(j) &\leq \Pr[e \text{ is covered at } j] \cdot n_j 2^j / f + \Pr[e \text{ is not covered at } j] \cdot F(j + 1) \\
        &\leq n_j 2^j / f + \{\prod\limits_{n_j  \ \text{dense sets}}(1 - 2^j / f) \} \cdot F(j + 1) \\
        &\leq n_j 2^j / f + \exp{(-n_j 2^j / f)} \cdot F(j + 1)
    \end{align*}
    Using the last equation, we prove by induction, on decreasing order of $j$, that $F(j) \leq \max\{D, C n_j 2^j / f\}$ for sufficiently large constants $C, D$ to be chosen later. For the base case we have $F_{\log f} \leq n_{\log f} \leq \max\{D, C n_{\log f}\}$ for $C \geq 1$.

    We consider two cases. 
    
    Case (i): $n_j \leq f / 2^{j-2}$. In this case we have $F(j) \leq 4 + e^{-n_j 2^j / f} \cdot F(j + 1)$. To complete the induction, we insist that $4 + D/e^4 \leq D$ and $4 + 4C \leq D$.

    Case (ii): $n_j \geq f / 2^{j-2}$. In this case, we have that
    \[
        n_j 2^j / f + \exp{(- 4n_j f / 2^j)} F(j+1) \leq x + \max\{D, C x\} e^{-x} \leq Cx
    \]
    whenever $x \geq 4$ and $C, D \geq 2$. 
    We set $C=2, D=12$. It follows now that $F(1) \leq \max\{12, 2 n_1 / f\} \leq 12$.
\end{proof}

The rest of the proof follows identically to the fractional algorithm.

\newcommand{\GetCv}[1]{\Call{coveredEstimate}{#1}}

\subsection{Converting \cref{alg:err2} into an integral LCA}

\paragraph{Modifications needed to the algorithm.} For every iteration $(i, j)$ and every set $S$, we sample a bit with probability $2^j / f$, denoting the coverage for that set. These bits inform us about which sets to choose in the integral cover.

Let $t_{S, i, j}$ denote the random bit that was sampled for $S$ during $(i, j)$. We execute \cref{alg:GetDeg,alg:GetWt} except that \cref{alg:GetWt} always returns either $0$ or $1$, denoting whether an element $e$ is covered or not. If the procedure returns $1$, then we will have queried a set that belongs to the integral cover. If the procedure does not return $1$, then $e$ may have been covered, but we have not queried that set yet. 
Hence $\GetWt{}$ now only has one sided error. Instead of computing weights, whenever we sample a set containing $e$ and find that it has been added to the cover, only then do we consider it covered. These changes are implemented in \cref{alg:GetWt2}.

\paragraph{\cref*{line:ele-sample}.} In the fractional version, we sampled sets containing $e$ randomly. However, in the integral version this is modified as follows: Query the entire neighborhood of $e$ and find the subset of sets $S$ that have $t_{S, i, j} = 1$. Simulate the density for only these sets. If there are more than $2^{\starj} \cdot 2^{\stari - i} \cdot K$ candidates, then pick these many uniformly at random.

\paragraph{Query complexity.} The recurrence for the query complexity remains almost unchanged. The only difference is that now, the random neighborhood of an element to explore is determined by random samples $t_{S, i, j}$. So we do the following, first query every set $S \ni e$ and its $t_{S, i, j}$ value. Then, take as many sets as possible that have $t_{S, i,j} = 1$ and consider this to be the sampled set. Note that if any of these sets were estimated to be dense, then we would add them to the Set Cover. With a similar argument to the previous time complexity analysis, we can derive $S(i, j, b) \leq f \cdot C^{t_{i,j}} \cdot 2^{\delta b}$ and $E(i, j, b) \leq f \cdot C^{t_{i,j} + 1/2} \cdot 2^{\delta \cdot b}$.

\paragraph{Approximation Ratio.} Arguing for the approximation ratio is much less straightforward. The major new challenge is in showing that when an element $e$ gets covered, the expected number of sets that contain it and are added to the cover is $O(1)$. We first show this inequality. 
\begin{algorithm}
    \begin{algorithmic}[1]
        \Procedure{coveredEstimate}{$i, j, \stari, \starj, e$} 
            \Statex \Comment{Estimate whether element $e$ is covered at the end of iteration $(i, j)$ after the retro-active updates up to phase $i^{\star}$, with failure probability boosted up to $\starj$}
            \Statex \textbf{Input:} Phases $i, i^{\star} (\geq i)$, iteration $j$ and element $e$ with $0 \leq i \leq \log \Delta$ and $1 \leq j \leq \log f$.
            \Statex \textbf{Parameters}: $K, \delta$ are sufficiently large constants.
            \Statex \textbf{Base Cases}: If $i = 0$ then return $0$.
            \State $\wt \gets \GetCv{\previ{i}, \previ{j}, i^{\star}, \starj, e}$ \Comment{$(\previ{i}, \previ{j})$ is the previous iteration to $(i, j)$ \footnote{If $j=1$ then $(\previ{i}, \previ{j}) = (i-1, \log f)$ otherwise $(\previ{i}, \previ{j}) = (i, j-1)$}} 
                \If{$\wt = 1$} \Comment{Check if $e$ is covered already.}
                    \State \Return $1$
                \EndIf 
            \State $D \gets $ The set of at most $2^{j} \cdot K \cdot 2^{\delta \cdot b}$ sets containing $e$ that have $t_{S, i, j} = 1$
            \Statex \Comment{For the $\LCA$, we query the entire neighborhood of $e$ and then their random bits for this iteration.}
            \For{$j' = 1, 2, \dots j$}
                \For{every set $S \in \hat{D}$}
                    \If{$\GetDeg{i, j', \stari, \starj, S} \geq \Delta / 2^i$}
                        \State \Return $1$
                    \EndIf 
                \EndFor
                \If{$|D| > 2^{j - j'+\delta + \delta \cdot b} \cdot K$}
                    \State \Return $0$
                \EndIf 
            \EndFor 
            \State \Return $0$
        \EndProcedure 
    \end{algorithmic}
    \caption{LCA oracle to estimate whether an element is covered.}
    \label{alg:GetWt2}
\end{algorithm}

We follow roughly the same idea. However, this time, it is not guaranteed that a set considered light during iteration $(i, j)$ is also considered light during iteration $(i, j-1)$. There is a (small) probability that this event occurs. In the previous approaches, we fixed the number of dense sets in each iteration and argued that regardless of this sequence, the expected value is $O(1)$. 
This was possible because the sequence is guaranteed to be decreasing. In the case of \cref{alg:err2}, this is not guaranteed. However, there is very low probability that the sequence increases exponentially, which turns out to be enough. 

\begin{lemma}
    It holds that $\xpect{X_e} \leq O(1)$.
\end{lemma}

\begin{proof}
Let $F(j, n_j)$ denote the expected value of the number of sets that are added to the cover when an element $e$ is first covered given that it remains uncovered starting from iteration $j$ with $n_j$ sets estimated to be dense.

We show that $F(j, n_j) \leq \max\{D, C \cdot 2^j n_j / f\}$ for suitably large constants $D, C$, by induction. The result then follows as $\xpect{X_e} \leq F(1, f) \leq O(1)$.

For the base case we consider $j = \log f$, in which case $F(\log f, n_{\log f}) \leq \max\{D, C \cdot n_{\log f}\}$ which is trivially true whenever $C \geq 1$.

For the induction step we consider two cases, conditioned on $n_{j+1}$.

We can write the recurrence $F(j, n_j) \leq F(j+1, n_{j+1}) \cdot \exp{-(n_j 2^j / f)} + n_j 2^j / f$.

For convenience, let $m_j = n_j 2^{j} / f$.

\begin{itemize}
    \item (Case 1). $n_{j+1} \leq 3 n_j$. In this case we can bound, $F(j, n_j) \leq m_j + \exp{(-m_j}) \cdot C m_{j + 1} \leq m_j + 3Cm_j\exp{(-m_j)} \leq 4C$
    \item (Case 2). $n_{j+1} > 3 n_j$. Out of the $n_{j+1}$ sets, at least $2n_j$ were not marked dense during iteration $(i, j)$. The probability that this can happen for a set is at most $1/8$. Going over the randomness in density estimates at iteration $j$, the probability is at most $\binom{n_{j+1}}{n_{j}} \cdot 8^{-(n_{j+1} - n_{j})} \leq 2^{-2n_{j+1} + n_j} \leq 2^{-5n_{j+1} / 3}$.

    Hence we get $F(j, n_j) \leq n_j \cdot 2^j / f + 2^{-5/3 \cdot n_{j+1}} \cdot \exp{-(n_j 2^j / f)} \cdot n_{j+1} \cdot 2^{j + 1} / f$.

    Using $2^{-5/3 x} \cdot x \leq 1$ for $x \geq 1$, we get $F(j, n_j) \leq m_j + \exp^{-m_j} \cdot 2m_j \leq 3 m_j$.
\end{itemize}

    Going over the two cases, it suffices to choose $C = 3$ and $D = 12$.

\end{proof}

Using the same arguments for analogous bad events, we can derive analogous formulae for \cref{eqn:deg:over-estimate2,eqn:deg:under-estimate2,eqn:set:sparsification2,eqn:set:sparsification2p,proof:eqn:ele:over-estimate2,proof:eqn:ele:under-estimate2}, which completes the proof of \cref{thm:main}.
\section{Proof of \cref{corollary}}
\label{sec:proof-of-corollary}

In this section we give the proof of \cref{corollary}. First we sample $100 \Delta f$ sets uniformly at random with replacement. For each set $S$, we execute the $\LCA$ \cref{alg:GetWt} to compute whether the $\LCA$ chose $S$ in the cover. Let $k$ denote the number of such sets. We output $m k / (100 \Delta f)$ as the estimate of the size of the cover. The query complexity is deterministically $O(\Delta^{O(\log f)})$. We next argue the approximation ratio.

First let us assume we fail if the $\LCA$ fails to output an $O(\log \Delta)$ approximate solution with probability at most $1/100$. Next, if the $\LCA$ selected $k'$ sets, then the expected number of them sampled is at least $k' \cdot 100 \Delta f / m \geq 100$. Hence, \cref{clm:weird:chernoff}, with probability at least $2/3$, the output is $O(\log \Delta)$-approximate.

To improve the space complexity, we first observe that \cref{alg:GetWt} can be implemented with just $O(\log \Delta \log f)$ space. The state of the recursion can be maintained by storing the current recursive call (set ID $S$, element ID $e$, and integers $i, j, \stari, \starj$) and the current iteration of the for loops that have been executed. To implement the pruning procedure, we only need to store the number of sets/elements that have passed so far. All of these information can be stores in $O(\poly(\log \Delta \log f))$ bits. When the $\LCA$ wants to a particular random bit, suppose that it is fed this as input by an oracle. As such, the total number of random bits needed by \cref{alg:GetWt} is $\Delta^{O(\log f)}$. We can then use the pseudo-random generator of \cite{Nisan92} to simulate $\Delta^{O(\log f)}$ truly random bits using only $O(\log \Delta \log f)$ bits. Finally, the output of a set $S$ depends only on the atmost $\Delta f$ other sets that share an element with $S$. Thus, the entire randomness needed by the algorithm can be generated from just $O(\Delta f \poly(\log \Delta \log f))$ truly random bits with only a $\poly(\log \Delta, \log f)$ loss in the local runtime of the algorithm.

\paragraph{Modifying the approach of \cite{YYI09} for $O(\log \Delta)$-approximate Set Cover.} The oracle of \cite{YYI09} simulates the following procedure: For $i = 1, 2, \dots \Delta$, sequentially, create a random ordering of the sets and then iterate over the sets in this order and greedily pick any that have size at least $\Delta - i + 1$. They then analyze the average case complexity of an $\LCA$ that simulates this procedure. The same analysis when modified with the following procedure yields a query complexity of $(\Delta f)^{O(\log \Delta)}$. The procedures is as follows: For $i = 1, 2, \dots \log \Delta$, sequentially, create a random ordering of the sets and then iterate over the sets in the random order and greedily pick any set that has size at least $\Delta / 2^i$. The analysis of \cite{YYI09} follows almost immediately for this version of the algorithm, at the cost of only losing $O(\log \Delta)$ factor in the approximation ratio.

\bibliographystyle{alpha}
\bibliography{references}

@string{acm = "ACM Press, New York"}

@string{eccc = "Electronic Colloquium in Computational Complexity"}

@string{ieee = "IEEE Computer Society Press, Los Alamitos, California"}

@string{jacm = "Journal of the ACM"}

@string{random = "Random Structures and Algorithms"}

@string{sicomp = "SIAM J. on Comput."}

@string{sicomp = "SIAM Journal on Computing"}

@string{springer = "Springer Verlag"}

@string{ieee = "IEEE"}

@string{acm = "ACM"}

@string{soda = "{ACM}-{SIAM} {S}ymposium on {D}iscrete {A}lgorithms"}

@string{focs = afocs}

@string{stoc = astoc}

@string{icalp = aicalp}

@string{stacs = astacs}

@string{structures = astructures}

@string{proc23 = "Proc.\ 23rd "}

@inproceedings{Ghaffari-LCA-FOCS22,
  author       = {Mohsen Ghaffari},
  title        = {Local Computation of Maximal Independent Set},
  booktitle    = {63rd {IEEE} Annual Symposium on Foundations of Computer Science, {FOCS}
                  2022, Denver, CO, USA, October 31 - November 3, 2022},
  pages        = {438--449},
  publisher    = {{IEEE}},
  year         = {2022},
  url          = {https://doi.org/10.1109/FOCS54457.2022.00049},
  doi          = {10.1109/FOCS54457.2022.00049},
  bibsource    = {dblp computer science bibliography, https://dblp.org}
}

@inproceedings{KMNT20,
  author       = {Michael Kapralov and
                  Slobodan Mitrovic and
                  Ashkan Norouzi{-}Fard and
                  Jakab Tardos},
  editor       = {Shuchi Chawla},
  title        = {Space Efficient Approximation to Maximum Matching Size from Uniform
                  Edge Samples},
  booktitle    = {Proceedings of the 2020 {ACM-SIAM} Symposium on Discrete Algorithms,
                  {SODA} 2020, Salt Lake City, UT, USA, January 5-8, 2020},
  pages        = {1753--1772},
  publisher    = {{SIAM}},
  year         = {2020},
  url          = {https://doi.org/10.1137/1.9781611975994.107},
  doi          = {10.1137/1.9781611975994.107},
  bibsource    = {dblp computer science bibliography, https://dblp.org}
}

@inproceedings{GMRV20,
  title={Improved local computation algorithm for set cover via sparsification},
  author={Grunau, Christoph and Mitrovi{\'c}, Slobodan and Rubinfeld, Ronitt and Vakilian, Ali},
  booktitle={Proceedings of the Fourteenth Annual ACM-SIAM Symposium on Discrete Algorithms},
  pages={2993--3011},
  year={2020},
  organization={SIAM}
}

@inproceedings{YYI09,
  author       = {Yuichi Yoshida and
                  Masaki Yamamoto and
                  Hiro Ito},
  editor       = {Michael Mitzenmacher},
  title        = {An improved constant-time approximation algorithm for maximum matchings},
  booktitle    = {Proceedings of the 41st Annual {ACM} Symposium on Theory of Computing,
                  {STOC} 2009, Bethesda, MD, USA, May 31 - June 2, 2009},
  pages        = {225--234},
  publisher    = {{ACM}},
  year         = {2009},
  url          = {https://doi.org/10.1145/1536414.1536447},
  doi          = {10.1145/1536414.1536447},
  bibsource    = {dblp computer science bibliography, https://dblp.org}
}

@inproceedings{LenzenL18,
  author       = {Christoph Lenzen and
                  Reut Levi},
  editor       = {Ioannis Chatzigiannakis and
                  Christos Kaklamanis and
                  D{\'{a}}niel Marx and
                  Donald Sannella},
  title        = {A Centralized Local Algorithm for the Sparse Spanning Graph Problem},
  booktitle    = {45th International Colloquium on Automata, Languages, and Programming,
                  {ICALP} 2018, July 9-13, 2018, Prague, Czech Republic},
  series       = {LIPIcs},
  volume       = {107},
  pages        = {87:1--87:14},
  publisher    = {Schloss Dagstuhl - Leibniz-Zentrum f{\"{u}}r Informatik},
  year         = {2018},
  url          = {https://doi.org/10.4230/LIPIcs.ICALP.2018.87},
  doi          = {10.4230/LIPICS.ICALP.2018.87},
  bibsource    = {dblp computer science bibliography, https://dblp.org}
}

@inproceedings{ParterRVY19,
  author       = {Merav Parter and
                  Ronitt Rubinfeld and
                  Ali Vakilian and
                  Anak Yodpinyanee},
  editor       = {Avrim Blum},
  title        = {Local Computation Algorithms for Spanners},
  booktitle    = {10th Innovations in Theoretical Computer Science Conference, {ITCS}
                  2019, January 10-12, 2019, San Diego, California, {USA}},
  series       = {LIPIcs},
  volume       = {124},
  pages        = {58:1--58:21},
  publisher    = {Schloss Dagstuhl - Leibniz-Zentrum f{\"{u}}r Informatik},
  year         = {2019},
  url          = {https://doi.org/10.4230/LIPIcs.ITCS.2019.58},
  doi          = {10.4230/LIPICS.ITCS.2019.58},
  bibsource    = {dblp computer science bibliography, https://dblp.org}
}

@inproceedings{ArvivCLP23,
  author       = {Rubi Arviv and
                  Lily Chung and
                  Reut Levi and
                  Edward Pyne},
  editor       = {Nicole Megow and
                  Adam D. Smith},
  title        = {Improved Local Computation Algorithms for Constructing Spanners},
  booktitle    = {Approximation, Randomization, and Combinatorial Optimization. Algorithms
                  and Techniques, {APPROX/RANDOM} 2023, September 11-13, 2023, Atlanta,
                  Georgia, {USA}},
  series       = {LIPIcs},
  volume       = {275},
  pages        = {42:1--42:23},
  publisher    = {Schloss Dagstuhl - Leibniz-Zentrum f{\"{u}}r Informatik},
  year         = {2023},
  url          = {https://doi.org/10.4230/LIPIcs.APPROX/RANDOM.2023.42},
  doi          = {10.4230/LIPICS.APPROX/RANDOM.2023.42},
  bibsource    = {dblp computer science bibliography, https://dblp.org}
}

@inproceedings{Behnezhad21,
  author       = {Soheil Behnezhad},
  title        = {Time-Optimal Sublinear Algorithms for Matching and Vertex Cover},
  booktitle    = {62nd {IEEE} Annual Symposium on Foundations of Computer Science, {FOCS}
                  2021, Denver, CO, USA, February 7-10, 2022},
  pages        = {873--884},
  publisher    = {{IEEE}},
  year         = {2021},
  url          = {https://doi.org/10.1109/FOCS52979.2021.00089},
  doi          = {10.1109/FOCS52979.2021.00089},
  bibsource    = {dblp computer science bibliography, https://dblp.org}
}

@inproceedings{NO08,
  author       = {Huy N. Nguyen and
                  Krzysztof Onak},
  title        = {Constant-Time Approximation Algorithms via Local Improvements},
  booktitle    = {49th Annual {IEEE} Symposium on Foundations of Computer Science, {FOCS}
                  2008, October 25-28, 2008, Philadelphia, PA, {USA}},
  pages        = {327--336},
  publisher    = {{IEEE} Computer Society},
  year         = {2008},
  url          = {https://doi.org/10.1109/FOCS.2008.81},
  doi          = {10.1109/FOCS.2008.81},
  bibsource    = {dblp computer science bibliography, https://dblp.org}
}

@inproceedings{RTVX11,
    author = {R. Rubinfeld and G. Tamir and S. Vardi and N. Xie},
    title = {Fast Local Computation Algorithms},
    booktitle = {Proc.\ Innovations in Computer Science Conference},
    year = {2011},
    location = {Beijing, China},
}

@inproceedings{ARVX12,
author={Noga Alon and Ronitt Rubinfeld and Shai Vardi and Ning Xie},
title={Space-efficient local computation algorithms},
booktitle={23rd Annual ACM-SIAM Symposium on Discrete Algorithms (SODA 2012)},
year={2012},
month={January},
address={Kyoto, Japan},
}

@inproceedings{Ghaffari16,
  author    = {Mohsen Ghaffari},
  title     = {An Improved Distributed Algorithm for Maximal Independent Set},
  booktitle = {Proceedings of the Twenty-Seventh Annual {ACM-SIAM} Symposium on Discrete
               Algorithms, {SODA} 2016, Arlington, VA, USA, January 10-12, 2016},
  pages     = {270--277},
  year      = {2016},
  url       = {http://dx.doi.org/10.1137/1.9781611974331.ch20},
}

@article{BarenboimEPS16,
  author    = {Leonid Barenboim and
               Michael Elkin and
               Seth Pettie and
               Johannes Schneider},
  title     = {The Locality of Distributed Symmetry Breaking},
  journal   = {J. {ACM}},
  volume    = {63},
  number    = {3},
  pages     = {20},
  year      = {2016},
}

@article{Nisan92,
  author       = {Noam Nisan},
  title        = {Pseudorandom generators for space-bounded computation},
  journal      = {Comb.},
  volume       = {12},
  number       = {4},
  pages        = {449--461},
  year         = {1992},
  url          = {https://doi.org/10.1007/BF01305237},
  doi          = {10.1007/BF01305237},
  bibsource    = {dblp computer science bibliography, https://dblp.org}
}

@inproceedings{EvenMR14,
  author    = {Guy Even and
               Moti Medina and
               Dana Ron},
  title     = {Deterministic Stateless Centralized Local Algorithms for Bounded Degree
               Graphs},
  booktitle = {Algorithms - {ESA} 2014 - 22th Annual European Symposium, Wroclaw,
               Poland, September 8-10, 2014. Proceedings},
  pages     = {394--405},
  year      = {2014},
}

@article{LeviRY17,
  author    = {Reut Levi and
               Ronitt Rubinfeld and
               Anak Yodpinyanee},
  title     = {Local Computation Algorithms for Graphs of Non-constant Degrees},
  journal   = {Algorithmica},
  volume    = {77},
  number    = {4},
  pages     = {971--994},
  year      = {2017},
}

@inproceedings{GhaffariU19,
  author    = {Mohsen Ghaffari and
               Jara Uitto},
  title     = {Sparsifying Distributed Algorithms with Ramifications in Massively
               Parallel Computation and Centralized Local Computation},
  booktitle = {Proceedings of the Thirtieth Annual {ACM-SIAM} Symposium on Discrete
               Algorithms, {SODA} 2019, San Diego, California, USA, January 6-9,
               2019},
  pages     = {1636--1653},
  publisher = {{SIAM}},
  year      = {2019},
  url       = {https://doi.org/10.1137/1.9781611975482.99},
  doi       = {10.1137/1.9781611975482.99},
  bibsource = {dblp computer science bibliography, https://dblp.org}
}

@inproceedings{chang2019complexity,
  title={{The complexity of ($\Delta + 1$) coloring in congested clique, massively parallel computation, and centralized local computation}},
  author={Chang, Yi-Jun and Fischer, Manuela and Ghaffari, Mohsen and Uitto, Jara and Zheng, Yufan},
  booktitle={Proceedings of the 2019 ACM Symposium on Principles of Distributed Computing},
  pages={471--480},
  year={2019}
}

@inproceedings{MRVX12,
author={Y. Mansour and A. Rubinstein and Shai Vardi and Ning Xie},
title={Converting Online Algorithms to Local Computation Algorithms},
booktitle={Unpublished manuscript},
year={2011},
}

@inproceedings{MansourV13,
  author    = {Yishay Mansour and
               Shai Vardi},
  title     = {A Local Computation Approximation Scheme to Maximum Matching},
  booktitle = {APPROX-RANDOM},
  year      = {2013},
  pages     = {260-273},
}

@article{HassidimMV16,
  author    = {Avinatan Hassidim and
               Yishay Mansour and
               Shai Vardi},
  title     = {Local Computation Mechanism Design},
  journal   = {{ACM} Trans. Economics and Comput.},
  volume    = {4},
  number    = {4},
  pages     = {21:1--21:24},
  year      = {2016},
}

@inproceedings{DLRR13,
 	author = {A. Dutta and R. Levi and D. Ron and R. Rubinfeld},
	title= {A Simple Online Competitive Adaptation of Lempel-Ziv
		Compression with Efficient Random Access Support},
	booktitle= {Proceedings of the Data Compression Conference (DCC)},
        year = {2013},
	pages = {113--122}
}

@inproceedings{CampagnaGR13,
  author    = {Andrea Campagna and
               Alan Guo and
               Ronitt Rubinfeld},
  title     = {Local Reconstructors and Tolerant Testers for Connectivity
               and Diameter},
  booktitle = {APPROX-RANDOM},
  year      = {2013},
  pages     = {411-424},
  ee        = {http://dx.doi.org/10.1007/978-3-642-40328-6_29},
  bibsource = {DBLP, http://dblp.uni-trier.de}
}

@inproceedings{EvenLMR17,
  author    = {Guy Even and
               Reut Levi and
               Moti Medina and
               Adi Ros{\'{e}}n},
  title     = {Sublinear Random Access Generators for Preferential Attachment Graphs},
  booktitle = {44th International Colloquium on Automata, Languages, and Programming,
               {ICALP} 2017, July 10-14, 2017, Warsaw, Poland},
  pages     = {6:1--6:15},
  year      = {2017},
}

@article{KalePS13,
  author    = {Satyen Kale and
               Yuval Peres and
               C. Seshadhri},
  title     = {Noise Tolerance of Expanders and Sublinear Expansion Reconstruction},
  journal   = {SIAM J. Comput.},
  volume    = {42},
  number    = {1},
  year      = {2013},
  pages     = {305-323},
  ee        = {http://dx.doi.org/10.1137/110837863},
  bibsource = {DBLP, http://dblp.uni-trier.de}
}

@inproceedings{behnezhad2022time,
  title={Time-optimal sublinear algorithms for matching and vertex cover},
  author={Behnezhad, Soheil},
  booktitle={2021 IEEE 62nd Annual Symposium on Foundations of Computer Science (FOCS)},
  pages={873--884},
  year={2022},
  organization={IEEE}
}

@inproceedings{behnezhad2022almost,
  title={Almost 3-approximate correlation clustering in constant rounds},
  author={Behnezhad, Soheil and Charikar, Moses and Ma, Weiyun and Tan, Li-Yang},
  booktitle={2022 IEEE 63rd Annual Symposium on Foundations of Computer Science (FOCS)},
  pages={720--731},
  year={2022},
  organization={IEEE}
}

@inproceedings{dalirrooyfard2024pruned,
  title={Pruned Pivot: Correlation Clustering Algorithm for Dynamic, Parallel, and Local Computation Models},
  author={Dalirrooyfard, Mina and Makarychev, Konstantin and Mitrovic, Slobodan},
  booktitle={Forty-first International Conference on Machine Learning},
    year={2024}
}

@article{parnas2007approximating,
  title={Approximating the minimum vertex cover in sublinear time and a connection to distributed algorithms},
  author={Parnas, Michal and Ron, Dana},
  journal={Theoretical Computer Science},
  volume={381},
  number={1-3},
  pages={183--196},
  year={2007},
  publisher={Elsevier}
}

@article{berger1994efficient,
  title={Efficient NC algorithms for set cover with applications to learning and geometry},
  author={Berger, Bonnie and Rompel, John and Shor, Peter W},
  journal={Journal of Computer and System Sciences},
  volume={49},
  number={3},
  pages={454--477},
  year={1994},
  publisher={Elsevier}
}

@inproceedings{kuhn2006price,
  title={The price of being near-sighted},
  author={Kuhn, Fabian and Moscibroda, Thomas and Wattenhofer, Roger},
  booktitle={SODA 2006: 17th ACM-SIAM Symposium on Discrete Algorithms, Miami, Florida, USA},
  year={2006}
}

@inproceedings{levi2015local,
  title={Local computation algorithms for graphs of non-constant degrees},
  author={Levi, Reut and Rubinfeld, Ronitt and Yodpinyanee, Anak},
  booktitle={Proceedings of the 27th ACM symposium on Parallelism in Algorithms and Architectures},
  pages={59--61},
  year={2015}
}

@inproceedings{bhattacharya2023dynamic,
  title={Dynamic $(1+ \epsilon)$-approximate matching size in truly sublinear update time},
  author={Bhattacharya, Sayan and Kiss, Peter and Saranurak, Thatchaphol},
  booktitle={2023 IEEE 64th Annual Symposium on Foundations of Computer Science (FOCS)},
  pages={1563--1588},
  year={2023},
  organization={IEEE}
}

@article{bhattacharya2024dynamic,
  title={Dynamic matching with better-than-2 approximation in polylogarithmic update time},
  author={Bhattacharya, Sayan and Kiss, Peter and Saranurak, Thatchaphol and Wajc, David},
  journal={Journal of the ACM},
  volume={71},
  number={5},
  pages={1--32},
  year={2024},
  publisher={ACM New York, NY}
}

@inproceedings{behnezhad2023dynamic,
  title={Dynamic algorithms for maximum matching size},
  author={Behnezhad, Soheil},
  booktitle={Proceedings of the 2023 Annual ACM-SIAM Symposium on Discrete Algorithms (SODA)},
  pages={129--162},
  year={2023},
  organization={SIAM}
}

@inproceedings{lange2022properly,
  title={Properly learning monotone functions via local correction},
  author={Lange, Jane and Rubinfeld, Ronitt and Vasilyan, Arsen},
  booktitle={2022 IEEE 63rd Annual Symposium on Foundations of Computer Science (FOCS)},
  pages={75--86},
  year={2022},
  organization={IEEE}
}

@inproceedings{lange2023agnostic,
  title={Agnostic proper learning of monotone functions: beyond the black-box correction barrier},
  author={Lange, Jane and Vasilyan, Arsen},
  booktitle={2023 IEEE 64th Annual Symposium on Foundations of Computer Science (FOCS)},
  pages={1149--1170},
  year={2023},
  organization={IEEE}
}

@article{lange2023local,
  title={Local Lipschitz filters for bounded-range functions},
  author={Lange, Jane and Linder, Ephraim and Raskhodnikova, Sofya and Vasilyan, Arsen},
  journal={arXiv preprint arXiv:2308.14716},
  year={2023}
}

@inproceedings{dorobisz2023local,
  title={Local Computation Algorithms for Hypergraph Coloring-Following Beck’s Approach},
  author={Dorobisz, Andrzej and Kozik, Jakub},
  booktitle={50th International Colloquium on Automata, Languages, and Programming (ICALP 2023)},
  year={2023},
  organization={Schloss Dagstuhl--Leibniz-Zentrum f{\"u}r Informatik}
}

@inproceedings{achlioptas2020simple,
  title={Simple local computation algorithms for the general Lov{\'a}sz Local Lemma},
  author={Achlioptas, Dimitris and Gouleakis, Themis and Iliopoulos, Fotis},
  booktitle={Proceedings of the 32nd ACM Symposium on Parallelism in Algorithms and Architectures},
  pages={1--10},
  year={2020}
}

@inproceedings{brandt2021randomized,
  title={The randomized local computation complexity of the Lov{\'a}sz local lemma},
  author={Brandt, Sebastian and Grunau, Christoph and Rozho{\v{n}}, V{\'a}clav},
  booktitle={Proceedings of the 2021 ACM Symposium on Principles of Distributed Computing},
  pages={307--317},
  year={2021}
}

@inproceedings{behnezhad2023local,
  title={Local Computation Algorithms for Maximum Matching: New Lower Bounds},
  author={Behnezhad, Soheil and Roghani, Mohammad and Rubinstein, Aviad},
  booktitle={2023 IEEE 64th Annual Symposium on Foundations of Computer Science (FOCS)},
  pages={2322--2335},
  year={2023},
  organization={IEEE}
}

@inproceedings{mitrovic2024locally,
  title={Locally Computing Edge Orientations},
  author={Mitrovi{\'c}, Slobodan and Rubinfeld, Ronitt and Singhal, Mihir},
  booktitle={32nd Annual European Symposium on Algorithms (ESA 2024)},
  year={2024},
  organization={Schloss Dagstuhl--Leibniz-Zentrum f{\"u}r Informatik}
}

@inproceedings{BF90,
    author = {D. Beaver J. Feigenbaum},
    title = {Hiding instances in multi-oracle queries },
    booktitle = {Proc.\ 7th Annual STACS conference},
    year = {1990},
    pages = {34-48},
}

@inproceedings{L89,
    author = {R. Lipton},
    title = {New directions in testing},
    booktitle = {Proc.\ DIMACS Workshop on Distributed Computing and Cryptography},
    year = {1989},
}

@inproceedings{GLR+91,
    author = {P. Gemmell and R. Lipton and R. Rubinfeld and M. Sudan and A. Wigderson  },
    title = {Self-testing/correcting for polynomials and for approximate functions},
    booktitle = proc23 # stoc,
    year = {1991},
	pages = {32--42}
}

@article{FF93,
    author = {J. Feigenbaum and L. Fortnow},
    title = {Random self-reducibility of complete sets},
    journal = sicomp,
    year = 1993,
    volume = 22,
    issue = {5},
    pages = {994-- 1005},
}

@inproceedings{KT00,
	author = {J. Katz and L. Trevisan},
	title = {On the efficiency of local decoding procedures for error-correcting codes},
	booktitle ={STOC},
	year = {2000},
	pages = {80--86}
}

@article{Yek10,
author={Sergey Yekhanin},
title={Private information retrieval},
journal={Commun. ACM},
volume={53},
number={4},
pages={68-73},
year={2010}
}

@article{GRS00,
author={Oded Goldreich and Ronitt Rubinfeld and Madhu Sudan},
title={Learning polynomials with queries: the highly noisy case},
journal={SIAM Journal on Discrete Mathematics},
volume={13},
number={4},
pages={535-570},
year={2000}
}

@article{AS03,
author={Sanjeev Arora and Madhu Sudan},
title={Improved low-degree testing and its applications},
journal={Combinatorica}, 
volume={23},
number={3},
pages={365-426}, 
year={2003}
}

@article{STV01,
author={Madhu Sudan and Luca Trevisan and Salil Vadhan},
title={Pseudorandom generators without the {XOR} lemma},
journal={Journal of Computer and System Sciences}, 
volume={62},
number={2},
pages={236-266}, 
year={2001}
}

@inproceedings{GKZ08,
author={Parikshit Gopalan and Adam R. Klivans and David Zuckerman},
title={List-decoding {R}eed {M}uller codes over small fields},
booktitle={Proc.\ 40th Annual ACM Symposium on Theory of Computing},pages={265-274},
year={2008}
}

@inproceedings{IW97,
author={Russell Impagliazzo and Avi Wigderson},
title={{\em P = B{PP}} if {{\em E}} requires exponential circuits: Derandomizing the {XOR} lemma},
booktitle={Proc.\ 29th Annual ACM Symposium on the Theory of Computing},
pages={220-229}, 
year={1997}
}

@misc{KS09,
author={Swastik Kopparty and Shubhangi Saraf},
title={Local list-decoding and testing of sparse random linear codes from high-error},
note={Technical Report 115, Electronic Colloquium on Computational Complexity (ECCC), 2009}
}

@misc{BET10,
author={Avraham Ben-Aroya and Klim Efremenko and Amnon Ta-Shma},
title={Local List-Decoding with a Constant Number of Queries},
note={Technical Report TR10-047, Electronic Colloquium on Computational Complexity, April 2010}
}

@techreport{MuthuSZ,
	author = { S. Muthukrishnan and M. Strauss and X. Zheng},
	title = {Workload-Optimal Histograms on Streams},
	institution = { DIMACS Technical Report},
	number = {2005-19},
	year = {2005}
}

@inproceedings{SG,
author = {K. Sadakane and R. Grossi},
title = {Squeezing Succinct Data Structures into Entropy Bounds},
booktitle = soda,
year = 2006,
pages = {1230--1239}
}

@inproceedings{FV07,
  author    = {P. Ferragina and
               R. Venturini},
  title     = {A simple storage scheme for strings achieving entropy bounds},
  booktitle = soda,
  year      = {2007},
  pages     = {690-696},
  ee        = {http://doi.acm.org/10.1145/1283383.1283457},
  bibsource = {DBLP, http://dblp.uni-trier.de}
}

@inproceedings{GN,
author = {R. Gonz\'{a}lez and G. Navarro},
title ={Statistical encoding of succint data structures},
booktitle = {Proceedings of CPM},
pages = {295--306},
year = 2006
}

@article{alon2006algorithmic,
  title={Algorithmic construction of sets for k-restrictions},
  author={Alon, Noga and Moshkovitz, Dana and Safra, Shmuel},
  journal={ACM Transactions on Algorithms (TALG)},
  volume={2},
  number={2},
  pages={153--177},
  year={2006},
  publisher={ACM New York, NY, USA}
}

@inproceedings{dinur2014analytical,
  title={Analytical approach to parallel repetition},
  author={Dinur, Irit and Steurer, David},
  booktitle={Proceedings of the forty-sixth annual ACM symposium on Theory of computing},
  pages={624--633},
  year={2014}
}

@article{feige1998threshold,
  title={A threshold of ln n for approximating set cover},
  author={Feige, Uriel},
  journal={Journal of the ACM (JACM)},
  volume={45},
  number={4},
  pages={634--652},
  year={1998},
  publisher={ACM New York, NY, USA}
}

@inproceedings{moshkovitz2012projection,
  title={The projection games conjecture and the NP-hardness of ln n-approximating set-cover},
  author={Moshkovitz, Dana},
  booktitle={International Workshop on Approximation Algorithms for Combinatorial Optimization},
  pages={276--287},
  year={2012},
  organization={Springer}
}

@inproceedings{raz1997sub,
  title={A sub-constant error-probability low-degree test, and a sub-constant error-probability PCP characterization of NP},
  author={Raz, Ran and Safra, Shmuel},
  booktitle={Proceedings of the twenty-ninth annual ACM symposium on Theory of computing},
  pages={475--484},
  year={1997}
}

@inproceedings{trevisan2001non,
  title={Non-approximability results for optimization problems on bounded degree instances},
  author={Trevisan, Luca},
  booktitle={Proceedings of the thirty-third annual ACM symposium on Theory of computing},
  pages={453--461},
  year={2001}
}

@article{biswas2024average,
  title={Average-Case Local Computation Algorithms},
  author={Biswas, Amartya Shankha and Cao, Ruidi and Pyne, Edward and Rubinfeld, Ronitt},
  journal={arXiv preprint arXiv:2403.00129},
  year={2024}
}

@article{ACC+08,
  author    = {N. Ailon and B. Chazelle and S. Comandur and D. Liu},
  title     = {Property-Preserving Data Reconstruction},
  journal   = {Algorithmica},
  volume    = {51},
   number    ={2},
  year      = {2008},
  pages     = {160--182}
}

@article{SS10,
	author    = {M. E. Saks and C. Seshadhri},
	title     = {Local monotonicity reconstruction},
	journal   = {SIAM Journal on Computing},
	volume    = {39(7)},
	year      = {2010},
	pages     = {2897--2926}
}

@article{BhattacharyyaGJJRW12,
  author    = {Arnab Bhattacharyya and
               Elena Grigorescu and
               Madhav Jha and
               Kyomin Jung and
               Sofya Raskhodnikova and
               David P. Woodruff},
  title     = {Lower Bounds for Local Monotonicity Reconstruction from
               Transitive-Closure Spanners},
  journal   = {SIAM J. Discrete Math.},
  volume    = {26},
  number    = {2},
  year      = {2012},
  pages     = {618-646},
}

@inproceedings{AwasthiJMR12,
  author    = {Pranjal Awasthi and
               Madhav Jha and
               Marco Molinaro and
               Sofya Raskhodnikova},
  title     = {Limitations of Local Filters of Lipschitz and Monotone Functions},
  booktitle = {APPROX-RANDOM},
  year      = {2012},
  pages     = {374-386},
}

@article{jha2013testing,
  title={Testing and reconstruction of Lipschitz functions with applications to data privacy},
  author={Jha, Madhav and Raskhodnikova, Sofya},
  journal={SIAM Journal on Computing},
  volume={42},
  number={2},
  pages={700--731},
  year={2013},
  publisher={SIAM}
}

@article{reingold2016new,
  title={New techniques and tighter bounds for local computation algorithms},
  author={Reingold, Omer and Vardi, Shai},
  journal={Journal of Computer and System Sciences},
  volume={82},
  number={7},
  pages={1180--1200},
  year={2016},
  publisher={Elsevier}
}

@article{even2018best,
  title={Best of two local models: Centralized local and distributed local algorithms},
  author={Even, Guy and Medina, Moti and Ron, Dana},
  journal={Information and Computation},
  volume={262},
  pages={69--89},
  year={2018},
  publisher={Elsevier}
}

@article{levi2017centralized,
  title={A (centralized) local guide},
  author={Levi, Reut and Medina, Moti and others},
  journal={Bulletin of EATCS},
  volume={2},
  number={122},
  year={2017}
}

@article{kuhn2016local,
  title={Local computation: Lower and upper bounds},
  author={Kuhn, Fabian and Moscibroda, Thomas and Wattenhofer, Roger},
  journal={Journal of the ACM (JACM)},
  volume={63},
  number={2},
  pages={1--44},
  year={2016},
  publisher={ACM New York, NY, USA}
}

@inproceedings{chang2023complexity,
  title={The complexity of distributed approximation of packing and covering integer linear programs},
  author={Chang, Yi-Jun and Li, Zeyong},
  booktitle={Proceedings of the 2023 ACM Symposium on Principles of Distributed Computing},
  pages={32--43},
  year={2023}
}

\newpage

\end{document}